\documentclass[
  reprint,
  floatfix,
  superscriptaddress,
  amsmath,amssymb,
  sort,
]{revtex4-2}

\usepackage{amsthm}
\usepackage{graphicx}
\usepackage{xcolor}
\usepackage{bm}
\usepackage{dcolumn}
\usepackage{booktabs}
\usepackage{float}
\usepackage{algorithm}
\usepackage{algpseudocode}
\usepackage{tikz}
\usepackage{newtxtext,newtxmath}
\usepackage{braket}
\usepackage[colorlinks=true,linkcolor=blue,citecolor=blue,urlcolor=blue]{hyperref}
\usepackage{rerunfilecheck}
\RerunFileCheck{supplementary.aux}{}{Rerun to get cross-references right}
\usepackage{makecell}
\usepackage{array}

\newtheorem{theorem}{Theorem}
\newtheorem{corollary}{Corollary}
\newtheorem{lemma}{Lemma}

\usepackage{chapterbib}

\makeatletter
\let\bibliographystyle\bibliographystyle@latex
\let\write@bibliographystyle\relax
\let\auto@bib\@empty
\let\thebibliography\NAT@thebibliography
\let\endthebibliography\endNAT@thebibliography
\let\NAT@bibsetnum\NATx@bibsetnum
\newcommand{\bibliographycontrols}{%
  \immediate\write\@auxout{\string\citation{REVTEX42Control,apsrev42Control}}%
}
\makeatother

\begin{document}
\title{\texorpdfstring{Harnessing problem structure for end-to-end quantum speed-ups
}{Harnessing problem structure for end-to-end quantum speed-ups
}}

\author{Qifan Jiang}
\thanks{These authors contributed equally to this work.}
\affiliation{College of Computer Science, and ZJU-Ningbo Global Innovation Center, Zhejiang University, Hangzhou, China}
\author{Xiao-Ming Zhang}
\thanks{These authors contributed equally to this work.}
\affiliation{School of Physics, South China Normal University, Guangzhou 510006, China}
\author{Debin Xiang}
\thanks{These authors contributed equally to this work.}
\affiliation{College of Computer Science, and ZJU-Ningbo Global Innovation Center, Zhejiang University, Hangzhou, China}
\author{Xiao Yuan}
\email{xiaoyuan@pku.edu.cn}
\affiliation{Center on Frontiers of Computing Studies, School of Computer Science, Peking University, Beijing 100871, China}
\author{Liqiang Lu}
\email{liqianglu@zju.edu.cn}
\affiliation{College of Computer Science, and ZJU-Ningbo Global Innovation Center, Zhejiang University, Hangzhou, China}
\author{Jianwei Yin}
\email{zjuyjw@cs.zju.edu.cn}
\affiliation{College of Computer Science, and ZJU-Ningbo Global Innovation Center, Zhejiang University, Hangzhou, China}
\begin{abstract}

Quantum algorithms can offer substantial computational speed-ups, yet these advantages may disappear once the cost of structure-agnostic classical data encoding is taken into account. Real-world problem instances, however, often possess rich internal structure. This raises a fundamental question: can such structure be harnessed to make quantum speed-ups survive end-to-end cost accounting? Here we show that it can. We introduce a general framework for structure-aware quantum data encoding that compiles compact recursive descriptions of problem structure into efficient state-preparation circuits, translating structural information directly into reduced encoding complexity. For the uncapacitated facility-location problem, structure-agnostic encoding admits a classical dequantization that eliminates the quadratic quantum speed-up, whereas exploiting the underlying problem structure restores this advantage even when state-preparation costs are included. More broadly, the same framework prepares structured quantum states for combinatorial optimization that capture non-trivial relations among constraints, including group balance, conflicts and synergistic rewards, extending previously reported super-polynomial quantum advantages to broader classes of objectives. These results establish exploitable problem structure as a computational resource for quantum algorithms and provide a systematic route towards realizing quantum advantage in data-intensive problems.

\end{abstract}

\maketitle

Quantum computing offers the prospect of substantial speed-ups for computationally challenging problems~\cite{shor1994algorithms,grover1996fast,jordan2025decoded}. Yet an algorithmic speed-up does not by itself guarantee an end-to-end quantum advantage. When problem instances are accessed or encoded without exploiting their underlying structure, the associated overhead can erase the gains of the subsequent quantum computation. Grover search, for example, achieves a quadratic reduction in query complexity for unstructured search~\cite{grover1996fast}, but this advantage can disappear once the full cost of realizing the problem oracle is taken into account~\cite{stoudenmire2024opening}. Similar concerns arise in quantum machine learning, where loading classical data into quantum states can offset the computational gains promised by quantum processing~\cite{biamonte2017quantum}. These limitations point to a central question: how can the intrinsic structure of real-world problems be exploited so that quantum speed-ups survive end-to-end complexity accounting?

Shor’s factoring algorithm provides a paradigmatic example, exploiting algebraic structure to achieve a superpolynomial speed-up over the best known classical algorithms~\cite{shor1994algorithms}. In real-world applications, however, structure is far more diverse: it may arise from constraints, correlations, symmetries, hierarchical organization or other relations among the degrees of freedom~\cite{bartschi2019deterministic,bartschi2022short,gleinig2021efficient,mozafari2022decision,araujo2024lowrank,malz2024mps,zhang2026bits}. 
Such structure can take different forms and affect the solution space in different ways. It remains unclear whether a common framework can exploit these structures to prepare quantum states efficiently and achieve computational speed-ups.

We address this question with a recursive framework for preparing quantum states over the feasible subspace induced by problem structure. Rather than enumerating exponentially many feasible solutions, the framework captures their structure in a compact recursive representation that is compiled directly into a state-preparation circuit. We quantify the complexity of the description by its \textit{recursive expansion dimension} $d$, which controls the size of the resulting circuit. When $d=O(1)$, the structural description remains compact throughout the recursion and yields an ancilla-free preparation circuit of polynomial depth. This condition is satisfied by several physically and algorithmically relevant structures, including fixed-Hamming-weight constraints~\cite{bartschi2019deterministic,bartschi2022short}, Rydberg-blockade chains~\cite{bernien2017probing,ebadi2022quantum}, and certain bounded-width divergence constraints~\cite{pardo2023resource,sharma2024gauss}. The same local-update principle can also be adapted to sparse states~\cite{gleinig2021efficient,li2024nearly,mao2024toward}, decision-diagram representations~\cite{mozafari2022decision}, and hierarchical low-rank states~\cite{araujo2024lowrank,malz2024mps}. Our framework therefore provides a universal route from compact structural descriptions to efficient quantum state preparation across otherwise distinct problem classes.

For hard-constrained problems, structure-aware encoding can restore quantum speed-ups that disappear when problem structure is ignored. We illustrate this with the uncapacitated facility-location problem, a canonical combinatorial-optimization problem~\cite{melo2009facility}. In a structure-agnostic formulation, the apparent quadratic speed-up offered by Grover search can be eliminated through tensor-network dequantization~\cite{stoudenmire2024opening}. By incorporating the underlying problem structure directly into the quantum encoding, our framework recovers this quadratic advantage at the end-to-end level. Under the strong exponential time hypothesis for unrestricted conjunctive normal form satisfiability (CNF-SETH)~\cite{hegerfeld2023tight}, we further prove a quadratic speed-up over deterministic exact classical algorithms. This example shows that how problem structure is represented can determine whether an algorithmic quantum speed-up survives as a genuine computational advantage.

The same connection extends from hard feasibility constraints to soft structural preferences. Standard decoded quantum interferometry (DQI) initializes a superposition of Dicke states, corresponding to objectives determined only by aggregate scores. Our framework enables more expressive initial states that encode relationships among constraints, including group balance, pairwise conflicts, precedence relations, and synergy rewards. As a representative example, we consider a balance preference between two groups and show that the previously reported DQI speed-up extends to this structure-aware objective under the same assumption used for the classical comparison. These results broaden the role of structured state preparation from restricting the feasible search space to encoding non-trivial relationships within the objective itself, providing a systematic route from problem structure to quantum computational advantage.

\begin{figure*}[t]
\centering
\includegraphics[width=0.72\textwidth]{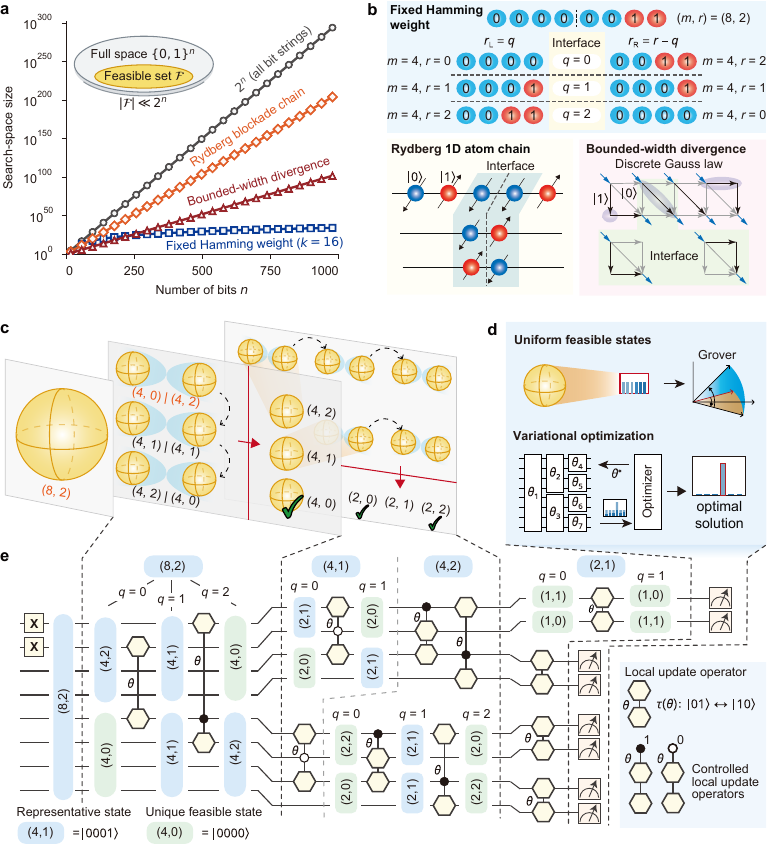}
\caption{\textbf{Recursive interfaces guide state preparation without enumeration.}
\textbf{a,} Exact support sizes for the full \(n\)-bit search space, fixed Hamming weight with \(k=16\) ones, a Rydberg blockade chain, and the bounded-width divergence ladder used as our benchmark. The structured feasible sets may be exponentially smaller than the full search space.
\textbf{b,} An interface stores the information needed to decide whether two child blocks can be joined. For fixed Hamming weight, \((m,r)\) labels a recursive block containing \(m\) bits and requiring exactly \(r\) ones. A split labeled \(q\) gives child weights \(q\) and \(r-q\). For a Rydberg blockade chain, compatibility depends on the occupations beside the split. For bounded-width divergence constraints, it depends on the net flows at boundary vertices.
\textbf{c,} Fixed-Hamming-weight recursion for \((n,k)=(8,2)\). Gold spheres represent the feasible sets of recursive blocks, and a vertical bar joins two child sets. Each block begins from one feasible bit string, called its representative. At the root, local updates transfer amplitude from the \(q=0\) choice to the \(q=1,2\) choices. They then repeat in child blocks containing more than one feasible string. Red arrows merge blocks with the same \((m,r)\) into reusable templates. Green checks mark blocks containing one feasible string.
\textbf{d,} Local updates prepare the uniform feasible state used in Grover search. Changing the local update angles produces weighted feasible states, including any prescribed feasible basis state.
\textbf{e,} Circuit realization of \textbf{c}. Blue and green labels mark nonsingleton and singleton block types. Dashed connectors map repeated block types to reusable modules. Each connected hexagon pair denotes a local update \(\tau(\theta)\), with \(\tau(\theta)|01\rangle=\cos(\theta)|01\rangle+\sin(\theta)|10\rangle\) when its controls are satisfied. Filled and open controls activate the update when the corresponding qubit is in \(\lvert1\rangle\) and \(\lvert0\rangle\), respectively. The angles determine the target probabilities over feasible strings.}
\label{fig:ncs-framework}
\end{figure*}

\vspace{.6cm}
\noindent\textbf{{State preparation in structured feasible subspaces}}

\noindent Many constrained problems have exponentially large search spaces, while their feasible solutions form a much smaller subset governed by compact structure. Harnessing this structure quantum mechanically requires preparing states over the feasible solutions without explicitly enumerating them, as such enumeration can itself be exponentially costly. We show that a compact recursive description of the constraints can instead be compiled directly into an efficient state-preparation circuit, such that the preparation cost is governed by the complexity of the underlying structure rather than by the number of feasible solutions.

Formally, a feasible set $\mathcal F$ contains the bit strings $x$ of length \(n\) that satisfy the constraints. Although $\mathcal F$ may be exponentially smaller than the full search space of the $2^n$ bit strings, it can still contain exponentially many elements and therefore be impractical to enumerate (Fig.~\ref{fig:ncs-framework}a). 
Our construction prepares states of the form
\begin{align}
|\psi_\pi\rangle=\sum_{x\in\mathcal F}\sqrt{\pi(x)}\,|x\rangle ,
\label{eq:sfs-target-state}
\end{align}
where $\pi(x)$ defines a probability distribution over the feasible strings $x\in\mathcal F$. By construction, all infeasible strings have zero amplitude.
Our approach builds on recursive decomposition and local updates, as detailed in Supplementary Sections I and II. Below, we illustrate the underlying idea using fixed-Hamming-weight states, while the same framework applies more broadly to general structured problems.

We first divide the variables into two groups of approximately equal size, called blocks, and describe the feasible strings in each block by their constraints. For fixed Hamming weight, \((n,k)\) denotes an \(n\)-bit block containing exactly \(k\) ones. For example, \((8,2)\) admits three splits into two four-bit blocks,
\[
(8,2)\longrightarrow
\begin{cases}
(4,0),\ (4,2),\\
(4,1),\ (4,1),\\
(4,2),\ (4,0).
\end{cases}
\]
The interface records the number of ones in the left block (equivalently the right block). 
Any strings satisfying these respective counts can be joined (Fig.~\ref{fig:ncs-framework}b).

To construct the circuit, we choose one feasible string as the representative of each constrained block according to a fixed rule. We place all ones at the right end. Local updates transfer amplitude between the representative basis states for the allowed splits. The circuit starts from the root representative \(|00000011\rangle=|0000\rangle|0011\rangle\), corresponding to the first split. Two local rotations introduce \(|0001\rangle|0001\rangle\) and \(|0011\rangle|0000\rangle\), corresponding to the second and third splits, respectively (Fig.~\ref{fig:ncs-framework}c). The same expansion continues within each child block until its constraints allow only one string. Choosing the probabilities at each split in proportion to the corresponding feasible counts gives a uniform superposition over all feasible strings. Blocks with the same length and required weight reuse the same circuit template.

The preparation cost is therefore governed by the recursive block structure and how many local updates are applied, rather than by the number of feasible strings. The recursive expansion dimension $d$ describes how the number of required local expansions scales with block size. A balanced recursion with an efficiently constructible update schedule then gives an ancilla-free preparation depth $O(n^{d+1})$, where $d=O(1)$ suffices for polynomial depth. Formal definitions and the resource theorem are given in Methods.

The same construction also applies beyond fixed Hamming weight. For a Rydberg blockade chain, where neighbouring sites cannot both be occupied~\cite{bernien2017probing,ebadi2022quantum}, the interface is the pair of occupations adjacent to the split. For bounded-width divergence constraints, the discrete Gauss law reduces the interface to the net flows at the split boundary~\cite{pardo2023resource,sharma2024gauss}. Both families have $d=O(1)$ (Fig.~\ref{fig:ncs-framework}b).

The local updates are ancilla-free, and updates on disjoint blocks at the same level of the tree can run in parallel (Fig.~\ref{fig:ncs-framework}e). Their angles determine how probability is divided among the allowed interface choices. Choosing these angles from the feasible counts gives the uniform superposition over $\mathcal F$, the initial state required by Grover-style search. Choosing a single branch with unit probability prepares any prescribed feasible basis state, while treating the angles as free parameters gives a variational ansatz that remains within the feasible subspace (Fig.~\ref{fig:ncs-framework}d). The same compiled structure can therefore support constrained search, structured sampling, and constrained variational optimization.

\begin{table*}[t]
\caption{\textbf{Resource scalings for state preparation in six structured families.}}
\label{tab:ncs-scope}
\centering
\begin{minipage}{\textwidth}
\centering
\footnotesize
\setlength{\tabcolsep}{3.5pt}
\renewcommand{\arraystretch}{1.22}
\begin{tabular*}{\textwidth}{@{\extracolsep{\fill}}lllll@{}}
\toprule
Family &
\makecell[l]{Constraint or state description} &
\makecell[l]{Recursive information} &
Depth &
Gates \\
\midrule
Fixed Hamming weight &
\(\sum_{i=1}^{n}x_i=k\) &
left weight \(q\) &
\(O(k\log(n/k))\) &
\(O(kn)\) \\
Rydberg blockade chain &
\(x_i+x_{i+1}\le1,\ i=1,\ldots,n-1\) &
boundary occupations &
\(O(\log n)\) &
\(O(n)\) \\
Bounded-width divergence &
\(\sum_{e\in\delta^+(v)}x_e-\sum_{e\in\delta^-(v)}x_e=b_v,\ v\in V\) &
    boundary flow vector &
\(O(\log n)\) &
\(O(n)\) \\
\addlinespace[2pt]
Sparse support &
\(\Omega\subseteq\{0,1\}^{n},\ |\Omega|=S\) &
support subset &
$O(Sn)$ &
$O(Sn)$ \\
Decision diagram &
\makecell[l]{\(n\)-bit strings with values divisible by \(3\)} &
prefix remainder modulo \(3\) &
\(O(n)\) &
\(O(n)\) \\
Hierarchical low-rank &
\makecell[l]{balanced recursive Schmidt tree, \(r=O(1)\)} &
Schmidt sector &
$O(n)$ &
$O(n\log n)$ \\
\bottomrule
\end{tabular*}
\par\vspace{2pt}
\raggedright\footnotesize
The listed depth and gate-count scalings are derived from explicit ancilla-free preparation circuits. The first three rows describe constraint-defined feasible sets covered by the general construction. The last three rows use specialized constructions based on the same local-update principle. For fixed Hamming weight, we take \(1\le k\le n/2\) without loss of generality by bit complementation. The feasible sets at \(k=0\) and \(k=n\) each contain a single bit string and require no recursive preparation. For bounded-width divergence, \(V\) is the vertex set, \(\delta^+(v)\) and \(\delta^-(v)\) are the outgoing and incoming edge sets at \(v\), and \(b_v\) is the prescribed divergence. The divergence bounds apply to constructible bounded-degree families with bounded-width recursive decompositions, constant-size local trades and controls that isolate the required updates. For sparse support, polynomial depth requires \(S=\operatorname{poly}(n)\). For hierarchical low-rank states, \(r\) denotes the bond rank, and the Schmidt-sector states on each child block are assumed to have pairwise disjoint computational-basis supports.
\par
\end{minipage}
\end{table*}

\begin{figure*}[t]
\centering
\includegraphics[width=\textwidth]{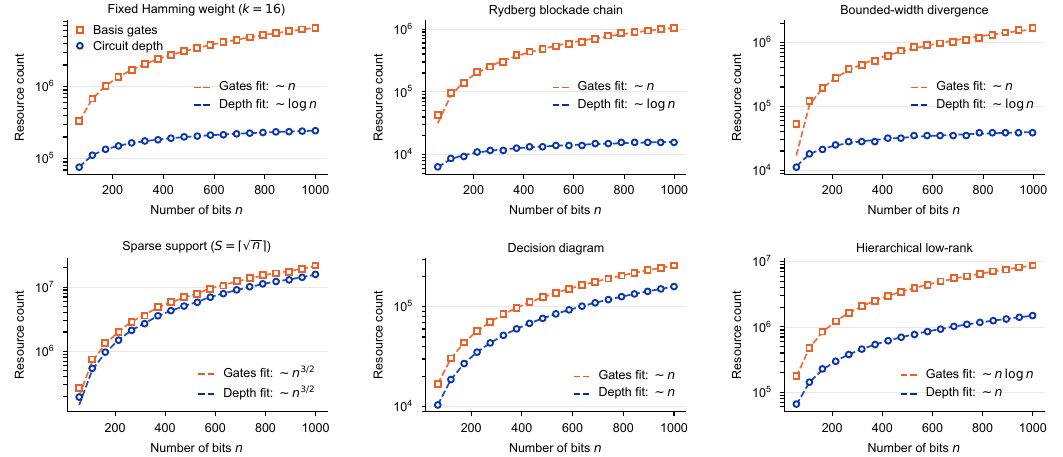}
\caption{\textbf{Circuit benchmarks are consistent with the analytical resource scalings.}
Orange squares show gate counts and blue circles show circuit depths after decomposition into the \(\{R_z,SX,X,\mathrm{CX}\}\) basis, where \(SX\) is the square-root-of-\(X\) gate. The benchmarks use \(k=16\), the bounded-width divergence ladder, \(S=\lceil\sqrt n\rceil\), the width-three divisibility-by-three decision diagram and an \(r=2\) even-parity Schmidt tree. Dashed curves show least-squares fits using forms fixed by the analytical bounds. The legends display the leading fitted terms, and complete models and coefficients are provided in Source Data.}
\label{fig:ncs-scope-scaling}
\end{figure*}

\vspace{.6cm}
 \noindent\textbf{Examples and resource scaling}

\noindent
Our framework enables polynomial-depth state preparation across diverse forms of structured data. We consider six representative families, each specified by a compact description, with analytical depth and gate-count scalings summarized in Table~\ref{tab:ncs-scope}. Explicit circuit constructions for systems of up to $n\approx1000$ confirm these scalings (Fig.~\ref{fig:ncs-scope-scaling}). Despite their distinct constraint structures, all six families are governed by the same structural quantity: the expansion dimension $d$, which remains bounded in each case and determines the resulting circuit complexity. This common dependence allows otherwise disparate forms of structure to be treated within a single state-preparation framework.

All three constraint-defined families have \(d=O(1)\) and therefore polynomial preparation depth. The value of \(d\) reflects how many residual conditions and compatible splits must be handled within a block. For fixed Hamming weight with \(k=\Theta(n)\), \(\Theta(m)\) residual weights can occur at block size \(m\). Expanding them requires \(\Theta(m^2)\) local updates and gives \(d=2\). A specialized ordering of the local updates gives preparation depth \(O(k\log(n/k))\), matching the best previously known scaling among constructions that do not use idle data qubits as temporary workspace~\cite{bartschi2022short}. 
Rydberg blockade chains have \(d=0\) because the numbers of residual boundary conditions and compatible labels remain constant. The same is true for the bounded-width divergence families specified in Table~\ref{tab:ncs-scope}. In both families, each update acts on only a constant number of qubits, so balanced recursion gives logarithmic depth.

The same local-update principle also supports states specified directly rather than through constraint systems. Specialized versions apply to sparse supports, decision diagrams, and hierarchical low-rank decompositions. A support list identifies the bit strings with nonzero amplitude. A decision diagram merges prefixes with the same allowed continuations, while a hierarchical low-rank description keeps only a bounded number of terms across each split. All three state families considered here admit polynomial-depth preparation. Supplementary Section IV gives the constructions and resource bounds.

We next show how the state preparation algorithm enables quantum speed-ups for practical problems.

\begin{figure*}[!t]
\centering
\includegraphics[width=0.86\textwidth]{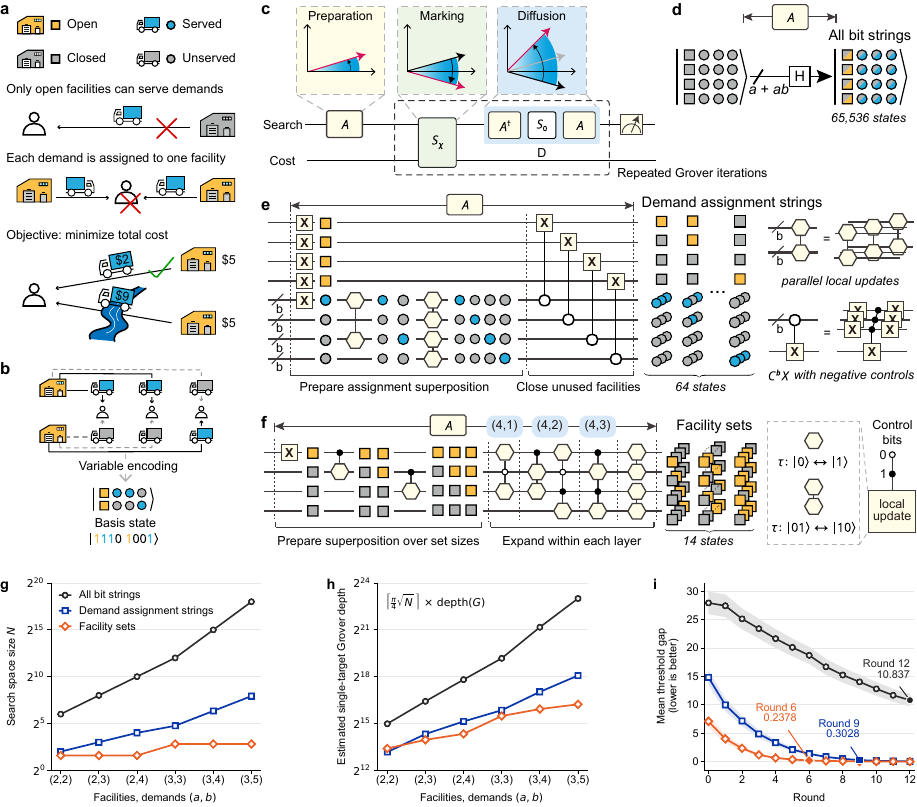}
\caption{\textbf{Facility-set search reduces the uncapacitated facility location problem (UFLP) search space.}
\textbf{a,} UFLP with \(a\) facilities and \(b\) demands. Each demand is assigned to one open facility, and the objective is to minimize total opening and service cost.
\textbf{b,} Binary encoding with one opening bit and \(b\) assignment bits per facility, giving \(a+ab\) bits. In the displayed basis state, each demand is assigned to one open facility and every open facility is used.
\textbf{c,} Grover-style search. \(A\) prepares the search state, \(S_\chi\) marks feasible solutions below the current cost threshold, and \(D=AS_0A^\dagger\) performs diffusion, with \(S_0\) the zero-state reflection. Each iteration applies \(S_\chi\) followed by \(D\), giving \(G=DS_\chi\).
\textbf{d,} Hadamard gates prepare a uniform superposition over the full search space of \(2^{a+ab}=65{,}536\) bit strings for the \((a,b)=(4,3)\) example. This space includes infeasible bit strings and redundant opening patterns.
\textbf{e,} Structured feasible subspace (SFS) preparation produces a uniform superposition over 64 demand assignment strings. Parallel local updates prepare assignments with all facilities open. A \(C^bX\) gate with negative controls closes a facility when all its assignment bits are zero.
\textbf{f,} SFS prepares a uniform superposition over the 14 facility sets with \(1\leq r\leq\min(a,b)=3\). Labels \((4,r)\) give the number of facilities and set size. The \(\tau\) gates are local updates. Filled and open controls activate an update when the corresponding qubit is in \(\lvert 1\rangle\) and \(\lvert 0\rangle\), respectively.
\textbf{g,} Exact sizes \(N\) of the three search spaces across six UFLP sizes.
\textbf{h,} Single-target Grover depth estimate \(\lceil(\pi/4)\sqrt N\rceil\operatorname{depth}(G)\), excluding the initial application of \(A\).
\textbf{i,} Noiseless Grover adaptive search. Curves show the mean difference between the current threshold and exact minimum across 50 instances, each averaged over 64 trajectories. Shading gives pointwise 95\% Student's \(t\) confidence intervals, truncated at zero. The horizontal axis shows the number of adaptive rounds. Panel \textbf{h} also averages 50 instances.}
\label{fig:ncs-uflp-scaling}
\end{figure*}

\vspace{.6cm}
\noindent\textbf{{Quantum speed-up for facility-set search}}\\
\noindent
Uncapacitated facility location (UFLP) is a standard benchmark for constrained optimization~\cite{melo2009facility}. A solution chooses which facilities are open and assigns every demand to one open facility while minimizing the total opening and service cost (Fig.~\ref{fig:ncs-uflp-scaling}a). For $a$ facilities and $b$ demands, the conventional binary encoding assigns one opening bit to each facility and one assignment bit to each facility and demand pair, producing a register of $a+ab$ bits (Fig.~\ref{fig:ncs-uflp-scaling}b). We additionally impose an opening limit $K$, requiring that at most $K$ facilities be open. Setting $K=a$ recovers standard UFLP, which is the setting used in Fig.~\ref{fig:ncs-uflp-scaling}.

Classically, with nonnegative costs, an optimum is determined entirely by which facilities are open, because each demand is then assigned to its cheapest open facility. Under unrestricted CNF-SETH~\cite{hegerfeld2023tight}, opening-limited UFLP requires \(a^{K-o(1)}\) deterministic exact time in the worst case for every fixed \(K\geq2\). For standard UFLP at \(K=a\), the same hypothesis rules out any deterministic exact algorithm running in \(O^*((2-\epsilon)^a)\) time for fixed \(0<\epsilon<1\), where \(O^*(\cdot)\) suppresses factors polynomial in the input size. These bounds set the classical targets for quantum search.

Grover-based quantum minimum finding repeatedly marks and amplifies basis states below a cost threshold within the chosen search space (Fig.~\ref{fig:ncs-uflp-scaling}c). Hadamard preparation spans the full search space of \(2^{a+ab}\) bit strings (Fig.~\ref{fig:ncs-uflp-scaling}d). Thus, Grover search needs approximately \(2^{(a+ab)/2}\) iterations, exceeding the classical baseline for every \(b\geq2\). 

Searching in the full space is far from optimal: many bit strings are infeasible because a demand is not assigned exactly once or is assigned to a closed facility, while others are redundant because they leave unused facilities open. One possible improvement is to restrict the search to demand assignment strings. Each string assigns every demand exactly once and closes every unused facility, so the space contains exactly \(a^b\) strings. Grover search over these strings needs on the order of \(a^{b/2}\) iterations. However, despite the exponential reduction, when \(b>2a/\log_2 a\), its \(O^*(a^{b/2})\) running time remains above the \(O^*(2^a)\) classical baseline for standard UFLP, whose leading exponent is tight under unrestricted CNF-SETH.

The structure has not yet been fully taken into account by the method above. 
To beat the classical baseline, we observe that for any \textit{non-empty} facility set \(Y\), an optimal assignment can be found efficiently by sending each demand to a least-cost facility in \(Y\).
Thus, instead of searching demand assignments, we search for an optimal non-empty facility set \(Y\). 
Accordingly, the objective value for \(Y\) is its opening cost plus the minimum assignment cost, and can be computed efficiently in a quantum register.
In this approach, the number of legal candidates is much smaller. Closing an unused facility will not increase the cost, so the target optimum solution has at most \(m=\min(K,b)\) open facilities. Accordingly, there are \(N_{a,b,K}=\sum_{r=1}^{m}\binom{a}{r}\) candidates in total. Quantum minimum finding therefore requires \(O(\sqrt{N_{a,b,K}})\) circuit applications~\cite{durr1996quantum}. 
For standard UFLP at \(K=a\), the running time is \(O^*(2^{a/2})\), giving an asymptotic quadratic speed-up over the \(O^*(2^a)\) classical scaling. We state in Theorem~\ref{thm:uflp-search-advantage} the standard-UFLP result, and give the proofs and further comparisons in Supplementary Section V.

Our SFS construction prepares the uniform quantum states required by both search methods.
The circuit for demand assignment strings first prepares a superposition over assignments with all facilities open, then closes every facility that serves no demand, so the resulting basis states assign each demand once and leave no unused facility open (Fig.~\ref{fig:ncs-uflp-scaling}e). The circuit for facility sets instead superposes the allowed set sizes \(1\leq r\leq\min(K,b)\) and expands each fixed-size layer (Fig.~\ref{fig:ncs-uflp-scaling}f). Both state preparation methods add only polynomial overhead.

Small-instance calculations show how the smaller search spaces change the required resources. The space of facility sets is the smallest of the three at every tested size (Fig.~\ref{fig:ncs-uflp-scaling}g). Its Grover iterate is deeper than the iterate for demand assignment strings at \((a,b)=(2,2)\), but the smaller search space gives a lower single-target depth estimate from \((2,3)\) onward (Fig.~\ref{fig:ncs-uflp-scaling}h).

We also compare the three search spaces using noiseless simulations of Grover adaptive search~\cite{gilliam2021grover} numerically. Under the initialization and round-based protocol defined in Methods, search over facility sets has the smallest mean threshold gap at the reported rounds across 50 independently generated instances. Its mean gap reaches \(0.2378\) by round 6. Search over demand assignment strings reaches \(0.3028\) by round 9, while search over all bit strings remains at \(10.837\) at round 12 (Fig.~\ref{fig:ncs-uflp-scaling}i). The curves therefore show faster convergence across adaptive rounds for the smaller search spaces.

\begin{figure*}[!t]
\centering
\includegraphics[width=0.90\textwidth]{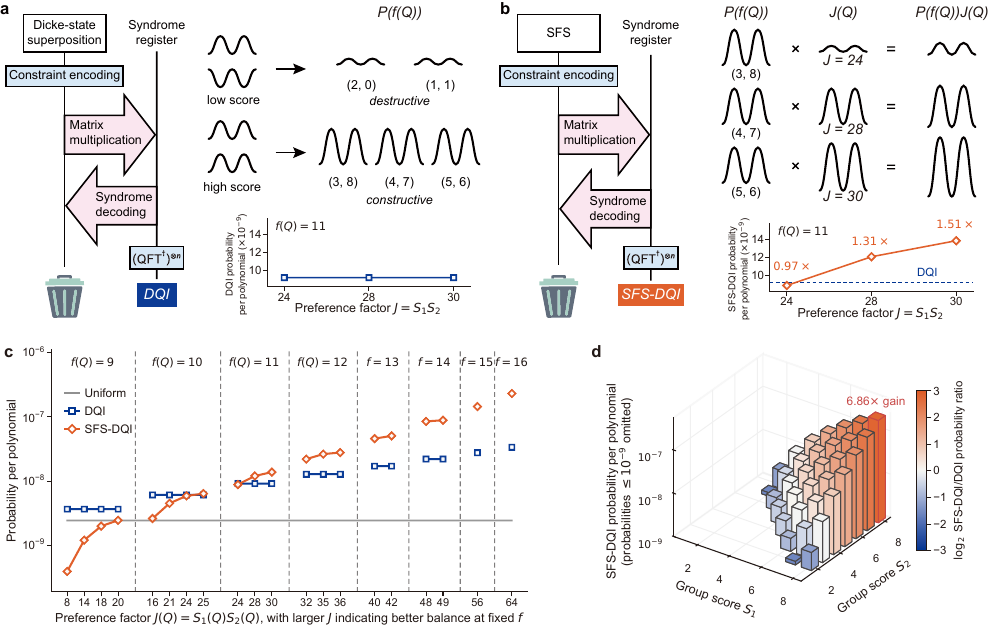}
\caption{\textbf{Structured initial states favor greater group balance among equal-score polynomials.}
\textbf{a,} Optimal polynomial intersection (OPI) scores a polynomial \(Q\) by the number \(f(Q)\) of evaluation constraints it satisfies. The group scores \(S_1(Q)\) and \(S_2(Q)\) define \(J(Q)=S_1(Q)S_2(Q)\), which measures balance at fixed total score. Standard decoded quantum interferometry (DQI) starts from a weighted superposition of Dicke states. Each Dicke state is uniform over constraint subsets of a fixed size, and the weights define the score-weighting function \(P\). These weights are chosen to raise the expected score. The circuit produces amplitudes proportional to \(P(f(Q))\) and therefore cannot distinguish equal-score polynomials. The box \((\mathrm{QFT}^{\dagger})^{\otimes n}\) denotes the inverse \(q\)-dimensional quantum Fourier transform applied independently to the \(n\) syndrome qudits, where \(q\) is the field size and \(n\) is the number of polynomial coefficients. The trash icon marks the error register cleared after decoding. The lower plot shows the DQI probability per polynomial at \(f(Q)=11\) for \(J(Q)=24,28,30\).
\textbf{b,} Our extension combines structured feasible subspace preparation with DQI (SFS-DQI). It replaces the initial state with one that records how many selected constraints come from each group. The same \(P\) is used. The remaining DQI circuit, including the inverse Fourier transforms and decoding, is unchanged. The output amplitudes are proportional to \(P(f(Q))J(Q)\). The lower plot shows the corresponding SFS-DQI probabilities. Orange labels give their ratios to the DQI probabilities in \textbf{a}.
\textbf{c,} Probabilities per polynomial within fixed score layers. DQI is constant within each layer, whereas SFS-DQI favors larger \(J(Q)\). The gray line is uniform over all polynomials.
\textbf{d,} SFS-DQI probabilities across \((S_1,S_2)\) profiles. Bar height gives the mean SFS-DQI output probability per polynomial. Color shows the base-two logarithm of the mean SFS-DQI/DQI probability ratio across instances. Profiles with mean SFS-DQI probability at most \(10^{-9}\) are omitted. All means give equal weight to 50 independently generated instances.}
\label{fig:ncs-dqi-reweighting}
\end{figure*}

\vspace{.6cm}
\noindent\textbf{{Speed up combinatorial optimization}}\\
\noindent
In many combinatorial optimization problems, candidates are ranked by a scalar score, such as the number of satisfied clauses in MAX-SAT. This score does not capture how the satisfied constraints are distributed. For example, if the constraints are partitioned into two groups, one may prefer solutions that satisfy both groups evenly over solutions concentrated in only one group with the same total score. We use group balance to demonstrate how such structural preferences can be incorporated into decoded quantum interferometry (DQI)~\cite{jordan2025decoded}. Other preferences, including conflicts, precedence relations and synergy rewards, are discussed in Supplementary Section VI E.

We demonstrate this approach for optimal polynomial intersection (OPI), in which each candidate is a low-degree polynomial \(Q\) with score \(f(Q)\), the number of evaluation constraints it satisfies. Standard DQI starts from a symmetric state that treats all subsets containing the same number of constraints alike. The weights of this state define a score-weighting function \(P\) and are chosen to raise the expected score (Fig.~\ref{fig:ncs-dqi-reweighting}a). Our structured initial state instead distinguishes the two groups. It adds a group-balance factor \(J(Q)\) to the output amplitude (Fig.~\ref{fig:ncs-dqi-reweighting}b). Changing only the initial state gives
\begin{equation}
P(f(Q))\longrightarrow P(f(Q))J(Q).
\label{eq:ncs-score-to-preference}
\end{equation}
\(P\) retains the chosen total-score weighting, while \(J\) distinguishes equal-score polynomials. We call this construction SFS-DQI. All later DQI operations remain unchanged.

We partition the constraints into two prescribed nonempty groups \(G_1\) and \(G_2\). Let \(S_g(Q)\) count the constraints in \(G_g\) satisfied by \(Q\). Then \(f(Q)=S_1(Q)+S_2(Q)\) and \(J(Q)=S_1(Q)S_2(Q)\). At fixed \(f(Q)\), \(J(Q)\) is larger when \(S_1(Q)\) and \(S_2(Q)\) are more balanced. We call the standard OPI score task with this added preference grouped OPI. Multivariate DQI can encode bounded-degree polynomials of group scores and has been optimized for linearly weighted sums~\cite{bu2026multivariate}. The initial-state circuit presented there prepares the groups independently, using ancilla registers to store their Hamming weights. Grouped OPI instead adds a nonlinear balance preference among equal-score solutions. Our ancilla-free SFS circuit directly prepares the required initial state, including correlations between the two group weights.

Because output probabilities are squared amplitudes, SFS-DQI reweights equal-score polynomials in proportion to \(J(Q)^2\) (Fig.~\ref{fig:ncs-dqi-reweighting}c). In the grouped OPI family studied here, every evaluation point has the same number of allowed values. Under the degree condition in Methods, we prove that this reweighting strictly increases the expected total score relative to score-only DQI using the same nonzero \(P\). Supplementary Section VI gives the exact distribution and proof.

To isolate structural reweighting, the numerical example uses the same \(P\) of degree one for both distributions. At \(f(Q)=11\), the mean SFS-DQI-to-DQI probability ratio increases from \(0.97\times\) at \(J(Q)=24\) to \(1.51\times\) at \(J(Q)=30\) (Fig.~\ref{fig:ncs-dqi-reweighting}a,b). Across the group-score profiles, the largest enhancement occurs at \((S_1,S_2)=(8,8)\), where the ratio reaches \(6.86\times\) (Fig.~\ref{fig:ncs-dqi-reweighting}d). With this matched \(P\), the structured initial state raises the expected total score from \(9.56\) to \(10.77\).

\vspace{.6cm}
\noindent\textbf{{Conclusion and outlook}}\\

\noindent Our framework compiles compact structural descriptions into quantum state-preparation circuits without enumerating feasible solutions. Hard constraints restrict the search space, while soft preferences reshape output probabilities. For standard UFLP, this yields a quadratic speed-up in the leading exponent over deterministic exact classical algorithms under unrestricted CNF-SETH. For grouped OPI, structured preparation favors group balance and increases the expected score for the same score-weighting function. The reported superpolynomial speed-up over known classical algorithms extends to this task under the standard DQI comparison assumption. Further work could automate recursive decomposition, reduce preparation costs through specialized compilation, and extend the construction to other problem families and variational quantum algorithms.

OpenAI recently reported a deterministic subexponential-time algorithm for satisfiability at every fixed clause width~\cite{openai2026eth}. If confirmed, this result would refute the exponential time hypothesis and its strong form~\cite{impagliazzo2001complexity}. However, polynomial-time reductions can increase instance size, so this result does \textit{not} directly yield a faster classical algorithm for UFLP. Our bounds rely on a weaker assumption, unrestricted CNF-SETH, which allows clause width to grow with input size. OpenAI's result for fixed clause widths does \textit{not} by itself refute this assumption. Under this assumption, we establish the classical lower bounds through reductions via Hitting Set and \(K\)-Set Cover. Yet, we do not exclude the possibility that the CNF-SETH assumption will be overturned as well. Even in that case, we believe that such advances could also help simplify or decompose our search problems and potentially lead to even faster quantum algorithms.

Beyond its immediate applications, our results highlight a broader principle: the structure of classical data can determine whether a quantum speed-up survives end to end. Grover search provides a quadratic advantage in the black-box setting~\cite{grover1996fast}, yet this advantage can disappear when the underlying oracle admits an efficient classical representation~\cite{stoudenmire2024opening}. Here we show that incorporating nontrivial feasible-space structure into the quantum encoding can restore this advantage. This perspective is consistent with a growing body of dequantization results. Quantum linear-system algorithms, for example, can achieve exponential speed-ups when classical data are sufficiently structured to permit efficient quantum access~\cite{harrow2009quantum}, whereas quantum-inspired classical algorithms substantially narrow this advantage for more general data-access models~\cite{chia2022sampling}. Together, these results suggest that quantum advantage should not be viewed as a property of an algorithm alone, but as emerging from the interplay between problem structure, data representation, and quantum computation. Developing systematic ways to identify, represent, and exploit such structure may therefore be essential for determining where meaningful quantum advantages can arise. More broadly, structure-aware quantum encoding provides a route towards designing quantum algorithms around the organization of real-world data, rather than treating data loading as an external cost or an idealized primitive.

\renewcommand{\bibsection}{\par\vspace{.6cm}\noindent\textbf{\large References}\par\nobreak\vspace{.3cm}}
\bibliographycontrols
\bibliographystyle{apsrev4-2}
\bibliography{refs}

\clearpage

\vspace{.6cm}
\noindent\textbf{\large{Methods}}\\
\noindent

\noindent\textbf{\large{Target states in feasible subspaces}}\\
\noindent
We first specify the states that the preparation circuit must produce. We model the hard constraints used in the main text as linear equalities and inequalities. Let $x\equiv x_1x_2\cdots x_n\in\{0,1\}^n$ denote an $n$-bit string. The feasible set is
\begin{equation}
\mathcal{F}\equiv \{x\in\{0,1\}^n \mid Ax=b,\ Gx\le h\}.\label{eq:fs}
\end{equation}
Here \(A\) and \(G\) are integer matrices with \(n\) columns, while \(b\) and \(h\) are integer vectors of the corresponding row dimensions. Each \(x\in\mathcal F\) labels a computational basis state \(|x\rangle\), and the corresponding feasible subspace is
\begin{align*}
\mathcal{H}_{\mathcal{F}} \equiv \mathrm{span}\{|x\rangle \mid x\in\mathcal{F}\}.
\end{align*}
Unlike a sparse support list, \(\mathcal F\) can be exponentially large. The goal is therefore to prepare quantum states supported entirely in \(\mathcal H_{\mathcal F}\) without enumerating the feasible bit strings.

For a target distribution \(\pi\) over \(\mathcal F\) defined by split probabilities along the recursion, the target state is \(|\psi_{\pi}\rangle \equiv \sum_{x\in\mathcal F}\sqrt{\pi(x)}\,|x\rangle\). This class includes the uniform distribution and any prescribed feasible basis state. The uniform choice \(\pi(x)=1/|\mathcal F|\) is a natural initialization for Grover-style search and related algorithms.

\vspace{.6cm}
\noindent\textbf{\large{Recursive interfaces between variable blocks}}\\
\noindent
A direct list of feasible strings can be exponentially long. We instead use a classical recursive description as the blueprint for state preparation. The feasible set \(\mathcal F\) in Eq.~\eqref{eq:fs} is the root block. A block is specified by \((I,\beta)\), where \(I\subseteq[n]\) denotes its variables and \(\beta\) is their residual constraint. Here \([n]:=\{1,\ldots,n\}\). We write \(F(I,\beta)\) for the set of bit strings on \(I\) satisfying \(\beta\). At the root, \(I=[n]\) and \(\beta\) is the global constraint system \(Ax=b\) and \(Gx\le h\). Whether a block is terminal may depend on \(\beta\), and every nonempty terminal feasible set contains exactly one bit string. Whenever a block is nonterminal, its prescribed nearly balanced split \(I=I_L\sqcup I_R\) depends only on \(I\). At each split, the interface records only the information needed to decide which two child blocks are compatible. An interface label \(\alpha\) selects a compatible pair of child residual constraints \((\beta_L^{(\alpha)},\beta_R^{(\alpha)})\), and the parent feasible set decomposes as
\begin{equation}
F(I,\beta)
=
\bigsqcup_{\alpha\in\mathcal A(I,\beta)}
\bigl(
F(I_L,\beta_L^{(\alpha)})\times F(I_R,\beta_R^{(\alpha)})
\bigr).
\label{eq:recursive-partition}
\end{equation}
Here \(\mathcal A(I,\beta)\) is the set of interface labels compatible with the block \((I,\beta)\) and its prescribed split. Recursively applying Eq.~\eqref{eq:recursive-partition} generates the decomposition tree.

We call the description constructible when a local classical rule processes one encoded block in time polynomial in its input and output lengths. It identifies empty and terminal blocks, returns a representative for each nonempty block and otherwise outputs the complete ordered list of nonempty interface labels and canonical child records. All records have polynomial-length encodings, and the first label reproduces the parent representative. Canonical records can be compared and deduplicated, and a feasible string can be assigned to its unique label at each split, in polynomial time.

\begin{lemma}[Efficient recursive decomposition]
\label{lem:structural-recursion}
A recursive decomposition satisfying \(\max\{|I_L|,|I_R|\}\le \rho |I|\) whenever a block is split for some constant \(\rho<1\), and Eq.~\eqref{eq:recursive-partition} as a disjoint cover of \(F(I,\beta)\), has height \(O(\log n)\). Each feasible string \(x\in\mathcal F\) induces a unique interface label at every nonterminal block.
\end{lemma}

\noindent The proof is given in Supplementary Section I.

\vspace{.6cm}
\noindent\textbf{\large{Compression through repeated block types}}\\
\noindent
Balanced recursion can produce the same constrained block many times. Building a separate expansion module for every occurrence would duplicate work. We therefore merge blocks that are equivalent up to qubit relabeling. They have the same size \(|I|\), the same form of residual constraint, and the same interface labels. Their child blocks also have the same structure. This compresses the recursion tree into a directed acyclic graph (DAG)~\cite{bryant1986graph,mozafari2022decision}.

We call the resulting DAG the compressed recursion graph and use it to organize the preparation circuit. Each constrained block \((I,\beta)\) is assigned a representative bit string
\begin{equation}
c_{I,\beta}
\in
F(I,\beta).
\label{eq:block-representative}
\end{equation}
Representatives are chosen consistently with the stored block isomorphisms and recursive splits. Each interface label in Eq.~\eqref{eq:recursive-partition} specifies an allowed pair of child blocks and the corresponding product representative \(c^{(\alpha)}=c_{I_L,\beta_L^{(\alpha)}}\sqcup c_{I_R,\beta_R^{(\alpha)}}\). The labels of every nonterminal block are ordered as \(\alpha_1,\ldots,\alpha_{|\mathcal A(I,\beta)|}\), with \(c_{I,\beta}=c^{(\alpha_1)}\).

An ordered expansion of \((I,\beta)\) uses \(|\mathcal A(I,\beta)|-1\) local updates to open the remaining product representatives from \(c_{I,\beta}\). Applying the same construction recursively to the child blocks yields the full feasible set. The disjoint cover in Lemma~\ref{lem:structural-recursion} ensures that each feasible string appears in exactly one branch of the recursion.

The same variable set \(I\) may occur with different residual constraints. We collect them as
\begin{equation}
\mathcal R(I)
=
\left\{\beta\ \middle|\
\begin{aligned}
&(I,\beta)\text{ occurs in the}\\[-1mm]
&\text{recursive decomposition}
\end{aligned}
\right\}.
\label{eq:residual-constraint-set}
\end{equation}
For this count, we set \(|\mathcal A(I,\beta)|=1\) when \((I,\beta)\) is terminal, since no local update is required. The block expansion count is
\begin{equation}
\Lambda(I)
=
\sum_{\beta\in\mathcal R(I)}
\bigl(|\mathcal A(I,\beta)|-1\bigr).
\label{eq:block-expansion-count}
\end{equation}
Throughout, \(d\) denotes the recursive expansion dimension. It is the smallest nonnegative integer for which \(\Lambda(I)=O(|I|^d)\) holds uniformly over all variable sets and instances in the problem family.

\vspace{.6cm}
\noindent\textbf{\large{Amplitude transfer between compatible blocks}}\\
\noindent
The compressed recursion graph now provides a direct preparation schedule. The circuit starts in the root representative basis state \(\lvert c_{[n],(Ax=b,\,Gx\le h)}\rangle\) defined by Eq.~\eqref{eq:block-representative}. It expands each block into the representative basis states associated with its interface labels and proceeds from the root to the leaves. At a given recursion layer, updates associated with disjoint variable sets \(I\) act on disjoint qubits and can run in parallel.

For a fixed block, the elementary step is a two-level local update between representative basis states. If \(c^{(\alpha)}\) and \(c^{(\alpha')}\) label the representative basis states associated with two interface labels \(\alpha\) and \(\alpha'\), we write
\begin{equation}
\tau(\theta)\,|c^{(\alpha)}\rangle
=
\cos(\theta)\,|c^{(\alpha)}\rangle
+
\sin(\theta)\,|c^{(\alpha')}\rangle.
  \label{eq:cos-sin}
\end{equation}
Eq.~\eqref{eq:cos-sin} adds one new representative basis state to those already present. Chaining the ordered representative pairs generates the full family \(\{|c^{(\alpha)}\rangle\}\) for that block. For representative strings \(x\) and \(x'\), the update acts on the bits where they differ and is controlled on their common bits. This isolates the pair and gives an ancilla-free update with \(O(|I|)\) depth and gate count. Supplementary Section I gives the circuit.

A target distribution \(\pi\) is defined by split probabilities, which determine the local update angles. For each block \((I,\beta)\), let \(\pi_{I,\beta}\) denote the induced probability distribution on \(F(I,\beta)\). Whenever a block is split, these distributions obey
\begin{equation}
\pi_{I,\beta}
=
\sum_{\alpha\in\mathcal A(I,\beta)}
p_{I,\beta}(\alpha)\,
\bigl(
\pi_{I_L,\beta_L^{(\alpha)}}\times \pi_{I_R,\beta_R^{(\alpha)}}
\bigr),
\label{eq:recursive-distribution}
\end{equation}
where \(p_{I,\beta}(\alpha)\ge 0\) and \(\sum_{\alpha\in\mathcal A(I,\beta)} p_{I,\beta}(\alpha)=1\). Here \(\pi_{I_L,\beta_L^{(\alpha)}}\times \pi_{I_R,\beta_R^{(\alpha)}}\) denotes the product distribution on \(F(I_L,\beta_L^{(\alpha)})\times F(I_R,\beta_R^{(\alpha)})\). Once the interface label \(\alpha\) is fixed, the two child blocks are distributed independently, and each continues recursively according to the same rule.

The split probabilities are assigned after the compressed recursion graph is constructed. Neither these probabilities nor the resulting angles enter block isomorphism.

For a block \((I,\beta)\), first omit any trailing suffix of zero probability labels. Write the retained ordered probabilities as \(p_{I,\beta}(\alpha_1),\ldots,p_{I,\beta}(\alpha_s)\) and let \(Q_j=\sum_{t=j}^s p_{I,\beta}(\alpha_t)\). The angle of update \(j\) then satisfies
\begin{equation}
\theta_j
=
\arcsin\sqrt{\frac{Q_{j+1}}{Q_j}},
\qquad
j=1,\ldots,s-1.
\label{eq:angle-mass}
\end{equation}
The theorem covers target distributions whose ordered updates on each variable set can be merged in polynomial time into one block-local schedule that preserves every chain. If a local basis state occurs in more than one residual block, every occurrence undergoes the same two-level action with the same angle at each step, or every occurrence is left unchanged. Once a branch reaches a terminal block, all later updates within that block must leave its basis state unchanged.

For the uniform feasible state, the split probabilities are evaluated bottom up from the same recursion. For a prescribed feasible basis state, its unique branch at each split receives probability one. These choices give the following result.

\begin{theorem}[Efficient layerwise state preparation]\label{thm:layerwise-preparation}
Given a balanced and constructible recursive decomposition of \(\mathcal F\) with recursive expansion dimension \(d\), let \(\pi\) be a target distribution of the form in Eq.~\eqref{eq:recursive-distribution} whose ordered updates admit this schedule. The target state \(|\psi_\pi\rangle=\sum_{x\in\mathcal F}\sqrt{\pi(x)}\,|x\rangle\) can be prepared exactly in the ideal gate model with arbitrary single-qubit rotations by an ancilla-free layerwise circuit with depth \(O(n^{d+1})\) and gate count \(O(n^{d+1}+n\log n)\).
\end{theorem}

Under the conditions of the theorem, when \(d=O(1)\), the recursive update schedule and exact split probabilities for the uniform feasible state can be generated in polynomial classical time without enumerating \(\mathcal F\). The required local update angles can be evaluated to \(b\) bits in time polynomial in \(n\) and \(b\).

Supplementary Sections I and II give the update construction, recursive data bound, split mass calculation, theorem proof and finite-precision analysis.

\vspace{.6cm}
\noindent\textbf{\large{Circuit benchmarks for structured families}}\\
\noindent
We construct explicit circuits to connect the analytical bounds to gate-level resource estimates for the six families in Table~\ref{tab:ncs-scope}. All families use the ancilla-free local-update cost convention of Theorem~\ref{thm:layerwise-preparation}. We report preparation depth and gate count in the \(\{R_z,SX,X,\mathrm{CX}\}\) basis. The circuits were constructed with Qiskit 2.4.1~\cite{qiskit2024}, then translated using the default passes at optimization level 0 without a coupling map. The analytical resource bounds fix the fitting forms in Fig.~\ref{fig:ncs-scope-scaling}, and least-squares regression estimates their coefficients. The coefficient of determination is \(R^2=1-\sum_i(y_i-\hat y_i)^2/\sum_i(y_i-\bar y)^2\), where \(y_i\) are the circuit resource values, \(\hat y_i\) their fitted values and \(\bar y\) their mean. All twelve series have \(R^2\geq0.98\).

For the constraint families in Figs.~\ref{fig:ncs-framework}a and \ref{fig:ncs-scope-scaling}, fixed Hamming weight uses \(k=16\), and the Rydberg blockade chain contains \(n\) sites in one dimension. The divergence benchmark is a two-row directed ladder with \(n=4L+1\) binary edge variables, defined in Supplementary Section III. The general scaling for bounded-width divergence constraints in Table~\ref{tab:ncs-scope} applies to bounded-degree graph families with bounded-width recursive decompositions, constant-size local path or cycle updates and constant-size controls that isolate each update.

For the structured-state benchmarks, the input is a compact description rather than a constraint predicate. The sparse-support benchmark uses a pseudorandom support of size \(S(n)=\lceil\sqrt n\rceil\). Each nonseed string is opened directly from the seed. The decision-diagram benchmark uses a width-three ordered binary decision diagram accepting binary values divisible by three. The hierarchical low-rank construction requires a supplied balanced Schmidt tree with bounded bond ranks and pairwise disjoint sector supports in the computational basis on each child block. Our benchmark prepares the uniform even-parity state, whose even--even and odd--odd child sectors give bond rank \(r=2\) and satisfy this support condition.

Supplementary Sections III and IV give the circuits for the constraint families and structured states, respectively. Problem sizes, family parameters and fitted resource trends are provided as source data.

\vspace{.6cm}
\noindent\textbf{\large{Search support and resources for facility location}}\\
\noindent
For \(a\) facilities and \(b\) demands, let \(x_{ij}\in\{0,1\}\) indicate whether demand \(i\) is assigned to facility \(j\), and \(y_j\in\{0,1\}\) whether facility \(j\) is open. The constraints are \(\sum_{j=1}^{a}x_{ij}=1\), \(x_{ij}\leq y_j\) and \(\sum_{j=1}^{a}y_j\leq K\), with \(1\leq K\leq a\). Setting \(K=a\) recovers standard UFLP. With nonnegative opening costs \(f_j\) and service costs \(d_{ij}\), the objective is \(\min \sum_{j=1}^{a}f_jy_j+\sum_{i=1}^{b}\sum_{j=1}^{a}d_{ij}x_{ij}\).

Closing an unused facility cannot increase this objective. Once a nonempty set of facilities is chosen, the optimal assignment sends each demand to its cheapest open facility. It is therefore enough to search the nonempty sets containing at most \(m=\min(K,b)\) facilities. Their number is \(N_{a,b,K}=\sum_{r=1}^{m}\binom{a}{r}\). Our SFS circuit prepares the exact uniform superposition over these choices. Quantum minimum finding then uses \(O(\sqrt{N_{a,b,K}})\) applications of the state-preparation, inverse-preparation and objective-oracle circuits~\cite{durr1996quantum}.

\begin{theorem}[Conditional search speed-up for standard UFLP]
\label{thm:uflp-search-advantage}
For standard UFLP with nonnegative costs represented using polynomially many bits, a bounded-error quantum algorithm outputs the exact optimum in \(O^*(2^{a/2})\) time~\cite{durr1996quantum}. Here, \(O^*(\cdot)\) suppresses factors polynomial in the input size. Under unrestricted CNF-SETH, no deterministic exact classical algorithm runs in \(O^*((2-\epsilon)^a)\) time for any fixed \(0<\epsilon<1\). Hence the quantum algorithm gives a conditional asymptotic quadratic speed-up in the leading exponential dependence on \(a\).
\end{theorem}

\noindent For the fixed-\(K\) hard families obtained from \(K\)-Set Cover~\cite{patrascu2010possibility}, the same construction uses \(O(a^{K/2})\) circuit applications. Each application requires \(O(ab\operatorname{polylog}(a+b))\) gates, giving total time \(O(a^{K/2+1}b\operatorname{polylog}(a+b))\). Under unrestricted CNF-SETH, comparing worst-case running times by total input length gives a conditional speed-up for every fixed \(K\geq3\). This speed-up becomes asymptotically quadratic as \(K\) increases (Supplementary Section V).

Our construction also prepares the demand assignment strings used in Fig.~\ref{fig:ncs-uflp-scaling}e. It first assigns each demand while keeping all facilities open, then closes every facility that serves no demand. By Theorem~\ref{thm:layerwise-preparation}, we obtain the following corollary.

\begin{corollary}[Demand-assignment state preparation]
\label{cor:uflp-state-preparation}
For a UFLP instance with \(a\) facilities and \(b\) demands, an ancilla-free circuit prepares the exact uniform superposition over the \(a^b\) strings that assign every demand exactly once and close every unused facility. It can also prepare any prescribed basis state in this support. The circuit has depth \(O(\log a+b)\) and gate count \(O(ab)\).
\end{corollary}
\noindent When \(K<\min(a,b)\), the marking oracle rejects strings that use more than \(K\) facilities. Supplementary Section V gives both preparation procedures, the reversible cost oracle and the two reduction proofs.

For a basis string \(z\), let \(C(z)\) be its UFLP cost. One Grover iterate is \(G=AS_0A^\dagger S_\chi\), where \(A\) prepares the uniform search state and \(S_0\) is the reflection about the zero state. The oracle \(S_\chi\) checks feasibility, computes \(C(z)\) reversibly, marks feasible strings below threshold \(B\), then uncomputes its work registers. Feasibility checks, including the opening limit, are omitted when enforced by the prepared support. Component depths in the \(\{R_z,SX,X,\mathrm{CX}\}\) basis are included in the source data.

The resource circuits used Qiskit 2.4.1 and default basis-translation passes at optimization level 0 without a coupling map.

Fig.~\ref{fig:ncs-uflp-scaling} compares all bit strings, demand assignment strings and facility sets, with sizes \(2^{a+ab}\), \(a^b\) and \(N_{a,b,K}\), respectively. Throughout the figure, \(K=a\), so no opening-limit check is needed.

Fig.~\ref{fig:ncs-uflp-scaling}g,h show results for \((a,b)=(2,2)\), \((2,3)\), \((2,4)\), \((3,3)\), \((3,4)\) and \((3,5)\). Panel \textbf{g} reports exact support sizes, and panel \textbf{h} reports the single-target depth estimate \(\lceil(\pi/4)\sqrt N\rceil\operatorname{depth}(G)\), excluding initial preparation. Here \(N\) is the search space size.

For each size, we generate 50 instances with opening and service costs sampled independently and uniformly from \(\{1,\ldots,20\}\). Panel \textbf{h} shows arithmetic means over the same instances for all three spaces. Each iterate uses threshold \(B=C_{\min}+2\), where \(C_{\min}\) is the exact optimum. The iteration multiplier is the rounded asymptotic single-target estimate, independent of the marked count at this threshold.

Fig.~\ref{fig:ncs-uflp-scaling}i applies Grover adaptive search (GAS) to the same 50 \((a,b)=(2,3)\) instances. Each trajectory starts with \(B_{\rm ub}=\sum_j f_j+\sum_i\max_j d_{ij}+1\), samples \(A|0\rangle\) once, and sets its threshold to the measured cost if feasible. Otherwise, it retains \(B_{\rm ub}\). Trajectories evolve independently using Qiskit Aer 0.17.2.

The real search-window size \(k_t\) starts at one. At round \(t\), the trajectory samples \(j_t\) uniformly from \(\{0,\ldots,\lceil k_t\rceil-1\}\), applies \(G^{j_t}\) to the prepared state and measures once. A feasible improvement lowers the threshold and resets \(k_{t+1}=1\). Otherwise, \(k_{t+1}=\min((8/7)k_t,\sqrt{|\Omega|})\), where \(\Omega\) is the search support.

For each instance and search space, the gap \(B_t-C_{\min}\) between the current threshold and exact optimum is averaged over 64 independent trajectories. The mean \(\bar g_t\) and sample standard deviation \(s_t\) are calculated across the 50 instance averages with equal weights. Shading shows pointwise 95\% Student's \(t\) confidence intervals \(\bar g_t\pm t_{0.975,49}s_t/\sqrt{50}\), truncated below at zero. Here \(t_{0.975,49}\) is the 97.5th percentile with 49 degrees of freedom. The trajectory simulations implement the threshold phase oracle by exact enumeration, with the same marking action as the reversible circuit costed in panel \textbf{h}. Panel \textbf{i} compares adaptive rounds, with independently evolving thresholds and no matching of circuit depths across search spaces.

\vspace{.6cm}
\noindent\textbf{\large{Encoding group balance in the DQI initial state}}\\
\noindent
We derive the group-resolved initial state and show how it changes the DQI output while leaving the constraint encoding and decoding unchanged. Let \(q\) be prime and let \(\mathbb F_q\) denote the prime field with \(q\) elements. We label its nonzero elements by \(i=1,\ldots,m\), where \(m=q-1\), and use them as evaluation points. The polynomial degree is below \(n\), with \(1\leq n<m\). An OPI instance specifies an allowed set \(F_i\subset\mathbb F_q\) for each evaluation point, with \(|F_i|=r\) for all \(i\) and \(0<r<q\). As in standard OPI DQI, each \(F_i\) is assumed to be explicit or accessible with polynomial cost. Each polynomial \(Q\) is represented by its \(n\) coefficients in \(\mathbb F_q\). Let \(f_i(Q)=1\) if \(Q(i)\in F_i\) and \(f_i(Q)=0\) otherwise, and define \(f(Q)=\sum_{i=1}^{m}f_i(Q)\).

Standard OPI DQI uses a score-weighting function \(P\), which is a degree-\(\ell\) polynomial. The initial-state weights are chosen to raise the expected score~\cite{jordan2025decoded}. It prepares a binary-valued support mask in the \(|0\rangle\) and \(|1\rangle\) levels of \(m\) \(q\)-dimensional registers in the state
\[
|\Psi_{\rm DQI}\rangle
=\sum_{k=0}^{\ell}w_k|D_{m,k}\rangle,
\qquad
\sum_{k=0}^{\ell}|w_k|^2=1,
\]
where \(|D_{m,k}\rangle\) is the uniform superposition over the \(m\)-bit strings of Hamming weight \(k\). At each occupied mask position, constraint encoding replaces \(|1\rangle\) with a superposition over nonzero finite-field errors. This gives error vectors \(y\in\mathbb F_q^m\) with at most \(\ell\) nonzero components. Syndrome compression, reversible decoding and inverse Fourier transformation then produce amplitudes proportional to \(P(f(Q))\).

For the grouped balance preference used in Fig.~\ref{fig:ncs-dqi-reweighting}, the constraints are partitioned into two prescribed groups \(G_1\) and \(G_2\) of comparable size. Define \(S_g(Q)=\sum_{i\in G_g}f_i(Q)\) for \(g=1,2\) and \(J(Q)=S_1(Q)S_2(Q)\). At fixed total score, \(J(Q)\) is larger when the satisfied constraints are distributed more evenly between the groups.
To obtain amplitudes proportional to \(P(f(Q))J(Q)\), expand both factors in the finite-field Fourier representation. With \(\omega_q=e^{2\pi\mathrm i/q}\), write \(P(f(Q))=\sum_{y\in\mathbb F_q^m}a_y\omega_q^{-\sum_{i=1}^{m}y_iQ(i)}\) and \(J(Q)=\sum_{u\in\mathbb F_q^m}b_u\omega_q^{-\sum_{i=1}^{m}u_iQ(i)}\). Nonzero coefficients require at most \(\ell\) nonzero components in \(y\) and two in \(u\), since \(J\) is quadratic in the indicators. Multiplication convolves these coefficients, giving the normalized error-register state
\begin{equation}
|W_{\rm SFS\text{-}DQI}\rangle
=\frac{1}{\sqrt C}
\sum_{y,u\in\mathbb F_q^m}a_yb_u|y+u\rangle.
\end{equation}
Here \(C\) normalizes the state and \(y+u\) denotes componentwise addition in \(\mathbb F_q^m\). SFS-DQI changes only the pre-encoding mask state that produces this error-register state. Let \(B\in\mathbb F_q^{m\times n}\) be the OPI Vandermonde evaluation matrix, with \(B_{i,j+1}=i^j\) for \(j=0,\ldots,n-1\). It maps the polynomial coefficients to their values at the \(m\) evaluation points. Provided that the Reed--Solomon syndrome map \(y\mapsto B^Ty\) is injective on the modified error support, the subsequent DQI operations remain unchanged and produce final amplitudes proportional to \(P(f(Q))J(Q)\).

The modified error weight is bounded by \(R=\ell+2\). We prepare the two-group weight profiles with \(O(R^2)\) local updates and phases controlled by the group weights, then refine both groups in parallel by fixed Hamming weight recursion. With \(m_g=|G_g|\) and \(R_g=\min(R,m_g)\), the preparation depth is \(O\!\left(R^2+\max_{g=1,2}[R_g\log(m_g/R_g)+R_g]\right)\), at most \(O(m^2)\).

We write \(\langle\cdot\rangle_{\rm DQI}\) and \(\langle\cdot\rangle_{\rm SFS\text{-}DQI}\) for expectations under the respective output distributions. After normalization, \(\langle f\rangle_{\rm SFS\text{-}DQI}=\langle fJ^2\rangle_{\rm DQI}/\langle J^2\rangle_{\rm DQI}\). Within each score layer of nonzero SFS-DQI probability, sampling probabilities are proportional to \(J(Q)^2\). For \(5\leq n<m\), choose \(\ell=\lfloor(n-5)/2\rfloor\), giving \(2R<n+1\) for unique decoding and \(2\ell+5\leq n\) for the expected-score proof. With \(0<r<q\), nonempty groups and nonzero \(P\), SFS-DQI then has a larger expected score than DQI using the same \(P\). This choice is two degrees below the maximum allowed by the standard OPI DQI expected-score theorem. When \(r/q\to1/2\) and \(n/q\to\nu\in(0,1)\), it lowers the normalized score benchmark by only \(O(1/m)\). The difference vanishes as \(m\to\infty\). The polynomial preparation overhead therefore preserves the reported superpolynomial speed-up over known classical algorithms under the standard DQI comparison assumption~\cite{jordan2025decoded}. Supplementary Section VI gives the construction, resource analysis, proofs and explicit comparison assumption.

The numerical benchmark uses \(q=17\), \(n=7\), \(m=16\) and \(r=8\), with the first and last eight constraints forming \(G_1\) and \(G_2\). Each \(F_i\) is sampled independently and uniformly from the eight-element subsets of \(\mathbb F_{17}\). Fig.~\ref{fig:ncs-dqi-reweighting} reports equal-weight means over 50 instances. We use \(\ell=1\) and \(R=3\), satisfying both degree conditions above. Standard DQI with the optimal degree-three \(P\) also permits unique decoding and gives a mean expected score of \(12.09\), showing the finite-size cost of reducing the degree. Fig.~\ref{fig:ncs-dqi-reweighting} holds the degree-one \(P\) fixed to isolate structural reweighting.

For each instance, the output distributions are evaluated by exhaustive enumeration of all \(17^7\) polynomials. At degree one, maximizing the standard DQI expected-score quadratic form gives the normalized leading eigenvector \(w=(w_0,w_1)^T\) of the \(2\times2\) matrix with off-diagonal entries \(\sqrt m\) and diagonal entries \(0\) and \((q-2r)/\sqrt{r(q-r)}\). We evaluate \(P(f(Q))\) through the corresponding normalized elementary symmetric expansion over the prime field and obtain SFS-DQI probabilities by reweighting with \(J(Q)^2\) and normalizing. Fig.~\ref{fig:ncs-dqi-reweighting}c shows probabilities per polynomial for \(f(Q)=9,\ldots,16\), and Fig.~\ref{fig:ncs-dqi-reweighting}d groups the same distributions by \((S_1,S_2)\) profiles.

\vspace{.6cm}
\noindent\textbf{\large{Data availability}}\\
\noindent
The data generated in this study, including the source data for the computational figures, are available from Zenodo at \url{https://doi.org/10.5281/zenodo.23042152} (ref.~\citenum{jiang_2026_23042152}).

\vspace{.6cm}
\noindent\textbf{\large{Code availability}}\\
\noindent
The code used to generate the computational data and figures is available from Zenodo at \url{https://doi.org/10.5281/zenodo.23042152} (ref.~\citenum{jiang_2026_23042152}).

\begin{samepage}
\vspace{.6cm}
\noindent\textbf{\large{Acknowledgements}}\\
\noindent
The authors acknowledge support from the Zhejiang Provincial Natural Science Foundation of China (grant no.~LR25F020002), the Zhejiang Pioneer (Jianbing) Project (grant no.~2026C01006), and the National Natural Science Foundation of China (grant no.~62472374).
The authors also acknowledge support from the National Natural Science Foundation of China (grant no.~12405013) and the Guangdong Provincial Quantum Science Strategic Initiative (grant nos.~GDZX2503008 and GDZX2503001).
Xiao Yuan is supported by Quantum Science and Technology-National Science and Technology Major Project (2023ZD0300200), 
the National Natural Science Foundation of China Grant (No.~12361161602), 
NSAF (Grant No.~U2330201), 
Beijing Natural Science Foundation Z250004, 
Beijing Science and Technology Planning Project (Grant No.~Z25110100810000),
and the High-performance Computing Platform of Peking University. 
\end{samepage}

\begin{samepage}
\vspace{.6cm}
\noindent\textbf{\large{Author contributions}}\\
\noindent
J.Y., L.L. and X.Y. conceived the project.
Q.J., X.-M.Z. and D.X. carried out the research and wrote the manuscript.
All authors discussed the results and contributed to the final manuscript.
\end{samepage}

\begin{samepage}
\vspace{.6cm}
\noindent\textbf{\large{Competing interests}}\\
\noindent
The authors declare no competing interests.

\end{samepage}

\clearpage
\onecolumngrid
\setcounter{section}{0}
\setcounter{figure}{0}
\setcounter{table}{0}
\setcounter{algorithm}{0}
\setcounter{equation}{0}
\setcounter{theorem}{0}
\setcounter{lemma}{0}
\setcounter{corollary}{0}
\setcounter{secnumdepth}{3}
\renewcommand{\thefigure}{S\arabic{figure}}
\renewcommand{\thetable}{S\arabic{table}}
\renewcommand{\thealgorithm}{S\arabic{algorithm}}
\renewcommand{\theequation}{S\arabic{equation}}
\renewcommand{\thetheorem}{S\arabic{theorem}}
\renewcommand{\thelemma}{S\arabic{lemma}}
\renewcommand{\thecorollary}{S\arabic{corollary}}
\renewcommand{\theHfigure}{S\arabic{figure}}
\renewcommand{\theHtable}{S\arabic{table}}
\renewcommand{\theHalgorithm}{S\arabic{algorithm}}
\renewcommand{\theHequation}{S\arabic{equation}}
\renewcommand{\theHtheorem}{S\arabic{theorem}}
\renewcommand{\theHlemma}{S\arabic{lemma}}
\renewcommand{\theHcorollary}{S\arabic{corollary}}
\renewcommand\figurename{Supplementary Figure}
\renewcommand\tablename{Supplementary Table}
\setcounter{tocdepth}{2}

\begin{center}
  {\large\bfseries Supplementary Information\par}
  \vspace{0.5em}
  {\bfseries Harnessing problem structure for end-to-end quantum speed-ups\par}
\end{center}

\makeatletter
\let\supplementarybibcontext\the@ipfilectr
\tableofcontents
\global\let\the@ipfilectr\supplementarybibcontext
\makeatother
\clearpage

\section{Structural foundations}
\label{sec:sm-foundations}
We use recursive blocks to keep track of variable sets, residual constraints, and representatives. Balanced splits give recursion height \(O(\log n)\), while the disjoint block decomposition routes each feasible bit string through one interface label at every nonterminal block. These facts prove Lemma~\ref{lem:structural-recursion}. The controlled update below implements the required real-amplitude rotation.

\subsection{Recursive blocks and interface labels}
\label{subsec:sm-recursive-blocks}

Fix a block \((I,\beta)\) in the recursive decomposition, where \(I\subseteq[n]\) is its variable set and \(\beta\) is its residual constraint. Its feasible set and selected representative are denoted by \(F(I,\beta)\) and \(c_{I,\beta}\). At the root, \(I=[n]\) and \(F(I,\beta)=\mathcal F\). Terminal status may depend on the full block \((I,\beta)\). Every nonempty terminal feasible set contains one bit string. For a nonterminal block, the prescribed split \(I=I_L\sqcup I_R\) depends only on \(I\). Each feasible string in \(F(I,\beta)\) is then obtained from one compatible pair of child strings, and the parent feasible set decomposes as
\begin{equation}
F(I,\beta)
=
\bigsqcup_{\alpha\in\mathcal A(I,\beta)}
\Bigl(
F(I_L,\beta_L^{(\alpha)})
\times
F(I_R,\beta_R^{(\alpha)})
\Bigr).
\label{eq:sm-block-partition}
\end{equation}
Here \(\mathcal A(I,\beta)\) is the set of admissible interface labels. Each label \(\alpha\) records how \(\beta\) is allocated to the two child residual constraints. The union is disjoint for each fixed \((I,\beta)\). Different recursive paths can reach the same variable set \(I\) with different residual constraints. We collect these constraints in \(\mathcal R(I)\). Each block is assigned a representative chosen consistently with block isomorphisms. At each nonterminal block, the ordered labels begin with \(\alpha_1\) satisfying \(c_{I,\beta}=c_{I,\beta}^{(\alpha_1)}\), where \(c_{I,\beta}^{(\alpha)}=c_{I_L,\beta_L^{(\alpha)}}\sqcup c_{I_R,\beta_R^{(\alpha)}}\).

\begin{proof}[Proof of Lemma~\ref{lem:structural-recursion}]
To bound the height, follow any root-to-leaf path in the recursive decomposition. Let \(h(m)\) be the maximum remaining depth from a block with \(|I|=m\). For splits satisfying \(\max\{|I_L|,|I_R|\}\le\rho |I|\) with a fixed \(\rho<1\), the height obeys \(h(m)\le 1+h(\rho m)\) until the block size reaches a constant cutoff. Iteration gives \(h(n)\le t+h(\rho^t n)\). Taking \(t\) as the first level below the cutoff gives \(t=O(\log n)\), so the recursion height is \(O(\log n)\).

For uniqueness, fix \(x\in\mathcal F\) and apply the disjointness of Eq.~\eqref{eq:sm-block-partition} for each \((I,\beta)\) reached by its restrictions. At the root, \(x\) belongs to exactly one interface component, which fixes the two child residual constraints. Repeating the argument in both child blocks gives one residual constraint and one interface label at every later block reached by \(x\). This proves Lemma~\ref{lem:structural-recursion}.
\end{proof}

\subsection{Controlled local updates}
\label{subsec:sm-controlled-updates}

For two interface labels \(\alpha,\alpha'\in\mathcal A(I,\beta)\), a local update transfers amplitude between the product representatives \(c_{I,\beta}^{(\alpha)}\) and \(c_{I,\beta}^{(\alpha')}\). Let \(x,x'\in\{0,1\}^{|I|}\) be distinct source and target bit strings, and let \(K=\{i:x_i\ne x'_i\}\). The signed update pattern on the support is
\begin{equation}
\mu_i=
\begin{cases}
+1, & x_i=0,\ x'_i=1,\\
-1, & x_i=1,\ x'_i=0,
\end{cases}
\qquad i\in K .
\end{equation}
Set \(\sigma_i^{(+1)}=|1\rangle\langle0|\) and \(\sigma_i^{(-1)}=|0\rangle\langle1|\). The oriented support flip and its adjoint are
\begin{equation}
F_\mu
=
\bigotimes_{i\in K}\sigma_i^{(\mu_i)},
\qquad
F_\mu^\dagger
=
\bigotimes_{i\in K}\sigma_i^{(-\mu_i)},
\end{equation}
with identity on qubits outside \(K\). The real-amplitude update uses the Hermitian generator
\begin{equation}
H_{\rm up}(\mu)
=
i\bigl(F_\mu-F_\mu^\dagger\bigr).
\end{equation}
On the ordered two-dimensional subspace \((|x\rangle,|x'\rangle)\), it has the matrix form of Pauli \(Y\). Thus
\begin{equation}
e^{-i\theta H_{\rm up}(\mu)}|x\rangle
=
\cos\theta\,|x\rangle+\sin\theta\,|x'\rangle,\qquad
e^{-i\theta H_{\rm up}(\mu)}|x'\rangle
=
-\sin\theta\,|x\rangle+\cos\theta\,|x'\rangle .
\label{eq:sm-two-level-action}
\end{equation}
The first relation gives the expansion direction used during preparation. The second relation fixes the unitary extension on the other basis vector in the same two-dimensional subspace. Other two-level phase choices can be converted to this real form by diagonal phase gates, with no asymptotic overhead.

To apply the update only when a specified set of control qubits has the desired bit pattern, let \(C\) be disjoint from \(K\), and let \(\eta=(\eta_j)_{j\in C}\) be that pattern. The projector onto this control pattern is \(\Pi_\eta=\prod_{j\in C}|\eta_j\rangle\langle\eta_j|\). The controlled update is
\begin{equation}
\tau_\mu^\eta(\theta)
=
e^{-i\theta\Pi_\eta H_{\rm up}(\mu)} .
\end{equation}
If the control register does not match \(\eta\), the operation is the identity. If it matches, the operation applies Eq.~\eqref{eq:sm-two-level-action} on the selected source and target pair. When \(C=\varnothing\), we omit the superscript and write \(\tau_\mu(\theta)\). In circuit figures, concrete values of \(\eta\) are written above \(\tau\) and concrete components of \(\mu\) below it. Multiple components are displayed as indexed tuples.

The generic update is determined by its representative pair. Set \(C=I\setminus K\) and \(\eta=x_C=x'_C\). These controls select exactly \(|x\rangle\) and \(|x'\rangle\), so every other local basis state is unchanged. The update uses \(|I|\) support and control qubits and has \(O(|I|)\) depth and gate count under the ancilla-free synthesis below. Problem-specific constructions may use fewer controls when smaller patterns isolate the same pair.

For each variable set \(I\), the ordered chains used by the target distribution are merged in polynomial time into at most \(\Lambda(I)\) distinct block-local updates while preserving the order of every chain. At each step, an occupied local basis state shared by different residual blocks receives the same oriented two-level action and angle in all of them, or the identity in all of them. Once a branch reaches a terminal block, all later updates within that block must leave its basis state unchanged.

Alg.~\ref{alg:controlled-local-update} implements the controlled update by canonicalizing the support, applying a controlled middle block, and undoing the canonicalization. Choose a pivot \(p\in K\), order the remaining support as \(i_1,\ldots,i_{|K|-1}\), and use the signed pattern directly. A site with \(\mu_i=+1\) has source bit \(0\), and a site with \(\mu_i=-1\) has source bit \(1\). Write \(q_i\) for the qubit indexed by \(i\) and \(q_S\) for the register indexed by a set \(S\). Let \(P_\mu\) denote the support canonicalization circuit. Let \(R_K^\eta(\theta)\) denote the controlled middle block
\begin{equation}
R_K^\eta(\theta)
=
R_y(2\theta)_{p}
\quad\text{controlled on}\quad
q_{i_1}=\cdots=q_{i_{|K|-1}}=1,\ q_C=\eta .
\end{equation}
After canonicalization, the source and target basis states differ only on the pivot qubit \(q_p\). The full update is
\begin{equation}
\tau_\mu^\eta(\theta)
=
P_\mu^\dagger R_K^\eta(\theta)P_\mu .
\end{equation}
The convention for \(R_y(2\theta)\) is \(|0\rangle\mapsto\cos\theta|0\rangle+\sin\theta|1\rangle\) and \(|1\rangle\mapsto-\sin\theta|0\rangle+\cos\theta|1\rangle\), matching Eq.~\eqref{eq:sm-two-level-action}.

\begin{algorithm}[H]
\caption{Controlled local-update module for \(\tau_\mu^\eta(\theta)\)}
\label{alg:controlled-local-update}
\begin{algorithmic}[1]
\Statex \textbf{Input.} Signed update pattern \(\mu\), control pattern \(\eta\), rotation angle \(\theta\)
\Statex \textbf{Registers.} Support qubits \(q_K\), pivot \(q_p\in q_K\), control qubits \(q_C\) with \(C\cap K=\varnothing\)
\If{\(\mu_p=-1\)}
    \State Apply \(X\) on \(q_p\)
\EndIf
\For{each \(i_j\in K\setminus\{p\}\)}
    \State Apply \(\mathrm{CX}(q_p\rightarrow q_{i_j})\)
\EndFor
\For{each \(i_j\in K\setminus\{p\}\) with \(\mu_{i_j}=+1\)}
    \State Apply \(X\) on \(q_{i_j}\)
\EndFor
\For{each \(q_s\in q_C\) with \(\eta_s=0\)}
    \State Apply \(X\) on \(q_s\)
\EndFor
\State Apply \(R_y(2\theta)\) on \(q_p\), controlled by \(q_{K\setminus\{p\}}\) and \(q_C\)
\For{each \(q_s\in q_C\) with \(\eta_s=0\)}
    \State Apply \(X\) on \(q_s\)
\EndFor
\For{each \(i_j\in K\setminus\{p\}\) with \(\mu_{i_j}=+1\)}
    \State Apply \(X\) on \(q_{i_j}\)
\EndFor
\For{each \(i_j\in K\setminus\{p\}\)}
    \State Apply \(\mathrm{CX}(q_p\rightarrow q_{i_j})\)
\EndFor
\If{\(\mu_p=-1\)}
    \State Apply \(X\) on \(q_p\)
\EndIf
\end{algorithmic}
\end{algorithm}

In Alg.~\ref{alg:controlled-local-update}, the first three groups of gates implement \(P_\mu\). The temporary \(X\) gates before and after the controlled rotation convert zero controls in \(\eta\) to positive controls and restore them. Together with the controlled \(R_y(2\theta)\), they implement \(R_K^\eta(\theta)\). The final three groups of gates implement \(P_\mu^\dagger\). A concrete instance is shown in Fig.~\ref{fig:sm-controlled-local-update-example}. Here \(C=\{c_1,c_2\}\) is ordered as \((c_1,c_2)\), and \(K=\{i_1,i_2,p\}\) is ordered as \((i_1,i_2,p)\). The control pattern is \((\eta_{c_1},\eta_{c_2})=(1,0)\), and the support transition is \(x_K=011\) to \(x'_K=100\). On the selected control subspace, the circuit gives
\begin{equation}
|10\rangle_C|011\rangle_K
\mapsto
|10\rangle_C\bigl(\cos\theta\,|011\rangle_K+\sin\theta\,|100\rangle_K\bigr).
\label{eq:sm-figure-example-action}
\end{equation}
On the selected control subspace, the module has the expansion direction used by the local-update sequence on \(I\). Unitarity fixes the inverse action on the other vector in the same two-dimensional subspace, and all other control patterns are left unchanged.

\begin{figure}[h]
\centering
\begin{tikzpicture}[
    x=0.95cm,
    y=1cm,
    line width=0.55pt,
    gate/.style={draw,fill=white,minimum width=0.72cm,minimum height=0.46cm,inner sep=1.5pt,font=\small},
    widegate/.style={draw,fill=white,minimum width=1.55cm,minimum height=0.55cm,inner sep=2pt,font=\small},
    ctrl/.style={circle,fill,inner sep=1.8pt},
    targ/.style={circle,draw,fill=white,minimum size=0.30cm,inner sep=0pt},
    lab/.style={font=\small,anchor=east},
    note/.style={font=\small},
    stage/.style={draw,dashed,rounded corners=2pt,line width=0.45pt}
]
\foreach \y/\name in {0/{q_{c_1}},-0.9/{q_{c_2}},-1.8/{q_{i_1}},-2.7/{q_{i_2}},-3.6/{q_p}} {
    \node[lab] at (-0.35,\y) {$\name$};
    \draw (0,\y) -- (16.44,\y);
}

\node[gate] at (0.85,-3.6) {$X$};

\draw (2.05,-3.6) -- (2.05,-1.8);
\node[ctrl] at (2.05,-3.6) {};
\node[targ] at (2.05,-1.8) {};
\draw (1.93,-1.8) -- (2.17,-1.8);
\draw (2.05,-1.92) -- (2.05,-1.68);

\draw (3.30,-3.6) -- (3.30,-2.7);
\node[ctrl] at (3.30,-3.6) {};
\node[targ] at (3.30,-2.7) {};
\draw (3.18,-2.7) -- (3.42,-2.7);
\draw (3.30,-2.82) -- (3.30,-2.58);

\node[gate] at (4.50,-1.8) {$X$};
\node[gate] at (6.07,-0.9) {$X$};

\draw (8.22,0) -- (8.22,-3.6);
\node[ctrl] at (8.22,0) {};
\node[ctrl] at (8.22,-0.9) {};
\node[ctrl] at (8.22,-1.8) {};
\node[ctrl] at (8.22,-2.7) {};
\node[widegate] at (8.22,-3.6) {$R_y(2\theta)$};

\node[gate] at (10.37,-0.9) {$X$};
\node[gate] at (11.94,-1.8) {$X$};

\draw (13.14,-3.6) -- (13.14,-2.7);
\node[ctrl] at (13.14,-3.6) {};
\node[targ] at (13.14,-2.7) {};
\draw (13.02,-2.7) -- (13.26,-2.7);
\draw (13.14,-2.82) -- (13.14,-2.58);

\draw (14.39,-3.6) -- (14.39,-1.8);
\node[ctrl] at (14.39,-3.6) {};
\node[targ] at (14.39,-1.8) {};
\draw (14.27,-1.8) -- (14.51,-1.8);
\draw (14.39,-1.92) -- (14.39,-1.68);

\node[gate] at (15.59,-3.6) {$X$};

\draw[stage] (0.30,-4.05) rectangle (5.05,0.45);
\draw[stage] (5.47,-4.05) rectangle (10.97,0.45);
\draw[stage] (11.39,-4.05) rectangle (16.14,0.45);
\node[note] at (2.68,0.72) {$P_\mu$};
\node[note] at (8.22,0.72) {$R_K^\eta(\theta)$};
\node[note] at (13.77,0.72) {$P_\mu^\dagger$};
\end{tikzpicture}
\caption{Example controlled local-update circuit for \(x_K=011\to x'_K=100\) on \(K=\{i_1,i_2,p\}\), ordered as \((i_1,i_2,p)\), with control condition \((q_{c_1},q_{c_2})=(1,0)\). On this control subspace, the module implements Eq.~\eqref{eq:sm-figure-example-action}. The dashed boxes mark \(P_\mu\), \(R_K^\eta(\theta)\), and \(P_\mu^\dagger\). The middle block includes the temporary control flips.}
\label{fig:sm-controlled-local-update-example}
\end{figure}
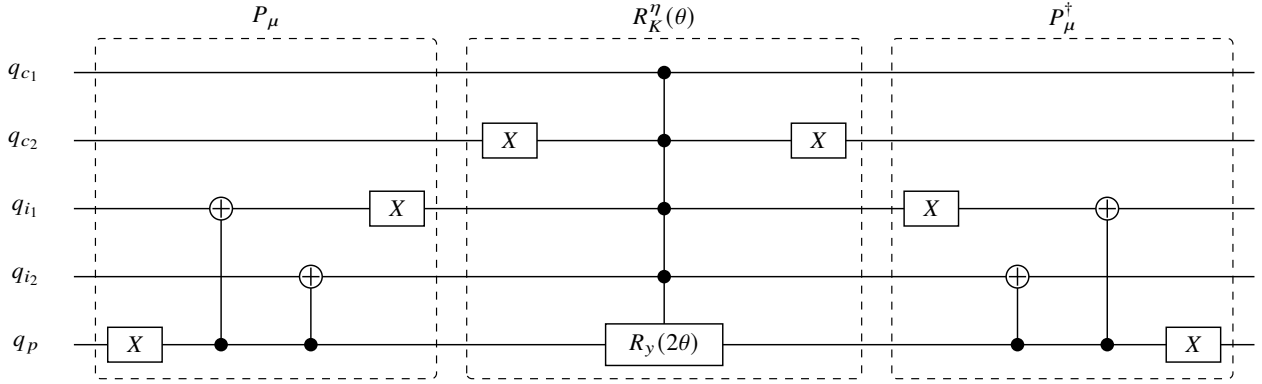

The resource count separates the support canonicalization from the middle multi-control. Let
\[
\kappa_u=|K|,\qquad
\kappa_c=|C|,\qquad
\kappa=\kappa_u+\kappa_c,
\]
so the middle \(R_y\) gate has \(\kappa-1\) positive controls. Also set
\[
z_K=|\{i\in K\setminus\{p\}:\mu_i=+1\}|,\qquad
z_C=|\{s\in C:\eta_s=0\}|,
\]
and let \(\delta_p=1\) when \(\mu_p=-1\) and \(\delta_p=0\) otherwise. We use \(G\) for gate count, \(D\) for depth, and \(A\) for the number of clean ancilla qubits. Before decomposing the middle multi-control, Alg.~\ref{alg:controlled-local-update} has
\[
G_X=2(\delta_p+z_K+z_C),\qquad
G_{\rm CX}=2(\kappa_u-1),
\]
where \(G_X\) counts all temporary \(X\) gates and \(G_{\rm CX}\) counts the CNOTs used for support canonicalization. With disjoint one-qubit gates placed in parallel, the depth outside the core with positive controls is
\[
D_{\rm out}=
\begin{cases}
2(\kappa_u-1)+2\delta_p+2, & z_K+z_C>0,\\
2(\kappa_u-1)+2\delta_p, & z_K+z_C=0 .
\end{cases}
\]
The extra \(+2\) is present only when at least one temporary \(X\) layer is needed on a nonpivot support qubit or a control qubit.

After the controls with \(\eta_s=0\) have been flipped, the core with at least two positive controls is implemented by
\[
R_y(\theta)_p\,\mathrm{C}^{\kappa-1}X\,R_y(-\theta)_p\,\mathrm{C}^{\kappa-1}X .
\]
A chosen synthesis of \(\mathrm{C}^{\kappa-1}X\) with elementary gate count \(g_{\kappa-1}\), circuit depth \(d_{\kappa-1}\), and clean ancilla count \(a_{\kappa-1}\) gives controlled local-update costs
\[
G=G_X+G_{\rm CX}+2g_{\kappa-1}+2,\qquad
D=D_{\rm out}+2d_{\kappa-1}+2,\qquad
A=a_{\kappa-1} .
\]
The small cases are direct. With no controls, the core is one \(R_y(2\theta)\). With one control, it is two CNOTs and two \(R_y\) rotations.

The resulting asymptotic tradeoff is
\begin{center}
\small
\setlength{\tabcolsep}{8pt}
\renewcommand{\arraystretch}{1.18}
\begin{tabular}{@{}ccc@{}}
\toprule
clean ancilla qubits & gate count & depth\\
\midrule
\(0\) & \(O(\kappa)\) & \(O(\kappa)\)\\
\(\Theta(\kappa)\) & \(O(\kappa)\) & \(O(\kappa_u+\log\kappa)\)\\
\bottomrule
\end{tabular}
\end{center}
The first row uses the linear-width multi-controlled-\(X\) synthesis adopted in the resource estimates~\cite{huang2024compiling}.
The last row uses a tree-style reduction of the middle multi-control, while the outer support canonicalization still contributes \(O(\kappa_u)\) depth. Theorem~\ref{thm:layerwise-preparation} uses the linear-width row without clean ancillas as the baseline setting.

\section{Recursive data, split weights and state preparation}
\label{sec:sm-split-weights}

This section bounds the recursive data, derives the split probabilities and proves the state preparation theorem. The uniform split masses are computed from the recursion itself.

\subsection{Recursive data bound}
\label{subsec:sm-recursive-data-bound}

\begin{lemma}[Polynomial recursive data]
\label{lem:recursive-data-bound}
For a fixed balanced split tree, a constructible decomposition with fixed expansion dimension \(d\) contains \(O(n^{d+1})\) block and interface records when \(d>0\), and \(O(n\log n)\) when \(d=0\). Isomorphism compression can only reduce these counts.
\end{lemma}

\begin{proof}[Proof of Lemma~\ref{lem:recursive-data-bound}]
Let \(r(I)=|\mathcal R(I)|\). The root satisfies \(r([n])=1\). Let \(J\) be a direct child of \(I\) in the fixed physical split tree. Every residual constraint on \(J\) is produced by an interface label of a residual block on \(I\). Hence
\begin{equation}
r(J)
\le
\sum_{\beta\in\mathcal R(I)}
|\mathcal A(I,\beta)|
=
r(I)+\Lambda(I).
\label{eq:sm-child-record-bound}
\end{equation}
Along a root-to-block path \(I_0=[n],I_1,\ldots,I_t\), repeated use of Eq.~\eqref{eq:sm-child-record-bound} gives \(r(I_t)\le1+\sum_{s=0}^{t-1}\Lambda(I_s)\). Balanced splitting gives \(|I_s|\le\rho^s n\). The resulting geometric sum gives \(r(I_t)=O(n^d)\) when \(d>0\), while \(r(I_t)=O(\log n)\) when \(d=0\).

The physical binary split tree has \(O(n)\) variable sets. Summing over them gives \(O(n^{d+1})\) residual records for \(d>0\) and \(O(n\log n)\) for \(d=0\). The number of interface records on \(I\) is at most \(r(I)+\Lambda(I)\), so it obeys the same total bounds. Each canonical record has polynomial length, and the local construction rule enumerates and deduplicates these records in polynomial time. Isomorphism compression can only reduce their number. This proves Lemma~\ref{lem:recursive-data-bound}.
\end{proof}

\subsection{Split weights and angle setting}
\label{subsec:sm-split-weight-angles}

Fix a nonterminal block \((I,\beta)\) and write its admissible interface labels as \(\alpha_1,\ldots,\alpha_s\), where \(s=|\mathcal A(I,\beta)|\). The first label satisfies \(c_{I,\beta}^{(\alpha_1)}=c_{I,\beta}\). For each label \(\alpha_j\), let \(M_{I,\beta}(\alpha_j)\ge0\) be a split mass and set \(M_{I,\beta}=\sum_{j=1}^s M_{I,\beta}(\alpha_j)\). For \(M_{I,\beta}>0\), the conditional probability is \(p_{I,\beta}(\alpha_j)=M_{I,\beta}(\alpha_j)/M_{I,\beta}\). For a nonuniform target distribution, the split probabilities may instead be specified directly.

For the uniform feasible state, the split masses are derived rather than supplied. Every nonempty terminal block is a singleton and has mass one. For a nonterminal block,
\begin{align}
M_{I,\beta}(\alpha)
&=
M_{I_L,\beta_L^{(\alpha)}}
M_{I_R,\beta_R^{(\alpha)}},
\label{eq:sm-uniform-branch-mass}\\
M_{I,\beta}
&=
\sum_{\alpha\in\mathcal A(I,\beta)}
M_{I,\beta}(\alpha).
\label{eq:sm-uniform-block-mass}
\end{align}
Induction from the terminal blocks gives \(M_{I,\beta}=|F(I,\beta)|\). Since \(M_{I,\beta}\le2^{|I|}\), every exact count uses at most \(|I|+1\) bits. Lemma~\ref{lem:recursive-data-bound} then makes the bottom-up calculation polynomial in the instance size.

For target distributions with zero-mass blocks, those blocks are omitted. The construction remains efficient when the nonzero split probabilities and resulting local update angles can be evaluated at the requested precision in polynomial time. These target-dependent quantities do not enter block isomorphism.

The masses set the update angles as follows. Define \(R_j(I,\beta)=\sum_{t=j}^{s}M_{I,\beta}(\alpha_t)\). Before update \(j\), let \(c_{I,\beta}^{(\alpha_j)}\) carry the residual mass \(R_j(I,\beta)\). The real-amplitude update in Eq.~\eqref{eq:sm-two-level-action} leaves mass \(R_j(I,\beta)\cos^2\theta_{I,\beta,j}\) on \(c_{I,\beta}^{(\alpha_j)}\) and moves mass \(R_j(I,\beta)\sin^2\theta_{I,\beta,j}\) to \(c_{I,\beta}^{(\alpha_{j+1})}\). Hence
\begin{equation}
\theta_{I,\beta,j}
=
\arcsin
\sqrt{\frac{R_{j+1}(I,\beta)}{R_j(I,\beta)}} ,
\qquad
j=1,\ldots,s-1,
\quad R_j(I,\beta)>0 .
\end{equation}
This relation gives \(R_j(I,\beta)\cos^2\theta_{I,\beta,j}=M_{I,\beta}(\alpha_j)\) and \(R_j(I,\beta)\sin^2\theta_{I,\beta,j}=R_{j+1}(I,\beta)\). By induction over \(j\), these updates prepare \(\sum_{\alpha\in\mathcal A(I,\beta)}\sqrt{p_{I,\beta}(\alpha)}\,|c_{I,\beta}^{(\alpha)}\rangle\) with nonnegative real amplitudes. The chain stops once the remaining mass vanishes.

Theorem~\ref{thm:layerwise-preparation} treats the angles as exact parameters in the ideal arbitrary-rotation model. For a numerical implementation, suppose each primitive rotation angle can be evaluated to absolute error \(2^{-b}\) in time polynomial in \(n\) and \(b\). Let \(N_{\rm rot}\) count all primitive arbitrary rotations after decomposing the controlled updates. Taking
\begin{equation}
b
=
\Theta\!\left(
\log N_{\rm rot}
+
\log\frac{1}{\varepsilon}
\right)
\label{eq:sm-angle-precision}
\end{equation}
with a sufficiently large constant bounds the accumulated operator-norm error by \(O(\varepsilon)\). The uniform masses are exact integers of polynomial bit length, so their angles satisfy this evaluation condition. Standard finite-gate-set synthesis adds its usual polylogarithmic overhead in \(N_{\rm rot}/\varepsilon\).

For fixed Hamming weight constraints, the split weight rule can be read off explicitly from binomial counts. For a block of length \(m\) and residual Hamming weight \(r\), we write \(c_{m,r}=0^{m-r}1^r\) for its right-packed representative. The root block for the fixed Hamming weight, or Dicke, example has \(m=4\) and \(r=2\), as shown in Fig.~\ref{fig:sm-dicke-mass-angle}. The three labels \(q=0,1,2\) have split weights \(M_0=1\), \(M_1=4\), and \(M_2=1\), so the residual masses are \(R_1=6\), \(R_2=5\), and \(R_3=1\). This angle relation gives \(\theta_1=\arcsin\sqrt{5/6}\) and \(\theta_2=\arcsin\sqrt{1/5}\). The same rule is applied recursively inside the nontrivial child blocks.

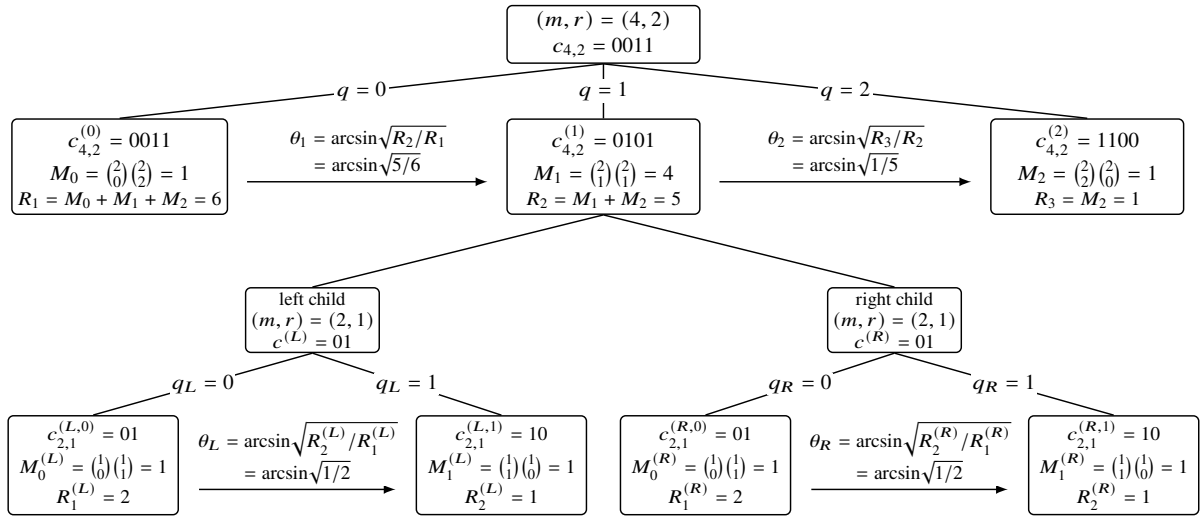
\begin{figure}[t]
\centering
\begin{tikzpicture}[
    x=1cm,
    y=1cm,
    >=latex,
    line width=0.55pt,
    every node/.style={font=\footnotesize},
    box/.style={draw, rounded corners=2pt, fill=white, inner sep=2.2pt, align=center},
    topbox/.style={box, minimum width=2.55cm},
    smallbox/.style={draw, rounded corners=2pt, fill=white, inner sep=1.7pt, align=center, font=\scriptsize},
    leafbox/.style={smallbox, minimum width=2.25cm},
    note/.style={font=\footnotesize},
    mass/.style={font=\footnotesize\itshape},
    ann/.style={font=\scriptsize}
]
\node[topbox] (root) at (0,0) {$(m,r)=(4,2)$\\$c_{4,2}=0011$};
\node[topbox] (q0) at (-6.4,-1.75) {$c_{4,2}^{(0)}=0011$\\$M_0=\binom20\binom22=1$\\{\scriptsize $R_1=M_0+M_1+M_2=6$}};
\node[topbox] (q1) at (0,-1.75) {$c_{4,2}^{(1)}=0101$\\$M_1=\binom21\binom21=4$\\{\scriptsize $R_2=M_1+M_2=5$}};
\node[topbox] (q2) at (6.4,-1.75) {$c_{4,2}^{(2)}=1100$\\$M_2=\binom22\binom20=1$\\{\scriptsize $R_3=M_2=1$}};
\draw (root.south) -- node[midway, fill=white, inner sep=1pt] {$q=0$} (q0.north);
\draw (root.south) -- node[midway, fill=white, inner sep=1pt] {$q=1$} (q1.north);
\draw (root.south) -- node[midway, fill=white, inner sep=1pt] {$q=2$} (q2.north);
\draw[->] ([xshift=0.25cm,yshift=-0.20cm]q0.east) -- node[ann, above=1pt, fill=white, inner sep=0.5pt, align=center] {$\theta_1=\arcsin\!\sqrt{R_2/R_1}$\\$=\arcsin\!\sqrt{5/6}$} ([xshift=-0.25cm,yshift=-0.20cm]q1.west);
\draw[->] ([xshift=0.25cm,yshift=-0.20cm]q1.east) -- node[ann, above=1pt, fill=white, inner sep=0.5pt, align=center] {$\theta_2=\arcsin\!\sqrt{R_3/R_2}$\\$=\arcsin\!\sqrt{1/5}$} ([xshift=-0.25cm,yshift=-0.20cm]q2.west);
\node[smallbox] (l21) at (-3.85,-3.78) {left child\\$(m,r)=(2,1)$\\$c^{(L)}=01$};
\node[smallbox] (r21) at (3.85,-3.78) {right child\\$(m,r)=(2,1)$\\$c^{(R)}=01$};
\draw (q1.south) -- (l21.north);
\draw (q1.south) -- (r21.north);
\node[leafbox] (l01) at (-6.75,-5.70) {$c_{2,1}^{(L,0)}=01$\\$M^{(L)}_0=\binom10\binom11=1$\\{\scriptsize $R^{(L)}_1=2$}};
\node[leafbox] (l10) at (-1.35,-5.70) {$c_{2,1}^{(L,1)}=10$\\$M^{(L)}_1=\binom11\binom10=1$\\{\scriptsize $R^{(L)}_2=1$}};
\node[leafbox] (r01) at (1.35,-5.70) {$c_{2,1}^{(R,0)}=01$\\$M^{(R)}_0=\binom10\binom11=1$\\{\scriptsize $R^{(R)}_1=2$}};
\node[leafbox] (r10) at (6.75,-5.70) {$c_{2,1}^{(R,1)}=10$\\$M^{(R)}_1=\binom11\binom10=1$\\{\scriptsize $R^{(R)}_2=1$}};
\draw (l21.south) -- node[midway, fill=white, inner sep=0.5pt] {$q_L=0$} (l01.north);
\draw (l21.south) -- node[midway, fill=white, inner sep=0.5pt] {$q_L=1$} (l10.north);
\draw (r21.south) -- node[midway, fill=white, inner sep=0.5pt] {$q_R=0$} (r01.north);
\draw (r21.south) -- node[midway, fill=white, inner sep=0.5pt] {$q_R=1$} (r10.north);
\draw[->] (-5.35,-6.02) -- node[ann, above=1pt, fill=white, inner sep=0.5pt, align=center] {$\theta_L=\arcsin\!\sqrt{R^{(L)}_2/R^{(L)}_1}$\\$=\arcsin\!\sqrt{1/2}$} (-2.75,-6.02);
\draw[->] (2.75,-6.02) -- node[ann, above=1pt, fill=white, inner sep=0.5pt, align=center] {$\theta_R=\arcsin\!\sqrt{R^{(R)}_2/R^{(R)}_1}$\\$=\arcsin\!\sqrt{1/2}$} (5.35,-6.02);
\end{tikzpicture}
\caption{Fixed Hamming weight, or Dicke, example for split weights and angle setting. The root block \((m,r)=(4,2)\) splits into two children of size 2, with interface label \(q\) equal to the weight of the left child. The uniform weights are the binomial counts \(M_q=\binom{2}{q}\binom{2}{2-q}\), namely \(1,4,1\). The horizontal arrows show the ordered updates among the three root representatives. Their residual weights satisfy \(R_1=M_0+M_1+M_2=6\), \(R_2=M_1+M_2=5\), and \(R_3=M_2=1\), giving \(\theta_1=\arcsin\sqrt{R_2/R_1}=\arcsin\sqrt{5/6}\) and \(\theta_2=\arcsin\sqrt{R_3/R_2}=\arcsin\sqrt{1/5}\). The middle case \(q=1\) further decomposes into left and right blocks of size 2 and weight 1. Their local weights and residual weights are denoted by \(M^{(L)},R^{(L)}\) and \(M^{(R)},R^{(R)}\), and both give the local angle \(\arcsin\sqrt{1/2}\).}
\label{fig:sm-dicke-mass-angle}
\end{figure}

Given a prescribed feasible string \(x^\star\), the constructible description identifies its unique label at each split. Assigning probability one to these labels prepares \(|x^\star\rangle\).

\begin{corollary}[Uniform-state construction]
\label{cor:uniform-state-construction}
Under the hypotheses of Theorem~\ref{thm:layerwise-preparation} for the uniform distribution with \(d=O(1)\), the recursive update schedule and exact split probabilities can be generated in polynomial classical time without enumerating \(\mathcal F\). The required local update angles can be evaluated to \(b\) bits in time polynomial in \(n\) and \(b\).
\end{corollary}

\begin{proof}[Proof of Corollary~\ref{cor:uniform-state-construction}]
Lemma~\ref{lem:recursive-data-bound} gives polynomially many records for fixed \(d\). A bottom-up pass evaluates the exact uniform masses in Eqs.~\eqref{eq:sm-uniform-branch-mass} and \eqref{eq:sm-uniform-block-mass} using integers of at most \(n+1\) bits. The representative pairs determine the controlled updates, and the resulting angles can be evaluated to \(b\) bits in time polynomial in \(n\) and \(b\). No step enumerates \(\mathcal F\).
\end{proof}

\subsection{Proof of the state preparation theorem}
\label{subsec:sm-preparation-proof}

\begin{proof}[Proof of Theorem~\ref{thm:layerwise-preparation}]
The proof separates the probability calculation from the circuit count. Let \(H\) be the number of recursion levels containing nonterminal blocks. For \(t=0,1,\ldots,H-1\), let \(\mathcal B_t\) contain the variable sets at level \(t\) that support at least one nonterminal block. Fix \(x\in\mathcal F\), and let \(\mathcal B_t(x)\subseteq\mathcal B_t\) contain the variable sets whose nonterminal blocks are reached when \(x\) is restricted recursively. For every \(I\in\mathcal B_t(x)\), this restriction selects a unique residual constraint \(\beta_I(x)\in\mathcal R(I)\) and interface label \(\alpha_I(x)\). A branch that reaches a terminal block stops and contributes no further conditional factor. The recursive probability assignment gives
\begin{equation}
\pi(x)
=
\prod_{t=0}^{H-1}
\prod_{I\in\mathcal B_t(x)}
p_{I,\beta_I(x)}\bigl(\alpha_I(x)\bigr),
\label{eq:sm-unrolled-distribution}
\end{equation}

On the branch containing \(x\), the common schedule on each \(I\in\mathcal B_t(x)\) realizes the ordered representative expansion for \(\beta_I(x)\) and contributes the factor \(p_{I,\beta_I(x)}(\alpha_I(x))\). Other occupied representatives are unchanged. Distinct variable sets \(I\) at the same level are disjoint, so their updates act on disjoint qubits. Equation~\eqref{eq:sm-unrolled-distribution} and induction over the levels give squared amplitude \(\pi(x)\) on each feasible basis state. The output is
\begin{equation}
|\psi_\pi\rangle
=
\sum_{x\in\mathcal F}\sqrt{\pi(x)}\,|x\rangle ,
\end{equation}
with the real nonnegative amplitudes claimed in Theorem~\ref{thm:layerwise-preparation}. Feasibility is preserved because every update connects representatives of admissible product components and acts as the identity on every other occupied representative.

Balanced recursion gives \(H=O(\log n)\), and \(m_t=\max_{I\in\mathcal B_t}|I|\) satisfies \(m_t\le\rho^t n\) for a fixed \(\rho<1\). All constants hidden below are independent of \(n\).

Consider one \(I\in\mathcal B_t\) with \(|I|=m\). Its schedule contains at most \(\Lambda(I)=O(m^d)\) updates. Each update has depth and gate count \(O(m)\) by the ancilla-free implementation in Sec.~\ref{subsec:sm-controlled-updates}, so sequential execution has depth and gate count \(O(m^{d+1})\).

For depth, the sets \(I\in\mathcal B_t\) are pairwise disjoint, so level \(t\) contributes only the largest depth associated with one \(I\). The depth obeys
\begin{equation}
D
\le
O\!\left(\sum_{t=0}^{H-1} m_t^{d+1}\right)
\le
O\!\left(\sum_{t=0}^{H-1}(\rho^t n)^{d+1}\right)
=
O\!\left(
n^{d+1}\sum_{t=0}^{H-1}\rho^{t(d+1)}
\right)
\le
O(n^{d+1}) .
\end{equation}
The finite geometric sum is bounded by \(1/(1-\rho^{d+1})\), a constant independent of \(n\), because \(d\ge0\) and \(\rho<1\).

For gates, we sum the costs associated with all \(I\). At level \(t\), the sets \(I\in\mathcal B_t\) are pairwise disjoint subsets of \([n]\), so \(\sum_{I\in\mathcal B_t}|I|\le n\). Since \(|I|\le m_t\), we have \(|I|^{d+1}\le m_t^d|I|\). This gives
\begin{equation}
G
\le
O\!\left(
\sum_{t=0}^{H-1}
\sum_{I\in\mathcal B_t}
|I|^{d+1}
\right)
\le
O\!\left(
\sum_{t=0}^{H-1}
n\,m_t^d
\right)
\le
O\!\left(
n^{d+1}\sum_{t=0}^{H-1}\rho^{td}
\right)
\le
O(n^{d+1}+n\log n) .
\end{equation}
For \(d>0\), the finite geometric sum is bounded by \(1/(1-\rho^d)\), giving \(O(n^{d+1})\). For \(d=0\), each level contributes \(O(n)\) gates and there are \(O(\log n)\) levels, giving the \(O(n\log n)\) term. Preparing the root representative basis state costs at most \(O(n)\) additional single-qubit gates. Combining the probability calculation with these bounds proves Theorem~\ref{thm:layerwise-preparation}.
\end{proof}

\section{Representative constrained families}
\label{sec:sm-families}

We give the constructions behind the constrained-family rows in the main text. For each family, we specify its recursive records and representatives, derive the uniform split masses and verify the update schedules and resource bounds.

\subsection{Fixed Hamming weight constraints}
\label{subsec:sm-hamming}

The hard constraint is \(F_{n,k}=\{x\in\{0,1\}^n\mid \sum_{i=1}^n x_i=k\}\). Within a block of length \(m\), \(r\) is the residual Hamming weight. We use the right-packed representative \(c_{m,r}=0^{m-r}1^r\).

The feasible sets at \(k=0\) and \(k=n\) each contain a single bit string and require no recursive preparation. If \(k>n/2\), the weight-\(k\) subspace is obtained from the weight-\((n-k)\) subspace by applying \(X\) to all qubits. It suffices to consider \(1\le k\le n/2\).

Split the block into two nearly balanced child blocks of lengths \(m_L\) and \(m_R\). The interface label is the weight of the left child, \(q\in\{\max(0,r-m_R),\ldots,\min(r,m_L)\}\).
The child residual weights are \(q\) and \(r-q\), and the corresponding representative is
\begin{equation}
c_{m,r}^{(q)}
=
c_{m_L,q}\sqcup c_{m_R,r-q}
=
0^{m_L-q}1^q\,0^{m_R-r+q}1^{r-q}.
\end{equation}
The local feasible set decomposes as \(F_{m,r}=\bigsqcup_q F_{m_L,q}\times F_{m_R,r-q}\). The cases \(r=0\) and \(r=m\) each contain one bit string and contribute no updates. For a variable set \(I\) of size \(m\), only residual weights satisfying \(1\le r\le\min(k,m-1)\) require expansion, and each fixed \(r\) has \(O(r)\) nontrivial label updates. Hence
\begin{equation}
\Lambda(I)
=
\sum_{r\in\mathcal R(I)}
\bigl(|\mathcal A(I,r)|-1\bigr)
\le
\sum_{r=1}^{\min(k,m-1)}
\bigl(|\mathcal A(I,r)|-1\bigr)
=
O(\min(k,m)^2)
\le O(m^2).
\end{equation}
For \(k=n/2\), either child of the root has size \(m=n/2\) and carries every residual weight \(r=0,\ldots,m\). The two endpoint weights contribute zero, leaving the sum \(\Theta(m^2)\). For the family in which \(k\) may scale with \(n\), the recursive expansion dimension is therefore \(d=2\).

For the uniform feasible superposition, the split weight routed through \(q\) is the number of terminal feasible strings under that label, \(M_q=\binom{m_L}{q}\binom{m_R}{r-q}\). The residual weight rule in Sec.~\ref{subsec:sm-split-weight-angles} fixes the angles. The residual records, ordered labels, child records, representatives and masses are computed directly from \((m,r)\), while a feasible string identifies \(q\) by its left child weight. Hence the recursion is constructible. Adjacent representatives \(c_{m,r}^{(q)}\) and \(c_{m,r}^{(q+1)}\) differ by moving one excitation from the right child to the left child. Let \(a=m_L-q\) be the changing site in the left child, and let \(b=m_L+m_R-r+q+1\) be the changing site in the right child. On these support bits, \(c_{m,r}^{(q)}\) has pattern \(01\), while \(c_{m,r}^{(q+1)}\) has pattern \(10\).

\begin{lemma}[Block-local controls for Dicke updates]
\label{lem:sm-dicke-block-local-controls}
Within the subspace spanned by the packed representatives for all residual weights in \(\mathcal R(I)\), the pair of neighboring representatives \(c_{m,r}^{(q)}\) and \(c_{m,r}^{(q+1)}\) can be isolated by a constant-size controlled update. Its support is \(\{a,b\}\), and at most four controls are needed, all inside \(I\).
\end{lemma}

\begin{proof}
The support bits \(a\) and \(b\) are updated, not used as controls. To isolate the intended pair among the currently populated packed representatives, use the neighboring conditions around the two right-packed interfaces. These are \(x_{a-1}=0\), \(x_{a+1}=1\), \(x_{b-1}=0\), and \(x_{b+1}=1\), omitting any condition whose site lies outside the corresponding child block.

These controls also handle the reverse action required by unitarity. The support pattern on \((a,b)\) is not fixed by the controls. It is the two-level subspace containing both \(01\) and \(10\). The neighboring conditions identify the packed boundaries of the left and right children, which fix \(q\), \(p=r-q\) and \(r=q+p\). The only compatible occupied representatives are \(c_{m,r}^{(q)}\) and \(c_{m,r}^{(q+1)}\). The update leaves every other occupied representative for the residual weights in \(\mathcal R(I)\) unchanged. All support and control bits lie inside \(I\).
\end{proof}

\begin{corollary}[Fixed Hamming weight resources]
For \(1\le k\le n/2\), the fixed Hamming weight construction gives an ancilla-free preparation with depth \(O(k\log(n/k))\) and gate count \(O(kn)\).
\end{corollary}

\begin{proof}
A more local ordering improves on the generic \(d=2\) bound. For a fixed residual weight \(r\) that requires expansion, there are \(O(r)\) ordered updates between neighboring labels. Lemma~\ref{lem:sm-dicke-block-local-controls} makes each update constant size, with two support bits and at most four controls. For a variable set \(I\) of size \(m\), only residual weights \(1\le r\le\min(k,m-1)\) require expansion. The total number of neighboring-label updates over all residual weights in \(\mathcal R(I)\) is at most
\[
\sum_{r=1}^{\min(k,m-1)} O(r)=O(\min(k,m)^2).
\]

The updates inside one variable set \(I\) admit a constant-color schedule. Write \(p=r-q\) for the weight in the right child before a transfer. An update is indexed by \((q,p)\), changes \(q\to q+1\) and \(p\to p-1\), and uses only the two support bits and the constant-size controls of Lemma~\ref{lem:sm-dicke-block-local-controls}. All support and control bits lie in constant-width neighborhoods of the left-packed boundary \(q\) and the right-packed boundary \(p\). Let \(w\) bound these neighborhood widths. Choose fixed constants \(L>2w\) and \(A>2w+1\). Schedule \((q,p)\) at time \(s=Aq+p\), and within each \(s\), execute the constant number of colors \(q\bmod L\) sequentially. Along a residual chain \(q+p=r\), replacing \((q,p)\) by \((q+1,p-1)\) increases \(s\) by \(A-1>0\), so the order is preserved. If two updates have the same \(s\) and color, their \(q\)-values differ by at least \(L\). Since \(p-p'=-A(q-q')\), their \(p\)-values differ by at least \(AL\). Their neighborhoods, including all control bits, are disjoint. The number of colors is constant and \(s\) takes only \(O(\min(k,m))\) values. The schedule has depth \(O(\min(k,m))\) and gate count \(O(\min(k,m)^2)\).

For the full recursion, let \(m_t\) be the largest block size on layer \(t\). The layer depth is \(O(\min(k,m_t))\), because blocks on that layer are disjoint. In the Hamming weight construction we use balanced bisection, so \(m_t=\Theta(n2^{-t})\) until the terminal layers, up to constant-size rounding effects. Let \(T\) be the first layer with \(m_T<k\), and otherwise take \(T\) to be the number of recursion layers. The geometric decrease gives \(T=O(\log(n/k))\). After the cutoff, the remaining block sizes sum to \(O(k)\). Splitting the depth sum at \(T\) gives
\begin{equation}
D
\le
O\!\left(\sum_t \min(k,m_t)\right)
=
O\!\left(
\sum_{t<T} k+\sum_{t\ge T} m_t
\right)
=
O\!\left(
k\log(n/k)+k
\right)
=
O(k\log(n/k)).
\end{equation}
For gates, layer \(t\) has \(O(n/m_t)\) blocks, and each block contributes \(O(\min(k,m_t)^2)\) updates between neighboring labels. The gate count obeys
\(G
\le
O\!\left(\sum_t \frac{n}{m_t}\min(k,m_t)^2\right)
\).
For \(t<T\), the summand is \(O(nk^2/m_t)\). Since \(m_t=\Theta(n2^{-t})\), these terms grow geometrically and are dominated by the last layer before the cutoff, where \(m_t=\Omega(k)\). Their total contribution is \(O(nk)\). For \(t\ge T\), the summand is \(O(nm_t)\), and the remaining geometric series satisfies \(\sum_{t\ge T}m_t=O(k)\). The two ranges give
\begin{equation}
G
\le
O\!\left(
\sum_{t<T}\frac{nk^2}{m_t}
+
\sum_{t\ge T}n m_t
\right)
=
O(nk).
\end{equation}
Together with the depth bound, this gives the stated fixed Hamming weight bounds.
\end{proof}

\begin{figure}[t]
\centering
\begin{tikzpicture}[
wire/.style={line width=0.45pt},
xgate/.style={draw, fill=white, minimum size=0.38cm, inner sep=0pt, font=\footnotesize},
rootgate/.style={draw, rounded corners=2pt, fill=white, minimum width=2.80cm, minimum height=2.70cm, align=center, font=\footnotesize},
childgate/.style={draw, rounded corners=2pt, fill=white, minimum width=3.10cm, minimum height=1.02cm, align=center, font=\footnotesize},
lab/.style={font=\footnotesize, anchor=east},
note/.style={font=\footnotesize, align=center}
]
\def\xstart{-0.45}
\def\xend{12.0}
\foreach \y/\name in {0/\(x_1\),-0.65/\(x_2\),-1.30/\(x_3\),-1.95/\(x_4\)} {
  \node[lab] at (\xstart-0.25,\y) {\name};
  \draw[wire] (\xstart,\y) -- (\xend,\y);
}
\draw[dashed, rounded corners=2pt] (-0.20,0.38) rectangle (0.70,-2.33);
\node[note] at (0.25,0.64) {seed \(=|0011\rangle\)};
\node[xgate] at (0.25,-1.30) {\(X\)};
\node[xgate] at (0.25,-1.95) {\(X\)};

\node[rootgate] (u1) at (2.50,-0.975)
{\(U_1\)\\[0.25em]
\(\tau_{(\mu_2,\mu_3)=(+1,-1)}(\theta_1)\)\\[0.25em]
\(0011\to0101\)};
\node[note] at (2.50,-2.72) {\(\theta_1=\arcsin\sqrt{5/6}\)};

\node[rootgate] (u2) at (5.70,-0.975)
{\(U_2\)\\[0.25em]
\(\tau_{(\mu_1,\mu_4)=(+1,-1)}^{(x_2,x_3)=(1,0)}(\theta_2)\)\\[0.25em]
\(0101\to1100\)};
\node[note] at (5.70,-2.72) {\(\theta_2=\arcsin\sqrt{1/5}\)};

\draw[dashed, rounded corners=2pt] (7.50,0.38) rectangle (11.80,-2.33);
\node[note] at (9.65,0.64) {parallel child layer};

\node[childgate] (ul) at (9.65,-0.325)
{\(U_L\)\\[-0.1em]
\(\tau_{(\mu_1,\mu_2)=(+1,-1)}(\theta_L)\), \(01\to10\)};

\node[childgate] (ur) at (9.65,-1.625)
{\(U_R\)\\[-0.1em]
\(\tau_{(\mu_3,\mu_4)=(+1,-1)}(\theta_R)\), \(01\to10\)};
\node[note] at (9.65,-2.72) {\(\theta_L=\theta_R=\arcsin\sqrt{1/2}\)};

\node[note, anchor=west] at (11.97,-0.975)
{output\\[-0.1em]
\(\frac{1}{\sqrt6}\bigl(|0011\rangle+|0101\rangle+|0110\rangle\)\\[-0.1em]
\(\hspace{1.1em}+|1001\rangle+|1010\rangle+|1100\rangle\bigr)\)};
\end{tikzpicture}
\caption{Controlled local-update circuit for the \((n,k)=(4,2)\) fixed Hamming weight example. The seed block prepares \(|0011\rangle\). The two root updates create the representatives \(0011\), \(0101\), and \(1100\), and the dashed child block refines the \(q=1\) representative in parallel. The output is the uniform superposition of the six Hamming-weight-two basis states.}
\label{fig:sm-dicke-controlled-circuit}
\end{figure}
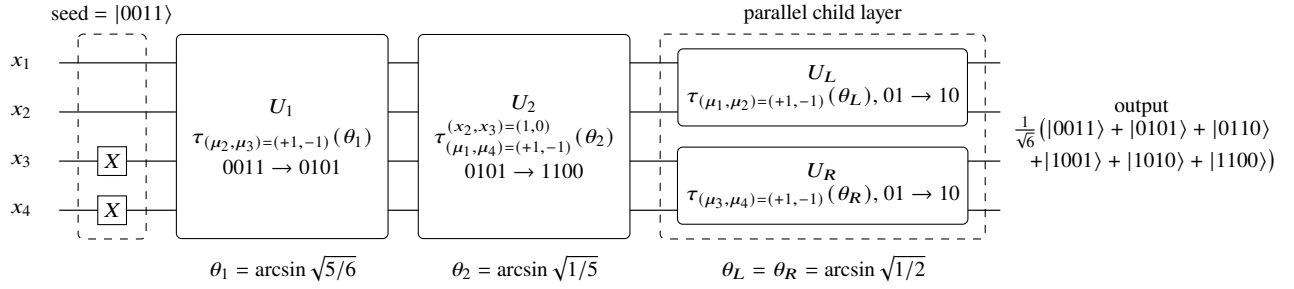

The \((n,k)=(4,2)\) example in Fig.~\ref{fig:sm-dicke-controlled-circuit} makes the role of controls explicit. The first two updates use the root angles from Fig.~\ref{fig:sm-dicke-mass-angle}. They prepare a superposition over the three root representative basis states \(|0011\rangle\), \(|0101\rangle\) and \(|1100\rangle\). The \(q=1\) representative is refined by two child updates with the local angle \(\arcsin\sqrt{1/2}\). These child updates act on disjoint supports and are scheduled in the same layer.

This example illustrates why the local update circuit is controlled rather than just a bare support update on the changing bits. A support pattern can appear in several already populated representatives after earlier updates. The control condition selects the intended pair of representative basis states and leaves the other representatives unchanged. When the current representatives are already separated by local patterns, as in the second refinement layer above, the same update can be used without extra controls.

\subsection{Rydberg blockade constraints}
\label{subsec:sm-rydberg}

The hard constraint is \(x_i+x_{i+1}\le1\) for \(i=1,\ldots,n-1\). To make interval recursion closed, attach a boundary condition \(\beta=(\lambda,\rho)\) to an interval \(I=[a,b]\), where \(\lambda,\rho\in\{\star,0,1\}\) either leave an endpoint free or fix it to the indicated value. Let \(F(I;\lambda,\rho)\) be the bit strings on \(I\) satisfying the blockade constraints inside \(I\) and the endpoint requirements. Empty sets are omitted. The representative \(c_{I;\lambda,\rho}\) is the all-zero string on \(I\), except that forced endpoint values equal to \(1\) are retained.

A one-site root is prepared directly by a single-qubit rotation. Size-one intervals below a split are terminal because the inherited interface value fixes the bit. For \(|I|=2\), the endpoint-compatible subset of \(\{00,10,01\}\) gives the split weights for the final update into two size-one intervals.

Split a nonterminal interval \(I=[a,b]\) at a site \(s\), with \(I_L=[a,s]\) and \(I_R=[s+1,b]\). The interface label is \(\alpha=(u,v)=(x_s,x_{s+1})\), namely the occupancies of the two sites adjacent to the split. The admissible interface labels are
\begin{equation}
\mathcal A(I;\lambda,\rho)
=
\bigl\{(u,v)\in\{0,1\}^2:\ u+v\le1,\ 
F(I_L;\lambda,u)\ne\varnothing,\ 
F(I_R;v,\rho)\ne\varnothing
\bigr\}.
\end{equation}
They give the disjoint decomposition
\begin{equation}
F(I;\lambda,\rho)
=
\bigsqcup_{\alpha=(u,v)\in\mathcal A(I;\lambda,\rho)}
F(I_L;\lambda,u)\times F(I_R;v,\rho).
\end{equation}
Both the nonterminal residual boundary conditions and the interface labels take only a constant number of values. Thus \(\Lambda(I)=O(1)\) and the recursive expansion dimension is \(d=0\).

For \(\alpha=(u,v)\), the representative is \(c_I^{(\alpha)}=c_{I_L;\lambda,u}\sqcup c_{I_R;v,\rho}\). For a fixed parent block, the representatives \(c_I^{(\alpha)}\) are identical on every site except possibly the two boundary sites at the split. For the uniform state, let \(N(I;\lambda,\rho)=|F(I;\lambda,\rho)|\) be the number of feasible strings in a block, and let \(M_{u,v}\) be the split weight routed through \((u,v)\). Then
\begin{equation}
N(I;\lambda,\rho)
=
\sum_{(u,v)\in\mathcal A(I;\lambda,\rho)}
N(I_L;\lambda,u)N(I_R;v,\rho),
\qquad
M_{u,v}=N(I_L;\lambda,u)N(I_R;v,\rho).
\end{equation}
These recurrences give the split masses and angles bottom up. The boundary records, labels, representatives and counts are obtained by a constant-state recursion, while a feasible string identifies its label as \((x_s,x_{s+1})\). Hence the description is constructible.

\begin{lemma}[Block-local controls for Rydberg updates]
\label{lem:sm-rydberg-block-local-controls}
For each interval block \(I=[a,b]\), every update between neighboring interface labels can be isolated using \(O(1)\) support and control bits, all inside \(I\).
\end{lemma}

\begin{proof}
There are only three possible interface labels, \(00\), \(10\), and \(01\). The following table lists each pair of neighboring labels and the controls that isolate it. Listed checks whose sites lie outside \(I\) are omitted.
\begin{center}
\small
\setlength{\tabcolsep}{10pt}
\renewcommand{\arraystretch}{1.18}
\begin{tabular}{@{}lll@{}}
\toprule
Interface label pair & Update support & Required zero controls \\
\midrule
\(\{00,10\}\) & \(\{s\}\) &
\(x_{s+1}=0\), and \(x_{s-1}=0\) if \(s-1\in I\) \\
\addlinespace[2pt]
\(\{00,01\}\) & \(\{s+1\}\) &
\(x_s=0\), and \(x_{s+2}=0\) if \(s+2\in I\) \\
\addlinespace[2pt]
\(\{10,01\}\) & \(\{s,s+1\}\) &
\(x_{s-1}=0\) if \(s-1\in I\), and \(x_{s+2}=0\) if \(s+2\in I\) \\
\bottomrule
\end{tabular}
\end{center}
The listed checks guarantee that both endpoints of the two-level subspace obey the blockade constraint at the split and at its two neighboring edges. When the boundary conditions in \(\mathcal R(I)\) are handled together, the update is also controlled on \(x_a=\lambda\) for \(\lambda\in\{0,1\}\) and on \(x_b=\rho\) for \(\rho\in\{0,1\}\), omitting controls already used as support or boundary checks. At a global boundary, the corresponding endpoint condition is always \(\star\), so no competing fixed condition occurs there. These controls distinguish all boundary conditions in \(\mathcal R(I)\). Each update uses at most two support qubits, two nearest-neighbor checks, and two endpoint selectors, all inside \(I\).
\end{proof}

\begin{corollary}[Rydberg blockade resources]
The one-dimensional Rydberg blockade construction gives an ancilla-free preparation with depth \(O(\log n)\) and gate count \(O(n)\).
\end{corollary}

\begin{proof}
For a fixed interval \(I\), only a constant number of boundary conditions occur, and each condition has at most three interface labels. The local-update sequence on \(I\) contains \(O(1)\) updates. Lemma~\ref{lem:sm-rydberg-block-local-controls} implements each update with \(O(1)\) depth and gates. Intervals in the same recursion layer are disjoint, so the layer depth is \(O(1)\). A balanced interval recursion has \(O(\log n)\) layers and \(O(n)\) nonterminal intervals in total, giving depth \(O(\log n)\) and gate count \(O(n)\).
\end{proof}

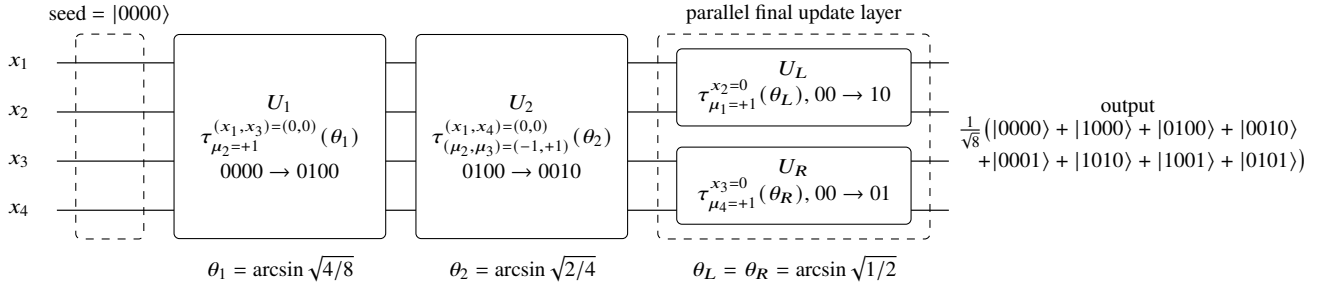
\begin{figure}[t]
\centering
\begin{tikzpicture}[
wire/.style={line width=0.45pt},
rootgate/.style={draw, rounded corners=2pt, fill=white, minimum width=2.80cm, minimum height=2.70cm, align=center, font=\footnotesize},
childgate/.style={draw, rounded corners=2pt, fill=white, minimum width=3.10cm, minimum height=1.02cm, align=center, font=\footnotesize},
lab/.style={font=\footnotesize, anchor=east},
note/.style={font=\footnotesize, align=center}
]
\def\xstart{-0.45}
\def\xend{11.35}

\foreach \y/\name in {0/\(x_1\),-0.65/\(x_2\),-1.30/\(x_3\),-1.95/\(x_4\)} {
  \node[lab] at (\xstart-0.25,\y) {\name};
  \draw[wire] (\xstart,\y) -- (\xend,\y);
}

\draw[dashed, rounded corners=2pt] (-0.20,0.38) rectangle (0.70,-2.33);
\node[note] at (0.25,0.64) {seed \(=|0000\rangle\)};

\node[rootgate] (u1) at (2.50,-0.975)
{\(U_1\)\\[0.25em]
\(\tau_{\mu_2=+1}^{(x_1,x_3)=(0,0)}(\theta_1)\)\\[0.25em]
\(0000\to0100\)};
\node[note] at (2.50,-2.72)
{\(\theta_1=\arcsin\sqrt{4/8}\)};

\node[rootgate] (u2) at (5.70,-0.975)
{\(U_2\)\\[0.25em]
\(\tau_{(\mu_2,\mu_3)=(-1,+1)}^{(x_1,x_4)=(0,0)}(\theta_2)\)\\[0.25em]
\(0100\to0010\)};
\node[note] at (5.70,-2.72)
{\(\theta_2=\arcsin\sqrt{2/4}\)};

\draw[dashed, rounded corners=2pt] (7.50,0.38) rectangle (11.10,-2.33);
\node[note] at (9.30,0.64) {parallel final update layer};

\node[childgate] (ul) at (9.30,-0.325)
{\(U_L\)\\[-0.1em]
\(\tau_{\mu_1=+1}^{x_2=0}(\theta_L)\), \(00\to10\)};

\node[childgate] (ur) at (9.30,-1.625)
{\(U_R\)\\[-0.1em]
\(\tau_{\mu_4=+1}^{x_3=0}(\theta_R)\), \(00\to01\)};

\node[note] at (9.30,-2.72)
{\(\theta_L=\theta_R=\arcsin\sqrt{1/2}\)};

\node[note, anchor=west] at (11.24,-0.975)
{output\\[-0.1em]
\(\frac{1}{\sqrt8}\bigl(|0000\rangle+|1000\rangle+|0100\rangle+|0010\rangle\)\\[-0.1em]
\(\hspace{1.35em}+|0001\rangle+|1010\rangle+|1001\rangle+|0101\rangle\bigr)\)};

\end{tikzpicture}
\caption{Controlled local-update circuit for the four-site one-dimensional Rydberg blockade example. The root updates assign amplitude to the representative basis states associated with interface labels \(10\) and \(01\) for the split \([1,4]=[1,2]\cup[3,4]\). The dashed final update layer expands the two length-two intervals in parallel. The output is the uniform superposition of the eight feasible basis states.}
\label{fig:sm-rydberg-controlled-circuit}
\end{figure}

For Fig.~\ref{fig:sm-rydberg-controlled-circuit}, the root split weights are
\begin{equation}
\begin{aligned}
M_{00}&=N([1,2];\star,0)N([3,4];0,\star)=4,\\
M_{10}&=N([1,2];\star,1)N([3,4];0,\star)=2,\\
M_{01}&=N([1,2];\star,0)N([3,4];1,\star)=2 .
\end{aligned}
\end{equation}
With label order \(00,10,01\), the sequential rule first separates weight \(4\) from total weight \(8\), then separates weight \(2\) from the remaining weight \(4\). This gives \(\theta_1=\arcsin\sqrt{4/8}\) and \(\theta_2=\arcsin\sqrt{2/4}\), both equal to \(\pi/4\).
The two nontrivial blocks in the final update layer each have two feasible strings after conditioning on the zero boundary value inherited from the split, so \(\theta_L=\theta_R=\arcsin\sqrt{1/2}\).

\subsection{\texorpdfstring{Bounded-width divergence constraints}{Bounded-width divergence constraints}}
\label{subsec:sm-divergence}

Let \(H=(V,E)\) be a directed graph with one binary variable \(x_e\in\{0,1\}\) on each edge, so \(n=|E|\). The hard constraints fix the net outgoing flow at every vertex,
\begin{equation}
\sum_{e\in\delta^+(v)}x_e-\sum_{e\in\delta^-(v)}x_e=b_v,
\qquad
v\in V .
\end{equation}
Here \(\delta^+(v)\) and \(\delta^-(v)\) denote outgoing and incoming edges. The \(b_v\)'s encode prescribed sources, sinks, or charges.

We consider graph families with maximum degree bounded by a constant \(\Delta\) and a balanced recursive edge decomposition whose split boundary contains only \(O(1)\) vertices. Their residual records and representatives, together with the trades and controls described below, are generated by a uniform polynomial-time rule. Neighboring interface labels are connected by divergence-preserving path or cycle trades on \(O(1)\) edges near the split. A block \(B\) has edge set \(E_B\), and each split produces child blocks whose edge sets are smaller by a fixed factor. We additionally require each trade with support \(K\) to have a set \(C\subseteq E_B\setminus K\) of \(O(1)\) block-local control edges and a pattern \(\eta\) such that its source and target are the only populated representatives on \(B\) matching \(\eta\) on \(C\). To make the edge recursion closed, each block records the net flow that it contributes at vertices shared with its complement. Let
\[
V_B=\{v\in V:(\delta^+(v)\cup\delta^-(v))\cap E_B\neq\varnothing\}
\]
be the vertices touched by the block edges, and let \(\partial B\subseteq V_B\) be the vertices incident to both \(E_B\) and \(E\setminus E_B\). A boundary condition is a vector \(\beta=(\beta_v)_{v\in\partial B}\), where \(\beta_v\) is the net flow contributed by the block edges at vertex \(v\). The vector \(\beta\) is a short integer list attached to the boundary vertices of the current block. The residual block \(F(B,\beta)\subseteq\{0,1\}^{E_B}\) consists of the edge assignments on \(E_B\) satisfying, for every \(v\in V_B\),
\begin{equation}
\sum_{e\in\delta^+(v)\cap E_B}x_e-\sum_{e\in\delta^-(v)\cap E_B}x_e
=
\begin{cases}
b_v, & v\in V_B\setminus\partial B,\\
\beta_v, & v\in\partial B.
\end{cases}
\end{equation}
The recursion terminates at constant-size blocks for which the inherited boundary conditions and vertex equations determine a unique edge assignment. Every nonempty terminal residual block is therefore a singleton, with single-edge blocks as a special case. Empty residual blocks are omitted. Since every boundary vertex has degree at most \(\Delta\) and \(x_e\in\{0,1\}\), each \(\beta_v\) lies in \(\{-\Delta,-\Delta+1,\ldots,\Delta\}\). Since \(|\partial B|=O(1)\), only \(O(1)\) boundary conditions are admissible for a fixed block.

When \(B\) is split into \(B_L\) and \(B_R\), each edge assignment in \(F(B,\beta)\) induces a unique compatible pair of child boundary conditions. This gives the disjoint decomposition
\begin{equation}
F(B,\beta)
=
\bigsqcup_{\alpha\in\mathcal A(B,\beta)}
F(B_L,\beta_L^{(\alpha)})\times F(B_R,\beta_R^{(\alpha)}).
\end{equation}
Here \(\alpha=(\beta_L^{(\alpha)},\beta_R^{(\alpha)})\) is the interface label and ranges over the compatible child conditions. The union is disjoint because each edge assignment restricts uniquely to the two children and hence fixes its child boundary conditions. For each block \(B\), \(\mathcal R(E_B)\) contains \(O(1)\) admissible boundary conditions. Each nonterminal condition has \(O(1)\) interface labels. Thus \(\Lambda(E_B)=O(1)\) and the recursive expansion dimension is \(d=0\).

For every nonempty residual block \(F(B,\beta)\), the rule returns one representative flow \(c_{B,\beta}\in F(B,\beta)\) recursively. At a nonterminal block, it is the concatenation \(c_{B_L,\beta_L^{(\alpha_0)}}\sqcup c_{B_R,\beta_R^{(\alpha_0)}}\) for one admissible interface label \(\alpha_0\). The representative associated with any label \(\alpha\) is defined by the same concatenation with \(\alpha\) in place of \(\alpha_0\). For the uniform feasible superposition, the split weight routed through \(\alpha\) is \(M_\alpha=|F(B_L,\beta_L^{(\alpha)})|\,|F(B_R,\beta_R^{(\alpha)})|\). Nonuniform target distributions are covered when their active split probabilities and local update angles are efficiently computable and use a common angle wherever the same update appears in more than one residual block.

A local trade is specified by two edge sets \(P^+\) and \(P^-\), with common edges canceled. Its edge-indexed signed update vector \(\mu\) satisfies \(\mu_e=1\) for \(e\in P^+\), \(\mu_e=-1\) for \(e\in P^-\) and \(\mu_e=0\) otherwise.
The update is applied only when \(x_e=0\) for all \(e\in P^+\) and \(x_e=1\) for all \(e\in P^-\), and it sends these occupations to \(1\) and \(0\), respectively. The path or cycle property means that \(\mu\) has zero divergence at every internal vertex and leaves the parent boundary condition unchanged. The controlled update remains inside the same residual block and keeps all variables binary.

\begin{lemma}[Block-local controls for divergence updates]
\label{lem:sm-divergence-block-local-controls}
For the graph families described above, each update between neighboring interface labels has \(O(1)\) support and \(O(1)\) controls inside the block. These controls leave every occupied representative for the other residual boundary conditions in \(\mathcal R(E_B)\) unchanged.
\end{lemma}

\begin{proof}
Let \(K=P^+\cup P^-\) be the support of the trade, so \(|K|=O(1)\). The signed toggle preserves divergence at every internal vertex and leaves the parent boundary condition unchanged. By the isolation condition above, there is a set \(C\subseteq E_B\setminus K\) of \(O(1)\) control edges and a pattern \(\eta\) whose only populated matches are the source and target representatives. Algorithm~\ref{alg:controlled-local-update} with controls \((C,\eta)\) realizes the required two-level update on this pair and acts as the identity on every other occupied representative. All support and control edges lie inside \(E_B\). Each update uses \(O(1)\) support and control edges without affecting another residual boundary condition in \(\mathcal R(E_B)\).
\end{proof}

\begin{corollary}[Bounded-width divergence resources]
For the graph families described above, the construction for divergence constraints gives an ancilla-free preparation with depth \(O(\log n)\) and gate count \(O(n)\).
\end{corollary}

\begin{proof}
All admissible boundary conditions in \(\mathcal R(E_B)\) require \(O(1)\) updates in total. Lemma~\ref{lem:sm-divergence-block-local-controls} implements each update with \(O(1)\) depth and gates. Blocks in the same layer have disjoint edge sets and can be scheduled in parallel, so one layer has constant depth. A balanced decomposition with constant-size leaves has \(O(\log n)\) layers and \(O(n)\) nonterminal blocks. The depth is \(O(\log n)\), and the total gate count is \(O(n)\).
\end{proof}

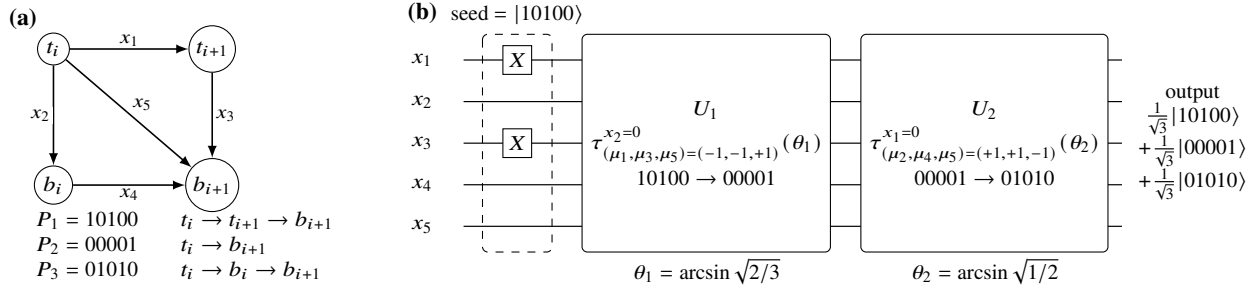
\begin{figure*}[t]
\centering
\begin{minipage}[c]{0.26\textwidth}
\centering
\begin{tikzpicture}[
>=latex,
flowedge/.style={->, line width=0.6pt},
vtx/.style={circle, draw, inner sep=1.5pt, font=\footnotesize},
edge lab/.style={font=\scriptsize, fill=white, inner sep=1pt},
note/.style={font=\footnotesize, align=left}
]
\node[font=\bfseries, anchor=west] at (-0.70,2.55) {(a)};
\node[vtx] (ti) at (0,2.15) {$t_i$};
\node[vtx] (tip) at (2.1,2.15) {$t_{i+1}$};
\node[vtx] (bi) at (0,0.35) {$b_i$};
\node[vtx] (bip) at (2.1,0.35) {$b_{i+1}$};
\draw[flowedge] (ti) -- node[edge lab, above] {$x_1$} (tip);
\draw[flowedge] (ti) -- node[edge lab, left] {$x_2$} (bi);
\draw[flowedge] (tip) -- node[edge lab, right] {$x_3$} (bip);
\draw[flowedge] (bi) -- node[edge lab, below] {$x_4$} (bip);
\draw[flowedge] (ti) -- node[edge lab, above right] {$x_5$} (bip);
\node[note, anchor=west] at (-0.35,-0.45)
{$P_1=10100\qquad t_i\to t_{i+1}\to b_{i+1}$\\
 $P_2=00001\qquad t_i\to b_{i+1}$\\
 $P_3=01010\qquad t_i\to b_i\to b_{i+1}$};
\end{tikzpicture}
\end{minipage}
\begin{minipage}[c]{0.7\textwidth}
\centering
\begin{tikzpicture}[
wire/.style={line width=0.45pt},
xgate/.style={draw, fill=white, minimum size=0.38cm, inner sep=0pt, font=\footnotesize},
rootgate/.style={draw, rounded corners=2pt, fill=white, minimum width=3.15cm, minimum height=2.88cm, align=center, font=\footnotesize},
lab/.style={font=\footnotesize, anchor=east},
note/.style={font=\footnotesize, align=center}
]
\def\xstart{-0.45}
\def\xend{8.26}
\node[note, anchor=west, font=\bfseries] at (-1.3,0.62) {(b)};

\foreach \y/\name in {0/\(x_1\),-0.55/\(x_2\),-1.10/\(x_3\),-1.65/\(x_4\),-2.20/\(x_5\)} {
  \node[lab] at (\xstart-0.3,\y) {\name};
  \draw[wire] (\xstart,\y) -- (\xend,\y);
}

\draw[dashed, rounded corners=2pt] (-0.20,0.34) rectangle (0.72,-2.54);
\node[note] at (0.26,0.60) {seed \(=|10100\rangle\)};
\node[xgate] at (0.26,0) {\(X\)};
\node[xgate] at (0.26,-1.10) {\(X\)};

\node[rootgate] (u1) at (2.76,-1.10)
{\(U_1\)\\[0.25em]
\(\tau_{(\mu_1,\mu_3,\mu_5)=(-1,-1,+1)}^{x_2=0}(\theta_1)\)\\[0.25em]
\(10100\to00001\)};
\node[note] at (2.76,-2.80)
{\(\theta_1=\arcsin\sqrt{2/3}\)};

\node[rootgate] (u2) at (6.44,-1.10)
{\(U_2\)\\[0.25em]
\(\tau_{(\mu_2,\mu_4,\mu_5)=(+1,+1,-1)}^{x_1=0}(\theta_2)\)\\[0.25em]
\(00001\to01010\)};
\node[note] at (6.44,-2.80)
{\(\theta_2=\arcsin\sqrt{1/2}\)};

\node[note, anchor=west] at (8.36,-1.10)
{output\\[-0.1em]
\(\frac{1}{\sqrt3}|10100\rangle\)\\[-0.1em]
\(+\frac{1}{\sqrt3}|00001\rangle\)\\[-0.1em]
\(+\frac{1}{\sqrt3}|01010\rangle\)};

\end{tikzpicture}
\end{minipage}
\caption{Local trade example for bounded-width divergence. (a) Three unit flow patterns \(P_1,P_2,P_3\) on a five-edge cell of the directed ladder. (b) Controlled local-update circuit for the same cell. The seed block prepares the basis state representing \(P_1\). The two updates realize the preparation steps \(P_1\to P_2\) and \(P_2\to P_3\), with angles computed in the text. The output is the uniform superposition of the three basis states representing these flow patterns.}
\label{fig:sm-divergence-controlled-circuit}
\end{figure*}

For the resource experiment, we use a ladder flow with two vertices \(t_i,b_i\) in each column \(i=0,\ldots,L\). For every cell \(i=0,\ldots,L-1\), include the directed edges
\[
t_i\to b_i,\qquad
t_i\to t_{i+1},\qquad
b_i\to b_{i+1},\qquad
t_i\to b_{i+1},
\]
and add the final vertical edge \(t_L\to b_L\). The number of binary edge variables is \(n=4L+1\). The prescribed divergence is \(1\) at \(t_i\) for \(i=0,\ldots,L-1\), \(0\) at \(t_L\) and \(b_0\), and \(-1\) at \(b_i\) for \(i=1,\ldots,L\).
Let \(x_i^{(v)},x_i^{(t)},x_i^{(b)},x_i^{(d)}\) label the vertical, top, bottom and diagonal edges in cell \(i\), with \(x_L^{(v)}\) on the final vertical edge. The top-row constraints are \(x_i^{(v)}+x_i^{(t)}+x_i^{(d)}-x_{i-1}^{(t)}=1\) for \(i=0,\ldots,L-1\), with \(x_{-1}^{(t)}=0\). The bottom-row constraints are \(x_0^{(b)}-x_0^{(v)}=0\) and \(x_i^{(b)}-x_{i-1}^{(b)}-x_{i-1}^{(d)}-x_i^{(v)}=-1\) for \(i=1,\ldots,L-1\). The terminal constraints are \(x_L^{(v)}=x_{L-1}^{(t)}\) and \(x_{L-1}^{(b)}+x_{L-1}^{(d)}+x_L^{(v)}=1\).
The cut after each cell is described by the three-bit assignment to the top, bottom, and diagonal edges crossing into the next column. The following result verifies the recursive construction used for this benchmark. This ladder flow is the divergence benchmark behind Fig.~\ref{fig:ncs-scope-scaling}.

\begin{lemma}[Recursive construction for the ladder flow]
\label{lem:sm-ladder-flow-construction}
For every \(L\geq 1\), the ladder flow admits a balanced recursion over contiguous ladder steps. Each split has at most three interface labels. Its canonical representatives can be connected by local updates with support at most four and with at most four block-local controls. The updates at each recursion level have a constant-color schedule. Consequently, this benchmark has recursive expansion dimension \(d=0\) and admits ancilla-free preparation with depth \(O(\log n)\) and gate count \(O(n)\).
\end{lemma}

\begin{proof}
Write the word of cell \(i\) in the order \(v_i u_i \ell_i d_i\), where the entries denote the vertical, top, bottom, and diagonal edge occupations. The cut state is \((u_i,\ell_i,d_i)\). Let \(A=(1,0,0)\), \(B=(0,1,0)\), and \(C=(0,0,1)\). The reachable transitions at every interior cell are
\begin{equation}
\begin{array}{c|c}
\text{incoming state} & \text{outgoing state and cell word} \\ \hline
A & (A,1100),\ (C,1001) \\
B\text{ or }C & (A,0100),\ (B,1010),\ (C,0001).
\end{array}
\label{eq:ladder-flow-transitions}
\end{equation}
The root state \(0=(0,0,0)\) has the same three outgoing words as \(B\) and \(C\). The final vertical edge accepts \(A\) with word \(1\) and accepts \(B\) or \(C\) with word \(0\). It follows by induction that only \(A\), \(B\), and \(C\) occur at internal cuts. States \(B\) and \(C\) have identical continuation languages, so the right child uses \(C\) as the canonical entry for either state without changing its feasible words.

Split each interval at its midpoint and stop at one ladder step. An ordinary step contains four edges and the final step contains one. The larger child therefore contains at most \(4/5\) of the parent edge variables. The recursion has \(O(\log n)\) levels, and the boundary states uniquely determine the word at every one-step leaf. A residual interval has at most \(4\times4=16\) boundary-state pairs. Its split state belongs to \(\{A,B,C\}\), so every residual block has at most three interface labels.

Choose the canonical representative of an interval by keeping every intermediate cut in state \(C\). Equation~\eqref{eq:ladder-flow-transitions} shows that this choice reaches every feasible right boundary from every feasible left boundary. Changing the split state changes only the final cell of the left child and the first step of the right child. The three words leaving \(C\) have pairwise Hamming distance at most three. Only a branch with split state \(A\) changes the canonical entry of the right child, and this changes one further bit. Two branch representatives therefore differ on at most four edge variables. These variables form the support of the local update.

At an internal left endpoint, the vertical and lower horizontal edge occupations of the first cell distinguish the canonical incoming states \(A\) and \(C\). At an internal right endpoint, the upper and lower horizontal edge occupations of the final cell distinguish \(A\), \(B\), and \(C\). Root and terminal endpoints are fixed. These endpoint occupations, excluding any already contained in the update support, give at most four block-local controls. Together with the split-state pattern on the support, they uniquely identify the source and target among all occupied representatives for the same interval.

There are at most sixteen residual boundary pairs and at most two ordered updates for each pair. Hence \(\Lambda(I)=O(1)\) uniformly over the intervals and \(d=0\). Color an update by its position in the ordered interface list, the parity of its interval rank, and its two endpoint states. This uses at most \(2\times2\times16\) colors. Updates with the same color act on disjoint intervals. Each recursion level therefore has constant depth. The recursion contains \(O(n)\) intervals, which gives \(O(n)\) gates and \(O(\log n)\) depth. The same constant-state recursion computes all split masses from the leaves upward. The finite transition table generates the records, representatives and branch labels in polynomial time, while the controls isolate every update. Thus the ladder recursion is constructible.
\end{proof}

The quotient \(B\sim C\) also supports nonuniform target distributions when the two contexts have the same normalized continuation probabilities throughout the right child. The uniform target used here satisfies this condition.

A single ladder cell illustrates the local updates used by this benchmark. Let \(x_1,\ldots,x_5\) label the top, left vertical, right vertical, bottom, and diagonal edges, respectively. Within this cell, each pattern contributes divergence \(1\) at \(t_i\), \(-1\) at \(b_{i+1}\), and \(0\) at \(t_{i+1}\) and \(b_i\). In the order \((x_1,\ldots,x_5)\), the three unit flow patterns are
\[
P_1=10100,\qquad
P_2=00001,\qquad
P_3=01010 .
\]
They correspond to the paths \(t_i\to t_{i+1}\to b_{i+1}\), \(t_i\to b_{i+1}\), and \(t_i\to b_i\to b_{i+1}\). Their common divergence ensures that the local trades preserve the global ladder constraints. Choose order \(P_1,P_2,P_3\). The first preparation step separates \(P_1\) from \(\{P_2,P_3\}\), so the root weights are \(1\) and \(2\), giving \(\theta_1=\arcsin\sqrt{2/3}\). The second preparation step separates \(P_2\) from \(P_3\), so the child weights are \(1\) and \(1\), giving \(\theta_2=\arcsin\sqrt{1/2}\). The two local trades in Fig.~\ref{fig:sm-divergence-controlled-circuit}(b) realize these preparation steps as \(10100\to00001\) on \(\{x_1,x_3,x_5\}\), with control \(x_2=0\), and \(00001\to01010\) on \(\{x_2,x_4,x_5\}\), with control \(x_1=0\).

\section{Broader structured state classes}
\label{sec:sm-generality}

Here the recursion comes from a compact state description rather than an explicit hard constraint.
Sparse support states use an ordered support list. Decision diagram states use reverse compression of incoming edge representatives. Hierarchical low-rank states use bounded-rank Schmidt trees with computational-basis sectors that are disjoint at each split. In each case, the compact description determines the representatives, local conditional probabilities and phases. Representative pairs determine the local updates, while the probabilities and phases determine their parameters.

\subsection{Sparse support states}
\label{subsec:sm-sparse-support}

A sparse-support state is specified by \(S\) bit strings and their amplitudes. Write the support as \(\Omega\subseteq\{0,1\}^n\), with \(|\Omega|=S\), and target state \(|\psi\rangle=\sum_{x\in\Omega}a_x|x\rangle\), normalized by \(\sum_{x\in\Omega}|a_x|^2=1\).
Choose an order \(\Omega=\{x^{(1)},x^{(2)},\ldots,x^{(S)}\}\), use \(x^{(1)}\) as the seed, and define the remaining support \(\Omega_j=\{x^{(j)},x^{(j+1)},\ldots,x^{(S)}\}\) with remaining mass \(R_j=\sum_{t=j}^S |a_{x^{(t)}}|^2\). Then \(R_1=1\), and the preparation proceeds through
\(\Omega_j=\{x^{(j)}\}\sqcup \Omega_{j+1}\) for \(j=1,\ldots,S-1\).
At step \(j\), the remaining support \(\Omega_j\) is the recursion block, its representative is \(x^{(j)}\), and the update opens \(x^{(j+1)}\). The update keeps probability \(|a_{x^{(j)}}|^2/R_j\) on \(x^{(j)}\) and transfers \(R_{j+1}/R_j\), so
\(\theta_j=\arcsin\sqrt{R_{j+1}/R_j}\).
In the worst case \(x^{(j)}\) and \(x^{(j+1)}\) differ on \(O(n)\) qubits, so the controlled update from Sec.~\ref{subsec:sm-controlled-updates} admits an ancilla-free implementation with depth and gate count \(O(n)\).
The coefficient phases can be absorbed into the corresponding updates or added afterward by controlled phase updates on the corresponding basis states in the support. Under the same implementation, this adds at most \(O(Sn)\) depth and gates and does not change the asymptotic bound below.

As an example with four strings, take
\[
|\psi_\Omega\rangle
=
\sqrt{\tfrac12}|0000\rangle
+\sqrt{\tfrac14}|0110\rangle
+\sqrt{\tfrac18}|1001\rangle
+\sqrt{\tfrac18}|1110\rangle .
\]
In the chain order \((0000,0110,1001,1110)\), the residual masses are \(1,1/2,1/4,1/8\), giving three angles \(\arcsin\sqrt{1/2}\). An equivalent grouped schedule is shown in Fig.~\ref{fig:sm-sparse-support-controlled-circuit}. It first separates the total mass \(1/4\) on \(\{1001,1110\}\), then refines the two pairs with conditional masses \(1/3\) and \(1/2\).

\begin{figure}[t]
\centering
\begin{tikzpicture}[
wire/.style={line width=0.45pt},
rootgate/.style={draw, rounded corners=2pt, fill=white, minimum width=2.80cm, minimum height=2.40cm, align=center, font=\footnotesize},
smallgate/.style={draw, rounded corners=2pt, fill=white, minimum width=2.80cm, minimum height=2.40cm, align=center, font=\footnotesize},
lab/.style={font=\footnotesize, anchor=east},
note/.style={font=\footnotesize, align=center}
]
\def\xstart{-0.45}
\def\xend{10.90}

\foreach \y/\name in {0/\(x_1\),-0.58/\(x_2\),-1.16/\(x_3\),-1.74/\(x_4\)} {
  \node[lab] at (\xstart-0.25,\y) {\name};
  \draw[wire] (\xstart,\y) -- (\xend,\y);
}

\draw[dashed, rounded corners=2pt] (-0.20,0.33) rectangle (0.60,-2.07);
\node[note] at (0.20,0.58) {seed \(=|0000\rangle\)};

\node[rootgate] (ur) at (2.42,-0.87)
{\(U_{\rm root}\)\\[0.25em]
\(\tau_{(\mu_1,\mu_4)=(+1,+1)}(\theta_{\rm root})\)\\[0.25em]
\(0000\to1001\)};
\node[note] at (2.42,-2.40)
{\(\theta_{\rm root}=\arcsin\sqrt{1/4}\)};

\node[smallgate] (u0) at (5.64,-0.87)
{\(U_0\)\\[0.25em]
\(\tau_{(\mu_2,\mu_3)=(+1,+1)}^{(x_1,x_4)=(0,0)}(\theta_0)\)\\[0.25em]
\(0000\to0110\)};
\node[note] at (5.64,-2.40)
{\(\theta_0=\arcsin\sqrt{1/3}\)};

\node[smallgate] (u1) at (9.08,-0.87)
{\(U_1\)\\[0.25em]
\(\tau_{(\mu_2,\mu_3,\mu_4)=(+1,+1,-1)}^{x_1=1}(\theta_1)\)\\[0.25em]
\(1001\to1110\)};
\node[note] at (9.08,-2.40)
{\(\theta_1=\arcsin\sqrt{1/2}\)};

\node[note, anchor=west] at (10.88,-0.87)
{output\\[-0.1em]
\(\sqrt{\tfrac12}|0000\rangle+\sqrt{\tfrac14}|0110\rangle\)\\[-0.1em]
\(\hspace{0.45em}+\sqrt{\tfrac18}|1001\rangle+\sqrt{\tfrac18}|1110\rangle\)};

\end{tikzpicture}
\caption{Controlled local-update circuit for a four-qubit sparse-support state. Starting from \(|0000\rangle\), \(U_{\rm root}\) opens \(|1001\rangle\), and \(U_0\) and \(U_1\) refine the two branches. The output is \(\sqrt{1/2}|0000\rangle+\sqrt{1/4}|0110\rangle+\sqrt{1/8}|1001\rangle+\sqrt{1/8}|1110\rangle\). The displayed angles follow from the branch probabilities.}
\label{fig:sm-sparse-support-controlled-circuit}
\end{figure}
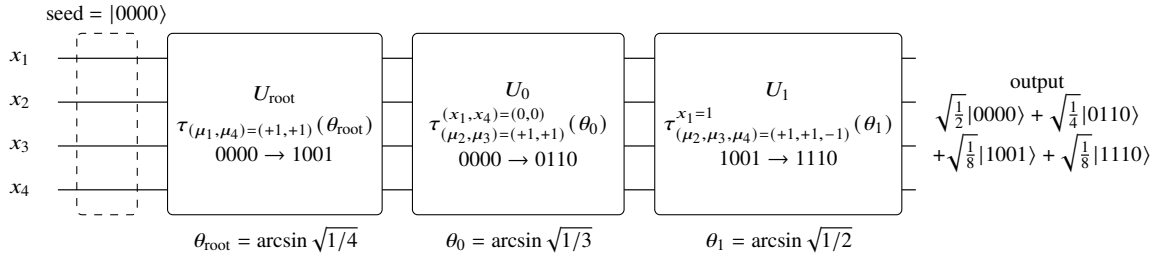

\begin{corollary}[Sparse support resources]
A support tree with \(S\) leaves gives an ancilla-free preparation with depth \(O(Sn)\) and gate count \(O(Sn)\).
\end{corollary}

\begin{proof}
In the chain schedule, the \(S-1\) nonseed bit strings are opened sequentially. The equivalent star schedule used for Fig.~\ref{fig:ncs-scope-scaling} opens each nonseed bit string directly from the seed.
Each update may change \(O(n)\) bits. It has an ancilla-free implementation with depth and gate count \(O(n)\) by Sec.~\ref{subsec:sm-controlled-updates}. Multiplying by the \(S-1\) updates gives depth \(O(Sn)\) and gate count \(O(Sn)\).
\end{proof}

\subsection{Decision diagram states}
\label{subsec:sm-decision-diagram}

A rooted decision diagram is a layered DAG whose root-to-terminal paths specify bit strings.
We use a deterministic levelized binary decision DAG \(Q=(V,E)\). It has one root in layer \(0\) and one common terminal \(v_{\rm term}\) in layer \(n\). Every edge joins adjacent layers and carries a bit label \(\sigma\in\{0,1\}\). Each node has at most one outgoing edge with each bit label. Zero-weight edges and non-live nodes are removed. For an edge \(e=(v,\sigma,u)\), let \(w_e\in\mathbb C\setminus\{0\}\) be its weight. The product of edge weights along a root-to-terminal path gives the unnormalized basis amplitude.

For a node \(v\), let \(\Omega_v\) be the suffix strings accepted below \(v\). Writing each outgoing edge as \(e=(v,\sigma,u)\), the suffix support decomposes as
\begin{equation}
\Omega_v
=
\bigsqcup_{e=(v,\sigma,u)}
\{\,\sigma y:\ y\in\Omega_u\,\}.
\end{equation}
This is a disjoint union because each edge consumes exactly the next bit and determinism gives at most one outgoing edge with any fixed bit from the same node. Removing non-live nodes ensures that every displayed component is nonempty.

Let \(\deg^-(u)\) denote the number of incoming edges at \(u\). The classical plan is most naturally read as an inverse compression circuit. It processes the levelized diagram from the root toward the terminal. After each source prefix has been compressed to a representative prefix \(p_v\), an incoming edge \(e=(v,\sigma,u)\) gives the representative prefix \(p_v\sigma\) for node \(u\). The inverse circuit compresses these incoming representatives into one selected survivor using \(\deg^-(u)-1\) Givens rotations. The preparation circuit reverses this schedule. It starts from the terminal representative and applies the stored rotations from the terminal layer back to the root.

To set these rotations without enumerating accepted strings, we compute the probability carried by each incoming representative from a bottom-up suffix mass and a top-down arrival mass. Let \(\mathcal P(v)\) be the set of paths from \(v\) to the terminal and let \(W(p)\) be the product of edge weights on \(p\). Define \(Z_v=\sum_{p\in\mathcal P(v)}|W(p)|^2\). If \(w_e\) is the edge weight on \(e=(v,\sigma,u)\), these suffix masses satisfy the bottom-up recursion
\begin{equation}
Z_v
=
\sum_{e=(v,\sigma,u)}
|w_e|^2 Z_u ,
\qquad
Z_{v_{\rm term}}=1.
\end{equation}
The normalized target state is \(|\psi_Q\rangle=Z_{\rm root}^{-1/2}\sum_{p\in\mathcal P({\rm root})}W(p)|x(p)\rangle\), where \(x(p)\in\{0,1\}^n\) is the bit string read along \(p\). The top-down probability mass reaching a node in this state is denoted by \(A_v\), with \(A_{\rm root}=1\). For an incoming edge \(e=(v,\sigma,u)\), the edge contribution is \(N_e=A_v |w_e|^2 Z_u/Z_v\), and the arrival mass at \(u\) is \(A_u=\sum_{e=(v,\sigma,u)}N_e\). The quantity \(A_v\) is a normalized node probability. In the unweighted equal-amplitude case, \(Z_v\) counts suffix paths and these equations reduce to normalized path counting.

Complex edge phases are propagated by a survivor phase. Set \(\phi_{\rm root}=1\). For an incoming edge \(e=(v,\sigma,u)\), define \(\zeta_e=\phi_v w_e/|w_e|\). If \(e_0\) is the selected survivor at \(u\), set \(\phi_u=\zeta_{e_0}\) and use the normalized local coefficients
\begin{equation}
{
\gamma_e
=
\sqrt{\frac{N_e}{A_u}}\,
\frac{\zeta_e}{\phi_u}.
}
\end{equation}
The survivor coefficient is then positive real, while the remaining \(\gamma_e\) carry the relative phases needed by the Givens rotations.

At node \(u\), each reverse local update acts on two incoming prefixes \(p_v\sigma\) and \(p_{v'}\sigma'\) and is tensored with the identity on the suffix below \(u\). These prefixes may differ on \(O(n)\) qubits, so each Givens rotation has an ancilla-free implementation with \(O(n)\) depth and gates by Sec.~\ref{subsec:sm-controlled-updates}. The sharing is accounted for because all incoming edges to the same diagram node are compressed together.

For the binary divisibility example below, the \(n=3\) target is
\[
{
|\psi_3\rangle
=
\frac{1}{\sqrt3}
\left(|000\rangle+|011\rangle+|110\rangle\right) .
}
\]
We number the modules in preparation order. Starting from \(000\), \(U_1\) opens \(011\) with \(\theta_1=\arcsin\sqrt{1/3}\) under the zero control on \(x_1\). Then \(U_2\) opens \(110\) from the remaining \(000\) amplitude with \(\theta_2=\pi/4\). It acts on \(x_1,x_2\) and as the identity on \(x_3\), so \(011\) is unchanged. The state is \(|\psi_3\rangle\), as shown in Fig.~\ref{fig:sm-decision-diagram-controlled-circuit}.

\begin{figure}[t]
\centering
\begin{tikzpicture}[
wire/.style={line width=0.45pt},
fullgate/.style={draw, rounded corners=2pt, fill=white, minimum width=2.92cm, minimum height=1.76cm, align=center, font=\footnotesize},
twobitgate/.style={draw, rounded corners=2pt, fill=white, minimum width=2.92cm, minimum height=1.04cm, align=center, font=\footnotesize},
lab/.style={font=\footnotesize, anchor=east},
note/.style={font=\footnotesize, align=center}
]
\def\xstart{-0.45}
\def\xend{7.57}

\foreach \y/\name in {0/\(x_1\),-0.62/\(x_2\),-1.24/\(x_3\)} {
  \node[lab] at (\xstart-0.25,\y) {\name};
  \draw[wire] (\xstart,\y) -- (\xend,\y);
}

\draw[dashed, rounded corners=2pt] (-0.20,0.34) rectangle (0.72,-1.58);
\node[note] at (0.26,0.59) {seed \(=|000\rangle\)};

\node[fullgate] at (2.58,-0.62)
{\(U_1\)\\[0.22em]
\(\tau_{(\mu_2,\mu_3)=(+1,+1)}^{x_1=0}(\theta_1)\)\\[0.22em]
\(000\to011\)};
\node[note] at (2.58,-1.91)
{\(\theta_1=\arcsin\sqrt{1/3}\)};

\node[twobitgate] at (5.90,-0.31)
{\(U_2\)\\[0.22em]
\(\tau_{(\mu_1,\mu_2)=(+1,+1)}(\theta_2)\)\\[0.22em]
\(00\to11\)};
\node[note] at (5.90,-1.91)
{\(\theta_2=\pi/4\)};

\node[note, anchor=west] at (7.61,-0.62)
{output\\[-0.1em]
\(\frac{1}{\sqrt3}\bigl(|000\rangle+|011\rangle+|110\rangle\bigr)\)};

\end{tikzpicture}
\caption{Controlled local-update circuit for the three-bit binary divisibility decision diagram. Starting from \(|000\rangle\), \(U_1\) is controlled by \(x_1=0\) and updates \(x_2,x_3\). The module \(U_2\) then updates \(x_1,x_2\) while leaving \(x_3\) unchanged. Together, they prepare the uniform state \((|000\rangle+|011\rangle+|110\rangle)/\sqrt3\). Both angles follow from the reverse compression plan.}
\label{fig:sm-decision-diagram-controlled-circuit}
\end{figure}
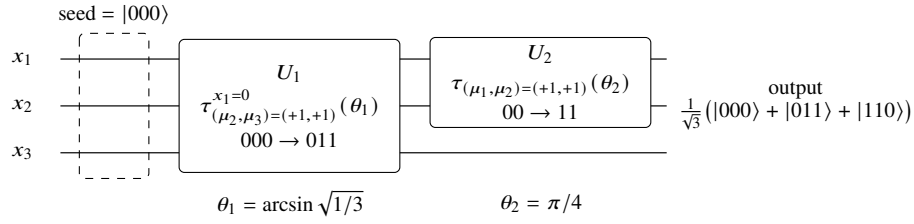

\begin{corollary}[Decision diagram resources]
For a compressed decision diagram \(Q=(V,E)\), the construction gives an ancilla-free preparation with depth \(O(|E|n)\) and gate count \(O(|E|n)\).
\end{corollary}

\begin{proof}
The count is over compression operations on incoming edges, not over root-to-terminal paths. At each diagram node \(w\), choose one incoming representative as the survivor and compress the other incoming representatives into it. The number of Givens rotations is at most
\begin{equation}
\sum_{w\ne {\rm root}}
\max\{\deg^-(w)-1,0\}
\le
|E| .
\end{equation}
Each rotation costs \(O(n)\) depth and gates in the worst case under the ancilla-free implementation of Sec.~\ref{subsec:sm-controlled-updates}. A fully sequential schedule gives \(O(|E|n)\) for both depth and gate count. Reversing the compression circuit gives the preparation circuit with the same asymptotic resources.
\end{proof}

The resource benchmark uses an ordered binary decision diagram of width three for binary divisibility by three. For \(x=x_1\cdots x_n\), let \(\operatorname{val}_2(x)=\sum_{j=1}^n2^{n-j}x_j\) and define
\begin{equation}
{
\Omega_n
=
\left\{
x\in\{0,1\}^n
\middle|
\operatorname{val}_2(x)\equiv0\pmod 3
\right\}.
}
\end{equation}
The support contains \(|\Omega_n|=\lceil2^n/3\rceil\) strings. Its uniform state is
\begin{equation}
{
|\psi_n\rangle
=
\frac{1}{\sqrt{|\Omega_n|}}
\sum_{x\in\Omega_n}|x\rangle .
}
\end{equation}

The diagram reads the bits from most to least significant. After \(j\) bits, a node records the prefix remainder \(r_j\in\{0,1,2\}\). An edge labeled \(b\) updates it as \(r_{j+1}=(2r_j+b)\bmod3\). The root has remainder zero, and the only accepting terminal has remainder zero. Removing unreachable and non-live nodes gives layer widths \(1,2,3,\ldots,3,2,1\). Hence, for \(n\ge4\), the numbers of live nodes and edges are
\begin{equation}
{
V(n)=1+2+3(n-3)+2+1=3n-3,
\qquad
E(n)=2+4+6(n-4)+4+2=6n-12.
}
\end{equation}

For a live remainder node in layer \(j\), choose the representative prefix
\begin{equation}
{
p_{j,0}=0^j,
\qquad
p_{j,1}=0^{j-1}1,
\qquad
p_{j,2}=0^{j-2}10 .
}
\end{equation}
The last representative is used only for \(j\ge2\). Let \(C_{j,t}\) be the number of length-\(j\) prefixes with remainder \(t\). These counts satisfy \(C_{0,0}=1\), \(C_{0,1}=C_{0,2}=0\) and
\begin{equation}
{
C_{j,t}
=
\sum_{\substack{r\in\{0,1,2\},\ b\in\{0,1\}\\
2r+b\equiv t\ ({\rm mod}\ 3)}}
C_{j-1,r} .
}
\end{equation}
Suppose a live node with remainder \(t\) in layer \(j\) has two incoming representatives. Let \(s\) be the parent remainder of the selected survivor and let \(a\) be the other parent remainder. Reversing the compression rotation opens the alternate representative with angle
\begin{equation}
{
\theta_{j,t}
=
\arcsin
\sqrt{
\frac{C_{j-1,a}}
{C_{j-1,s}+C_{j-1,a}}
}.
}
\end{equation}

For \(3\le j\le n-2\), all three remainder nodes and all six incoming edges are live. Their incoming representative pairs are
\begin{equation}
{
\begin{aligned}
0^{j-3}000&\longleftrightarrow0^{j-3}011,\\
0^{j-3}001&\longleftrightarrow0^{j-3}100,\\
0^{j-3}010&\longleftrightarrow0^{j-3}101.
\end{aligned}
}
\end{equation}
The first pair is a \(00\leftrightarrow11\) update controlled on the preceding zero. The other two pairs are three-qubit local updates. Layer \(j=2\) uses only \(00\leftrightarrow11\). The penultimate and terminal layers use the first two three-bit pairs and the first pair, respectively.

The updates do not interfere. Before layer \(j\) is compressed, every occupied source node in layer \(j-1\) has its representative prefix. Its incoming representatives at layer \(j\) have the patterns listed above, with every earlier bit equal to zero. The three pairs are disjoint. The zero control on the first pair prevents it from acting on the second pair, while the other updates act on all three displayed bits. Each update is the identity on the suffix below the common target node. It combines the two amplitudes entering that node without changing another occupied branch. The boundary layers use the subsets specified above. The angle changes their coefficients from \(\sqrt{C_{j-1,s}}\) and \(\sqrt{C_{j-1,a}}\) to \(\sqrt{C_{j,t}}\) on the survivor. Induction over the layers compresses \(|\psi_n\rangle\) to \(|0^n\rangle\). Reversing the circuit prepares the uniform state over \(\Omega_n\).

For \(n\ge4\), every nonroot live node has at least one incoming edge. The exact number of local updates is
\begin{equation}
{
U(n)
=
\sum_{w\ne{\rm root}}\bigl(\deg^-(w)-1\bigr)
=
E(n)-V(n)+1
=
3n-8 .
}
\end{equation}
Each update acts on at most three consecutive qubits and has constant depth and gate count under the controlled local-update implementation. At most three updates occur in one diagram layer. The ancilla-free preparation has depth \(O(n)\) and gate count \(O(n)\).

Sparse support states also fit the general decision diagram model above.
A state with \(S\) nonzero basis amplitudes has a trie with \(S\) root-to-terminal paths and \(O(Sn)\) edges, while shared downstream structure can reduce the diagram further. The direct sparse support construction above gives the sharper \(O(Sn)\) bound when one only uses the support list. The decision diagram bound is not intended to dominate the support list bound for an uncompressed trie. It is useful when repeated downstream structure makes \(|E|\ll Sn\).

\subsection{Hierarchical low-rank states}
\label{subsec:sm-low-rank}

A hierarchical low-rank state is specified by recursive Schmidt decompositions on a balanced binary tree. Related low-rank methods~\cite{araujo2024lowrank} and matrix-product-state methods~\cite{malz2024mps} exploit bounded ranks across bipartitions. This description includes low-entanglement many-body states such as GHZ-type families~\cite{jin2025ghz,cao2024ghz}. At each block, the split is written in Schmidt form
\begin{equation}
|\psi\rangle
=
\sum_{s=1}^{r}
\lambda_s\,
|\psi_s^{(L)}\rangle
|\psi_s^{(R)}\rangle .
\end{equation}
Here \(s=1,\ldots,r\) labels the Schmidt sectors, \(r\) is the Schmidt rank and \(\lambda_s\) are complex sector amplitudes whose magnitudes are the Schmidt coefficients. The child states \(|\psi_s^{(L)}\rangle\) and \(|\psi_s^{(R)}\rangle\) live on the left and right child blocks. Orthonormality on each side makes \(|\lambda_s|^2\) the split weights. We consider efficiently specified balanced trees with every bond rank bounded by \(r=O(1)\). Their sector coefficients are efficiently computable, and following a nonzero sector recursively identifies one computational basis representative in each sector support. Each block carries at most \(r\) sector states, each with at most \(r\) child sectors. At every split, the sector states on the same child block have pairwise disjoint supports in the computational basis.

Using these representatives, the update at a block opens the product representatives of its Schmidt sectors in a fixed order, with split weights \(|\lambda_s|^2\), and then refines the two child representatives recursively. Pairwise disjoint child supports prevent recursive refinements in different sectors from interfering. The controlled local-update module of Sec.~\ref{subsec:sm-controlled-updates} isolates each pair using only qubits in the current block.

The magnitudes \(|\lambda_s|\) set the rotations through the split weights \(|\lambda_s|^2\). The phase of each \(\lambda_s\) is written by a complex two-level Givens update. Relative phases within the supplied child-sector states are propagated recursively by the same updates. An update that opens a sector with relative phase \(\phi\) has first column \((\cos\theta,e^{i\phi}\sin\theta)\) on its representative pair. A phase attached to a terminal sector is absorbed into its parent update. These phase-aware updates use the same support and controls as the real rotations and do not change the asymptotic resource bounds. For example, \((|00\rangle+i|11\rangle)/\sqrt2\) is obtained from \(|00\rangle\) with \(\theta=\pi/4\) and \(\phi=\pi/2\).

A four-qubit example makes the sector isolation explicit.
Split the qubits as \(12|34\) and define
\[
|\chi\rangle
=
\sqrt{\tfrac35}|00\rangle_{12}|00\rangle_{34}
+\sqrt{\tfrac25}|\phi_L\rangle_{12}|\phi_R\rangle_{34},
\]
where \(|\phi_L\rangle=(|01\rangle+|10\rangle)/\sqrt2\) and \(|\phi_R\rangle=(\sqrt3|01\rangle+|10\rangle)/2\). Use \(0000\) as the seed and \(0101\) as the representative of the second sector. The root update transfers weight \(2/5\), so \(\theta_{\rm root}=\arcsin\sqrt{2/5}\).

Inside the opened sector, the child refinements are
\[
|01\rangle_{12}
\longmapsto
\frac{|01\rangle_{12}+|10\rangle_{12}}{\sqrt2},
\quad
|01\rangle_{34}
\longmapsto
\sqrt{\tfrac34}|01\rangle_{34}
+\sqrt{\tfrac14}|10\rangle_{34}.
\]
Thus \(\theta_L=\arcsin\sqrt{1/2}\) and \(\theta_R=\arcsin\sqrt{1/4}\). The bare \(01\to10\) update on the left child leaves the first sector unchanged.
The right update leaves the first sector untouched because its right child remains \(00\). The final state is
\[
\sqrt{\tfrac35}|0000\rangle
+\sqrt{\tfrac{3}{20}}\bigl(|0101\rangle+|1001\rangle\bigr)
+\sqrt{\tfrac{1}{20}}\bigl(|0110\rangle+|1010\rangle\bigr),
\]
as shown in Fig.~\ref{fig:sm-low-rank-controlled-circuit}.

The uniform even-parity state used for the resource experiment satisfies this condition because its even and odd child sectors have disjoint supports in the computational basis at every split. It is prepared by the same blockwise local updates, with at most two parity sectors in each block.

\begin{figure}[t]
\centering
\begin{tikzpicture}[
wire/.style={line width=0.45pt},
rootgate/.style={draw, rounded corners=2pt, fill=white, minimum width=2.95cm, minimum height=2.40cm, align=center, font=\footnotesize},
childgate/.style={draw, rounded corners=2pt, fill=white, minimum width=2.92cm, minimum height=2.40cm, align=center, font=\footnotesize},
rightgate/.style={draw, rounded corners=2pt, fill=white, minimum width=2.92cm, minimum height=1.10cm, align=center, font=\footnotesize},
lab/.style={font=\footnotesize, anchor=east},
note/.style={font=\footnotesize, align=center}
]
\def\xstart{-0.45}
\def\xend{10.85}

\foreach \y/\name in {0/\(x_1\),-0.58/\(x_2\),-1.16/\(x_3\),-1.74/\(x_4\)} {
  \node[lab] at (\xstart-0.25,\y) {\name};
  \draw[wire] (\xstart,\y) -- (\xend,\y);
}

\draw[dashed, rounded corners=2pt] (-0.20,0.33) rectangle (0.60,-2.07);
\node[note] at (0.20,0.58) {seed \(=|0000\rangle\)};

\node[rootgate] (uroot) at (2.48,-0.87)
{\(U_{\rm root}\)\\[0.25em]
\(\tau_{(\mu_2,\mu_4)=(+1,+1)}(\theta_{\rm root})\)\\[0.25em]
\(0000\to0101\)};
\node[note] at (2.48,-2.40)
{\(\theta_{\rm root}=\arcsin\sqrt{2/5}\)};

\node[childgate] (ul) at (5.81,-0.87)
{\(U_L\)\\[0.25em]
\(\tau_{(\mu_1,\mu_2)=(+1,-1)}(\theta_L)\)\\[0.25em]
\(0101\to1001\)};
\node[note] at (5.81,-2.40)
{\(\theta_L=\arcsin\sqrt{1/2}\)};

\node[rightgate] (ur) at (9.14,-1.45)
{\(U_R\)\\[0.25em]
\(\tau_{(\mu_3,\mu_4)=(+1,-1)}(\theta_R)\)\\[0.25em]
\(01\to10\)};
\node[note] at (9.14,-2.40)
{\(\theta_R=\arcsin\sqrt{1/4}\)};

\node[note, anchor=west] at (10.89,-0.87)
{output\\[-0.1em]
\(\sqrt{\tfrac35}|0000\rangle+\sqrt{\tfrac{3}{20}}(|0101\rangle+|1001\rangle)\)\\[-0.1em]
\(\hspace{0.6em}+\sqrt{\tfrac{1}{20}}(|0110\rangle+|1010\rangle)\)};

\end{tikzpicture}
\caption{Controlled local-update circuit for the hierarchical low-rank example. The first Schmidt sector is \(|00\rangle_{12}|00\rangle_{34}\), and the second is \(|\phi_L\rangle_{12}|\phi_R\rangle_{34}\). Starting from \(|0000\rangle\), the root update opens \(|0101\rangle\) as the representative of the second sector. The modules \(U_L\) and \(U_R\) refine its left and right child factors into \(|\phi_L\rangle\) and \(|\phi_R\rangle\). The displayed angles match the Schmidt and child weights in the text.}
\label{fig:sm-low-rank-controlled-circuit}
\end{figure}
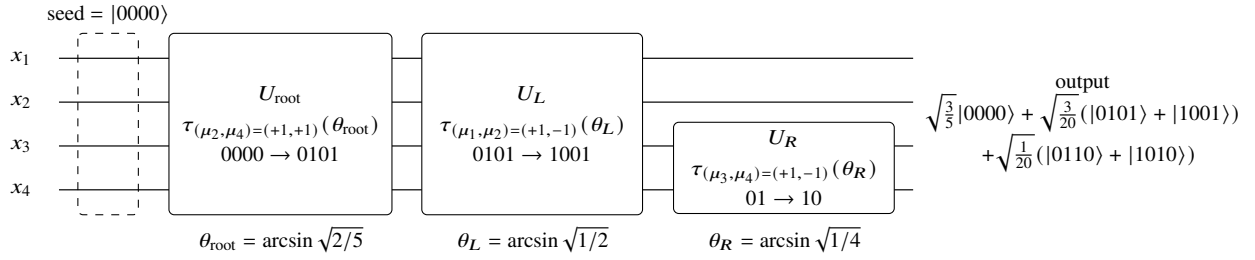

\begin{corollary}[Hierarchical low-rank resources]
For efficiently specified balanced Schmidt trees with \(r=O(1)\) whose sector states on each child block have pairwise disjoint supports in the computational basis at every split, the blockwise construction gives an ancilla-free preparation with depth \(O(n)\) and gate count \(O(n\log n)\).
\end{corollary}

\begin{proof}
Consider a recursion layer whose blocks have size \(m_j=\Theta(n2^{-j})\), up to rounding. Each block carries at most \(r\) sector states, each with at most \(r\) child sectors, and needs \(O(r^2)=O(1)\) local updates. Disjoint sector supports keep every update within its intended sector. Each update may act on \(O(m_j)\) qubits and has ancilla-free depth and gate count \(O(m_j)\) by Sec.~\ref{subsec:sm-controlled-updates}. Blocks on the same layer are disjoint, so the layer depth is \(O(m_j)\), while its total gate count is \(O(2^jm_j)=O(n)\). Summing the layer depths gives \(O(n)\), and summing the gate counts over \(O(\log n)\) layers gives \(O(n\log n)\).
\end{proof}

\section{UFLP quantum search and speed-up proof}
\label{sec:sm-uflp}

The UFLP construction uses an all-open intermediate representation for state preparation.
Facility opening bits are closed only after the superposition over demand assignments has been prepared.
The search theorem then uses a smaller register over sets of open facilities and a reversible cost evaluation.

\subsection{Inequality encoding and all-open intermediate representation}
\label{subsec:sm-uflp-open}

For \(a\) facilities and \(b\) demands, let \([a]=\{1,\ldots,a\}\) and \([b]=\{1,\ldots,b\}\). The variable \(x_{ij}=1\) means that demand \(i\) is assigned to facility \(j\), and \(y_j=1\) means that facility \(j\) is open. The hard constraints are
\begin{equation}
\begin{aligned}
\sum_{j=1}^{a}x_{ij} &= 1 && \forall\, i\in[b],\\
x_{ij} &\le y_j && \forall\, i\in[b],\ \forall\, j\in[a],\\
\sum_{j=1}^{a}y_j &\le K, && 1\le K\le a.
\end{aligned}
\end{equation}
The opening limit \(K\) bounds the number of facilities that may be open, and \(K=a\) recovers standard UFLP. With nonnegative opening costs \(f_j\ge0\) and service costs \(d_{ij}\ge0\), the objective is \(\min \sum_j f_j y_j+\sum_i\sum_j d_{ij}x_{ij}\). The formulation uses \(a+ab\) problem bits but contains redundant bit strings because unused facilities may remain open. Closing unused facilities never increases the objective. For the recursive stage, we use the temporary all-open representation
\begin{equation}
\mathcal F_{\rm open}
=
\left\{
(x,y):
\sum_{j=1}^{a}x_{ij}=1\ \forall i,\ 
y_j=1\ \forall j
\right\}.
\end{equation}
This set contains exactly one bit string for each assignment matrix \(x\), so \(|\mathcal F_{\rm open}|=a^b\). It is an intermediate circuit representation and need not satisfy the opening limit when \(K<a\).

The closing pass removes the redundant opening bits. For an assignment \(x\), define \(\widetilde y_j(x)=\bigvee_i x_{ij}\), where \(\bigvee_i x_{ij}=1\) exactly when \(x_{ij}=1\) for some \(i\). The resulting set of demand assignment strings is
\begin{equation}
\mathcal F_{\rm asg}
=
\left\{
\bigl(x,\widetilde y(x)\bigr):
\sum_{j=1}^{a}x_{ij}=1\ \forall i
\right\}.
\end{equation}
It also has size \(a^b\). The full set satisfying the assignment constraints and the implications \(x_{ij}\le y_j\) is larger because unused facilities may remain open, but \(\mathcal F_{\rm open}\) and \(\mathcal F_{\rm asg}\) encode the same assignment information. Strings in \(\mathcal F_{\rm asg}\) that use more than \(K\) facilities are rejected by the marking oracle. The recursive stage works with assignment rows only, and the implications \(x_{ij}\le y_j\) are restored deterministically at the end.

Alg.~\ref{alg:uflp-schedule} records the recursive schedule and closing pass. The balanced recursion on \([a]\) defines the subsets \(S\). Let \(e_j\in\{0,1\}^a\) be defined by \((e_j)_r=1\) for \(r=j\) and \((e_j)_r=0\) otherwise. For demand row \(i\), let \(W_i(S)\) be the target weight on facilities in \(S\), and choose representative \(c_{i,S}=e_{\max(S)}\). For the uniform target, \(W_i(S)=|S|\). For a target assignment \(t\), \(W_i(S)=1\) if \(t_i\in S\) and \(0\) otherwise.

\begin{algorithm}[H]
\caption{Recursive preparation schedule for demand assignment strings. The circuit prepares one assignment per demand with all facilities initially open, then closes every facility that serves no demand.}
\label{alg:uflp-schedule}
\begin{algorithmic}[1]
\Statex \textbf{Input.} Facilities \([a]\), demands \([b]\), and either the uniform target or a target assignment \(t=(t_1,\ldots,t_b)\)
\Statex \textbf{Registers.} Assignment bits \(x\), facility opening bits \(y\)
\State Initialize each assignment row \(x_i\gets c_{i,[a]}=e_a\) and set \(y_j\gets1\)
\For{each layer of the balanced recursion on facility subsets}
    \State Execute the row/subset updates below in parallel
    \For{each demand row \(i\) and each subset \(S\) with \(W_i(S)>0\)}
        \State Split \(S\) into child subsets \(S_L\sqcup S_R\), ordered so that \(c_{i,S}=c_{i,S_R}\)
        \State Set \(\sin^2\theta_{i,S}\gets W_i(S_L)/W_i(S)\), where \(W_i(S)=W_i(S_L)+W_i(S_R)\)
        \State Apply the row-\(i\) reassignment update from \(c_{i,S_R}\) to \(c_{i,S_L}\) with angle \(\theta_{i,S}\)
    \EndFor
\EndFor
\For{each facility \(j\in[a]\) in parallel}
    \State Apply the controlled local update \(1\to0\) on \(y_j\), conditioned on \(x_{1j}=\cdots=x_{bj}=0\)
\EndFor
\State Output the demand assignment state with unused facilities closed
\end{algorithmic}
\end{algorithm}

\subsection{Recursive assignment construction}
\label{subsec:sm-uflp-recursion}

The seed assigns every demand to the facility with the largest index and keeps all facilities open.

Fix a demand \(i\). For a facility subset \(S\subseteq[a]\) with \(W_i(S)>0\), the residual constraint requires demand \(i\) to choose one facility in \(S\). The local feasible set is \(F_{i,S}=\{e_j:j\in S\}\), and the representative rule chooses \(c_{i,S}=e_{\max(S)}\).

Split \(S\) into two nearly balanced subsets \(S_L\sqcup S_R=S\), ordered so that \(\max(S)\in S_R\). The recursive decomposition is \(F_{i,S}=F_{i,S_L}\sqcup F_{i,S_R}\). The two child representatives are \(e_{\max(S_L)}\) and \(e_{\max(S_R)}\), and \(e_{\max(S_R)}=e_{\max(S)}\) is the representative already present before this split.

For the uniform superposition over demand assignments, the split weights are \(|S_L|\) and \(|S_R|\), so the angle rule gives \(\sin^2\theta=|S_L|/|S|\).
The source representative \(e_{\max(S_R)}\) keeps probability \(|S_R|/|S|\), and the left representative receives probability \(|S_L|/|S|\).
For a target assignment, the split assigns probability one to the child side containing the target facility.

The nontrivial local reassignment toggles only the two assignment bits attached to the child representatives, \(x_{i,\max(S_R)}:1\to0\) and \(x_{i,\max(S_L)}:0\to1\).
All \(y\)-bits remain fixed at \(1\). Both before and after the update, the inequality \(x_{ij}\le y_j\) holds for every \(j\), so the recursive stage never leaves the all-open intermediate representation.

\begin{lemma}[Locality of UFLP reassignment]
Each recursive reassignment update has constant support inside one demand row and preserves the all-open intermediate representation.
\end{lemma}

\begin{proof}
For a split \(S=S_L\sqcup S_R\), the nontrivial transition acts only on the two bits \(x_{i,\max(S_R)}\) and \(x_{i,\max(S_L)}\). With all facility bits fixed to \(y_j=1\), the two relevant local patterns are
\((1,0)\) before the update and \((0,1)\) after it. The controlled update between the source and target patterns acts only on this pair. The row constraint \(\sum_j x_{ij}=1\) is preserved, since one assignment bit is turned off while another is turned on. Since all facility bits remain fixed at \(y_j=1\), every inequality \(x_{ij}\le y_j\) continues to hold. The update support lies in row \(i\). At a fixed recursion layer, different subsets in the same row have disjoint changing pairs \((\max(S_R),\max(S_L))\), and different demand rows use disjoint \(x\)-registers. The updates in the same layer remain parallel without controls spanning the row to select the subset.
\end{proof}

The construction is applied independently to every demand row. At a fixed recursion layer, the facility subsets \(S\) being split are disjoint within each row, and different rows use disjoint \(x\)-registers. The local reassignments can be scheduled in parallel.

\subsection{Resource bounds and closing pass}
\label{subsec:sm-uflp-resources}

For one demand row, the balanced recursion on \(a\) facilities has \(O(\log a)\) layers and \(a-1\) nonterminal splits. Each split uses one reassignment with constant support. One row costs depth \(O(\log a)\) and gate count \(O(a)\). Since the \(b\) rows are processed in parallel layer by layer, the all-open recursive stage has depth \(O(\log a)\) and gate count \(O(ab)\).

The closing pass acts independently on each facility column. Since the recursive stage starts from \(y_j=1\), a controlled \(1\to0\) update conditioned on \(x_{1j}=\cdots=x_{bj}=0\) sets \(y_j\) to the value \(\widetilde y_j(x)\) defined above.
Column costs multiply by \(a\), while the depth equals that of one column when the listed workspace is available for every column in parallel.

For one column, the resource of this controlled update depends on the available clean workspace.
The implementation without clean ancillas used here gives depth \(O(b)\) and gate count \(O(b)\) for the \(C^bX\) gate with negative controls~\cite{huang2024compiling}.
With \(\Theta(b)\) clean ancilla qubits per column, a reversible balanced OR tree gives depth \(O(\log b)\) and gate count \(O(b)\).

The total UFLP scalings are
\begin{center}
\small
\setlength{\tabcolsep}{8pt}
\renewcommand{\arraystretch}{1.18}
\begin{tabular}{@{}ccc@{}}
\toprule
closing pass implementation & gate count & depth\\
\midrule
without clean ancillas & \(O(ab)\) & \(O(\log a+b)\)\\
\(\Theta(b)\) clean ancilla per column & \(O(ab)\) & \(O(\log a+\log b)\)\\
\bottomrule
\end{tabular}
\end{center}

\begin{lemma}[Deterministic closing pass]
The closing pass maps \(\mathcal F_{\rm open}\) bijectively to \(\mathcal F_{\rm asg}\). The listed regimes in the table above are achieved by the corresponding implementation of the zero test for each column.
\end{lemma}

\begin{proof}
For each facility \(j\), the output bit \(\widetilde y_j(x)=\bigvee_i x_{ij}\) is uniquely determined by the assignment column \(x_{1j},\ldots,x_{bj}\). The map leaves \(x\) unchanged and sets \(y\) uniquely from \(x\), so it is a bijection from \(\mathcal F_{\rm open}\) to \(\mathcal F_{\rm asg}\). The controlled \(1\to0\) update acts independently on each facility column. Multiplying the gate count per column in the table by \(a\), while running disjoint columns at the same depth under the listed workspace convention, gives the total resource bounds.
\end{proof}

Corollary~\ref{cor:uflp-state-preparation} uses the row without clean ancillas, which gives depth \(O(\log a+b)\) and gate count \(O(ab)\).

\begin{proof}[Proof of Corollary~\ref{cor:uflp-state-preparation}]
The all-open representation separates the assignment recursion from the inequalities \(x_{ij}\le y_j\). The first step is to prepare the desired assignment distribution in \(\mathcal F_{\rm open}\), where all facility bits satisfy \(y_j=1\) and hence every inequality \(x_{ij}\le y_j\) holds automatically.

For one demand row \(i\), write \(|e_j\rangle_i\) for the assignment to facility \(j\). At a subset \(S\) with \(W_i(S)>0\), the representative is \(|e_{\max(S)}\rangle_i\). A split \(S=S_L\sqcup S_R\) uses split weights \(|S_L|\) and \(|S_R|\) for the uniform target, so the angle rule in Sec.~\ref{subsec:sm-split-weight-angles} sets the conditional probabilities to \(|S_L|/|S|\) and \(|S_R|/|S|\). Induction over the balanced splits gives
\begin{equation}
|e_{\max(S)}\rangle_i
\longmapsto
\frac{1}{\sqrt{|S|}}
\sum_{j\in S}|e_j\rangle_i .
\end{equation}
Taking \(S=[a]\) and applying the same layer schedule to all \(b\) rows gives
\begin{equation}
\bigotimes_{i=1}^{b}|e_a\rangle_i
\longmapsto
\frac{1}{\sqrt{a^b}}
\sum_{j_1,\ldots,j_b\in[a]}
|e_{j_1}\rangle_1\cdots |e_{j_b}\rangle_b ,
\end{equation}
which is the exact uniform superposition over the \(a^b\) all-open basis states. For a prescribed demand assignment, each split puts probability one on the unique child subset containing the target facility for that row. The same recursion maps the seed to the corresponding basis state.

The local reassignment updates only \(x\) while leaving \(y_j=1\), so every intermediate state remains in \(\mathcal F_{\rm open}\). The closing pass computes \(\widetilde y_j=\bigvee_i x_{ij}\), which maps \(\mathcal F_{\rm open}\) bijectively to \(\mathcal F_{\rm asg}\) without changing the assignment amplitudes.

The resource bounds for Corollary~\ref{cor:uflp-state-preparation} use the linear-width closing pass. The recursive assignment stage has depth \(O(\log a)\) and gate count \(O(ab)\). The closing pass has depth \(O(b)\) and gate count \(O(ab)\). The total depth is \(O(\log a+b)\) and the total gate count is \(O(ab)\), proving Corollary~\ref{cor:uflp-state-preparation}.
\end{proof}

\subsection{Demand-assignment and facility-set search}
\label{subsec:sm-uflp-search}

First consider search over the demand assignment strings. Let \(A_{\rm asg}\) be the preparation unitary from Corollary~\ref{cor:uflp-state-preparation}, so that
\[
A_{\rm asg}|0\rangle
=
\frac{1}{\sqrt{a^b}}
\sum_{z\in\mathcal F_{\rm asg}} |z\rangle,
\qquad
|\mathcal F_{\rm asg}|=a^b .
\]
This unitary has depth \(O(\log a+b)\) and gate count \(O(ab)\).

For a threshold \(B\), the marking oracle evaluates the UFLP cost, compares it with \(B\), applies the phase and uncomputes the cost~\cite{gilliam2019foundational,gilliam2021grover}. It also enforces the opening limit when \(K<\min(a,b)\). If \(M_B\) strings are marked, Grover search uses \(O(\sqrt{a^b/M_B})\) iterations~\cite{grover1996fast,brassard2000quantum}. The worst case therefore uses \(O(a^{b/2})\) state-preparation and marking-oracle circuit applications. Direct search over the full \(a+ab\) problem bits instead uses \(O(2^{(a+ab)/2})\) circuit applications in the worst case.

For optimization, the assignment bits need not be searched explicitly. Nonnegative opening costs ensure that closing an unused facility never increases the objective. An optimum therefore exists with no unused facility open and with between \(1\) and \(m=\min(K,b)\) open facilities. Define
\[
\mathcal Y_{\le m}
=
\{Y\subseteq[a]:1\le |Y|\le m\},
\qquad
N_{a,b,K}=|\mathcal Y_{\le m}|=\sum_{r=1}^{m}\binom{a}{r}.
\]
For a fixed nonempty set \(Y\) of open facilities, the best assignment sends each demand to its cheapest open facility. The induced cost is
\[
C(Y)
=
\sum_{j\in Y}f_j+\sum_{i=1}^{b}\min_{j\in Y}d_{ij}.
\]
After an optimal \(Y\) is found, an optimal assignment is recovered by the same rule that selects the cheapest open facility, using a fixed tie-breaking convention.

The search register is the \(a\)-bit facility register, with \(y_j=1\) exactly when \(j\in Y\). Treat \(\mathcal Y_{\le m}\) as a bounded Hamming weight SFS. For \(1\le r\le m\), let \(c_{a,r}=0^{a-r}1^r\) be the right-packed representative of the weight-\(r\) layer. Seed the register at \(c_{a,1}\). The top expansion assigns probability \(\binom{a}{r}/N_{a,b,K}\) to \(c_{a,r}\), giving
\[
\sum_{r=1}^{m}
\sqrt{\frac{\binom{a}{r}}{N_{a,b,K}}}\,
|c_{a,r}\rangle .
\]
Within each populated layer, apply the fixed Hamming weight recursion of Sec.~\ref{subsec:sm-hamming} with uniform split weights. On the component generated from \(c_{a,r}\), this refinement performs
\[
|c_{a,r}\rangle
\longmapsto
\frac{1}{\sqrt{\binom{a}{r}}}
\sum_{|Y|=r}|Y\rangle .
\]
The controls of Lemma~\ref{lem:sm-dicke-block-local-controls} isolate the intended packed representative pairs within each layer, and the fixed weight refinements preserve \(|Y|\). The layer mass and within layer normalization cancel, so every \(Y\in\mathcal Y_{\le m}\) has the same amplitude. The resulting unitary \(A_Y\) satisfies
\[
A_Y|0\rangle
=
\frac{1}{\sqrt{N_{a,b,K}}}
\sum_{Y\in\mathcal Y_{\le m}}|Y\rangle.
\]
The top expansion uses \(O(m)\) constant-support updates. The populated weight layers are then expanded with one common layerwise schedule. Counting the updates in their recursive expansion graphs gives a gate count of \(O(m^2a)\) and polynomial depth for \(A_Y\). The required angles can be computed classically in polynomial time.

The marking oracle needs a reversible evaluation of \(C(Y)\). Let \(L\) bound the bit length of each opening or service cost. Since \(C(Y)\) is a sum of at most \(a+b\) such terms, the cost register and threshold use
\[
w=L+\lceil\log_2(a+b+1)\rceil+O(1).
\]
Controlled additions of \(f_j\) give the opening cost term.

For one demand \(i\), load candidates
\[
r_{ij}=
\begin{cases}
d_{ij}, & y_j=1,\\
M_i, & y_j=0,
\end{cases}
\qquad
M_i>\max_j d_{ij}.
\]
A reversible comparator-swap network moves \(\min_j r_{ij}\) to a designated register. Its comparison outcomes are kept as history bits while the network is live. After adding the minimum to the cost register, reverse the network to uncompute these bits, restore the candidates, and unload them.

Using standard ripple-carry reversible arithmetic~\cite{vedral1996quantum,barenco1995elementary}, controlled additions, comparisons and conditional swaps on \(w\)-bit records have conservative bounds of \(O(w)\) for both depth and gate count. The complete cost evaluation has depth and gate count \(O((a+ab)w)\) and uses \(O(aw)\) work qubits, which are reused across demands.

For a threshold \(B\), the marking oracle compares the cost register with \(B\), applies the phase and uncomputes the cost evaluation. Quantum minimum finding returns the exact optimum with bounded error using \(O(\sqrt{N_{a,b,K}})\) applications of \(A_Y\), \(A_Y^\dagger\) and the objective oracle~\cite{durr1996quantum}.

\subsection{Conditional complexity bounds}
\label{subsec:sm-uflp-complexity}

We use the strong exponential time hypothesis for unrestricted conjunctive normal form satisfiability (CNF-SETH)~\cite{hegerfeld2023tight}. For every fixed \(0<\delta<1\), it rules out deterministic satisfiability algorithms with running time \(2^{(1-\delta)n}\ell^{O(1)}\). Here \(n\) is the number of variables and \(\ell\) is the formula length, with no bound on clause width.

\begin{proof}[Proof of Theorem~\ref{thm:uflp-search-advantage}]
For \(K=a\), the search space has size \(N_{a,b,K}\leq2^a\), so quantum minimum finding uses \(O(2^{a/2})\) circuit applications. Each application has polynomial cost, including state preparation, cost evaluation and rotation synthesis to keep the total error bounded~\cite{ross2016optimal}. This gives the quantum running time \(O^*(2^{a/2})\).

For the classical lower bound, we reduce an arbitrary CNF formula to Hitting Set with universe size arbitrarily close to \(n\). Clause shortening keeps the size of each resulting instance polynomial in the formula length~\cite{dantsin2005derandomization}.

Let \(M\) be the number of clauses. Fix \(\eta>0\) and set \(k=\lceil\log_2(M+n+1)/\eta\rceil\). Remove repeated literals and tautological clauses, and decide formulas with no clauses or an empty clause directly. If a clause \(C\) has more than \(k\) literals, choose \(k\) of them and let \(D\) be their disjunction. Branch by either replacing \(C\) with \(D\) or setting every literal of \(D\) to false and simplifying. The formula is satisfiable exactly when at least one branch is satisfiable.

Along any path, the first branch occurs at most \(M\) times, since it permanently shortens one long clause. The second fixes \(k\) distinct remaining variables and occurs at most \(\lfloor n/k\rfloor\) times. Thus the number of output formulas is at most
\[
\binom{M+\lfloor n/k\rfloor}{\lfloor n/k\rfloor}
\leq (M+n+1)^{n/k}
\leq 2^{\eta n}.
\]
Each output has clause width at most \(k=O(\log\ell)\), and the total construction time is \(2^{\eta n}\operatorname{poly}(\ell)\).

We encode each output using the constant-weight construction of Cygan et al.~\cite{cygan2016cnfsat}. Fix an odd integer \(p\geq3\) and pad the variables into \(g\leq\lceil n/p\rceil\) blocks of \(p\) variables. Give block \(i\) a disjoint universe \(U_i\) of size \(q=p+2\lceil\log_2p\rceil\), and put \(r=(q+1)/2\). Since \(\binom qr\geq2^q/(q+1)\geq2^p\), its assignments can be represented by distinct \(r\)-subsets \(\psi_i(\alpha_i)\) of \(U_i\).

Set the hitting-set budget to \(rg\) and include every \(r\)-subset of each \(U_i\) as a set to be hit. Since \(q=2r-1\), hitting all these sets requires at least \(r\) elements from each block. The budget forces equality. For every \(r\)-subset \(B\) not used to represent an assignment, also include \(U_i\setminus B\). This excludes exactly the unused representations, so the remaining choices correspond to block assignments.

For a clause \(C\), let \(I_C\) be the blocks containing its variables. For every tuple \((\alpha_i)_{i\in I_C}\) that falsifies all literals of \(C\), add the set
\[
\bigcup_{i\in I_C}\bigl(U_i\setminus\psi_i(\alpha_i)\bigr).
\]
A valid block encoding misses this set exactly when it selects that tuple. It therefore hits every clause set if and only if the encoded assignment satisfies the formula.

The block constraints contribute at most \(2g\binom qr\leq g2^q\) sets, and each clause contributes at most \(2^{pk}\). The resulting universe size and number of sets satisfy
\[
a=qg\leq\frac qp n+q,
\qquad
b\leq g2^q+M2^{pk}.
\]
For fixed \(p\), the bound \(k=O(\log\ell)\) makes both the construction time and the complete output length polynomial in \(\ell\).

Turn the \(a\) universe elements into facilities and the \(b\) sets \(E_i\) into demands. Use unit opening costs, with \(d_{ij}=0\) when \(j\in E_i\) and \(d_{ij}=a+1\) otherwise. Opening all facilities costs \(a\), while missing a set incurs a service cost of at least \(a+1\). An optimum must therefore hit every set, and its cost is the number of open facilities. Thus the UFLP optimum equals the minimum hitting-set size, with the same parameter \(a\).

Suppose a deterministic exact UFLP algorithm runs in \(O^*(2^{\rho a})\) time for some fixed \(0<\rho<1\). Since \(q/p\to1\) as \(p\to\infty\), choose \(p\) and \(\eta\) so that \(\eta+\rho q/p<1\). Applying the algorithm to all \(2^{\eta n}\) instances would decide the original formula in time
\[
2^{(\eta+\rho q/p)n}\operatorname{poly}(\ell),
\]
contradicting unrestricted CNF-SETH.
\end{proof}

\begin{lemma}[Classical lower bound for a fixed opening limit]
\label{lem:sm-fixed-k-lower}
Assume unrestricted CNF-SETH. For every fixed \(K\geq2\) and \(0<\epsilon<K\), opening-limited UFLP has no deterministic exact algorithm running in time \(O(a^{K-\epsilon}\operatorname{poly}(b,L))\). This holds even with unit opening costs and service costs in \(\{0,K+1\}\).
\end{lemma}

\begin{proof}
For fixed \(K\), the split-and-list reduction to \(K\)-Set Cover gives the required dependence on \(a\)~\cite{patrascu2010possibility}. Start with an \(n\)-variable formula with \(M\) clauses and pad it into \(K\) blocks of \(\lceil n/K\rceil\) variables. The set-cover universe contains one element for each clause and one marker for each block. For each block assignment, create a set containing its marker and the clauses satisfied by that assignment. This gives \(a=K2^{\lceil n/K\rceil}\) sets and \(b=M+K\) elements. A cover with at most \(K\) sets must choose one assignment from each block to cover all markers. It covers the clause elements exactly when the combined assignment satisfies the formula.

Make each set a facility and each element a demand. Use opening limit \(K\), unit opening costs and service cost zero for membership or \(K+1\) otherwise. A cover gives a UFLP solution of cost at most \(K\). Conversely, any solution of cost at most \(K\) uses only zero-cost service edges and its open facilities form a cover. The costs have constant bit length for fixed \(K\).

An \(O(a^{K-\epsilon}\operatorname{poly}(b,L))\)-time algorithm would therefore decide satisfiability in time
\[
\left(2^{n/K}+2^{(1-\epsilon/K)n}\right)\operatorname{poly}(\ell),
\]
including construction of the UFLP instance. Both exponent coefficients are smaller than one, contradicting unrestricted CNF-SETH.
\end{proof}

To compare total running times, store the \(a+ab\) costs explicitly in \(L\)-bit records, giving input length \(s=\Theta(abL)\). Lemma~\ref{lem:sm-fixed-k-lower} rules out deterministic exact time \(O(s^{K-\epsilon})\), since this would be \(O(a^{K-\epsilon}\operatorname{poly}(b,L))\). Classical enumeration takes \(O(a^Kbw)=O(s^K\log s)\) time, so the optimal classical exponent in the full input length is \(K\).

For fixed \(K\), state preparation uses \(O(a)\) local updates and cost evaluation uses \(O(abw)\) gates, where \(w=L+O(\log(a+b))\). The preparation angles follow from binomial coefficients of order at most \(K\). Computing and synthesizing each rotation to the required precision costs \(\operatorname{polylog}(a+b)\) time~\cite{ross2016optimal}. Each circuit application therefore takes \(O(s\operatorname{polylog}s)\) time. Quantum minimum finding uses \(O(a^{K/2})\) such applications. Since \(a\leq s\), the total running time satisfies
\[
T_{\mathrm Q}
=O\!\left(a^{K/2}s\operatorname{polylog}s\right)
=O\!\left(s^{K/2+1}\operatorname{polylog}s\right).
\]
This is below the classical threshold for every fixed \(K\geq3\). The exponent ratio is \((K/2+1)/K=1/2+1/K\), so the conditional speed-up becomes asymptotically quadratic as \(K\) increases.

\subsection{End-to-end threshold search trajectories}
\label{subsec:sm-uflp-trajectory}

For a given search register, let \(\Omega\) be its search support, let \(C(z)\) be the UFLP cost of \(z\in\Omega\), and let \(A\) prepare the uniform state over \(\Omega\). The phase oracle checks any feasibility condition not enforced by the support. Figure~\ref{fig:ncs-uflp-scaling} uses \(K=a\), so the displayed circuits require no separate opening-limit check.

At round \(t\), the phase oracle \(S_\chi\) marks \(|z\rangle\) when \(z\) is feasible and \(C(z)<B_t\), and \(S_0\) is the reflection about the zero state. One Grover iterate is \(G=AS_0A^\dagger S_\chi\).
If the marked count
\[
M_t
=
\left|
\left\{
z\in\Omega
\ \middle|\
z\text{ is feasible and }C(z)<B_t
\right\}
\right|
\]
were known and positive, the first success probability peak would be reached by choosing
\[
j_t^\star
=\max\left\{0,\operatorname{round}\!\left(
\frac{\pi}{4\arcsin\sqrt{M_t/|\Omega|}}-\frac12
\right)\right\}
\]
Grover iterates. We set \(j_t^\star=0\) when \(M_t=0\). In an optimization algorithm, however, \(M_t\) is generally not known in advance and changes whenever the threshold changes.
The sampled trajectories in Fig.~\ref{fig:ncs-uflp-scaling}i use the standard adaptive search rule for unknown marked fractions~\cite{boyer1998tight,durr1996quantum,gilliam2021grover}.
Each independent trajectory keeps a real search window \(k_t\), initialized to one. At round \(t\), it draws \(j_t\) uniformly from \(\{0,\ldots,\lceil k_t\rceil-1\}\), applies \(G^{j_t}\), measures once, and updates \(B_{t+1}\) only if the measured string is feasible and has lower cost. A successful round resets \(k_{t+1}=1\). An unsuccessful round sets \(k_{t+1}=\min((8/7)k_t,\sqrt{|\Omega|})\).
Qiskit checks on small instances verify the corresponding state preparation, cost threshold phase, uncomputation and diffusion blocks by exact state vector comparison.

\section{SFS initial states for structural preferences}
\label{sec:sm-sfs-dqi}

SFS-DQI changes only the initial state prepared before the constraint encoding of standard OPI DQI. The subsequent DQI operations remain unchanged. We first derive the modified output amplitude and then specialize the construction to the two-group balance preference.

\subsection{Modified DQI state}
\label{subsec:sm-dqi-general-sector}

Let \(q\) be prime and let \(\mathbb F_q\) denote the corresponding prime field~\cite{jordan2025decoded}. Label the \(m=q-1\) nonzero field elements by \(i=1,\ldots,m\) and use them as evaluation points, with \(1\leq n<m\). An instance specifies an allowed set \(F_i\subset\mathbb F_q\) for each point, with \(|F_i|=r\) for all \(i\) and \(0<r<q\). Each polynomial is \(Q(z)=\sum_{j=0}^{n-1}\xi_jz^j\), with coefficient vector \(\bm\xi=(\xi_0,\ldots,\xi_{n-1})^T\in\mathbb F_q^n\).
Let \(f_i(Q)=1\) when \(Q(i)\in F_i\) and \(f_i(Q)=0\) otherwise. Then \(f(Q)=\sum_{i=1}^{m}f_i(Q)\) counts the satisfied constraints.

To obtain the normalized local Fourier state used below, we center and rescale the constraint indicator. A uniformly sampled \(z\in\mathbb F_q\) lies in \(F_i\) with probability \(\alpha=r/q\), and the centered indicator has \(\ell_2\) norm \(\sigma=\sqrt{r(q-r)/q}\). Define \(g_i(z)=(1-\alpha)/\sigma\) for \(z\in F_i\) and \(g_i(z)=-\alpha/\sigma\) otherwise. It follows that \(\sum_{z\in\mathbb F_q}g_i(z)=0\) and \(\sum_{z\in\mathbb F_q}|g_i(z)|^2=1\). To reduce repeated factors to products with distinct constraint indices, we use the quadratic identity
\begin{equation}
g_i(z)^2=\lambda g_i(z)+q^{-1},
\qquad
\lambda=\frac{q-2r}{\sqrt{qr(q-r)}}.
\label{eq:sm-opi-quadratic-reduction}
\end{equation}
The constraint encoding converts a binary support mask into amplitudes on the \(\mathbb F_q\) error register using the finite-field Fourier expansion of each \(g_i\). Let \(\mathbb F_q^\times=\mathbb F_q\setminus\{0\}\), let \(\omega_q=\exp(2\pi\mathrm i/q)\) and let \(\widetilde g_i(a)=q^{-1/2}\sum_{z\in\mathbb F_q}\omega_q^{az}g_i(z)\). Centering gives \(\widetilde g_i(0)=0\), while Parseval's identity gives \(\sum_{a\ne0}|\widetilde g_i(a)|^2=1\). Fourier inversion gives \(g_i(z)=q^{-1/2}\sum_{a\in\mathbb F_q^\times}\omega_q^{-az}\widetilde g_i(a)\). Each binary mask bit is represented by the \(|0\rangle\) and \(|1\rangle\) levels of a \(q\)-dimensional register. The normalized nonzero Fourier coefficients define the unitary
\begin{equation}
\Gamma_i|0\rangle=|0\rangle,
\qquad
\Gamma_i|1\rangle=\sum_{a\in\mathbb F_q^\times}\widetilde g_i(a)|a\rangle.
\end{equation}
The two images of \(\Gamma_i\) are normalized and orthogonal, so this action admits a unitary extension.

To track the output amplitude produced by a general mask state, let \(B\in\mathbb F_q^{m\times n}\) be the Vandermonde matrix with \(B_{i,j+1}=i^j\). It collects the polynomial evaluations as \((B\bm\xi)_i=Q(i)\). Write \([m]=\{1,\ldots,m\}\). For \(S\subseteq[m]\), write \(g_S(Q)=\prod_{i\in S}g_i(Q(i))\), with \(g_\varnothing(Q)=1\). The subset form ensures that every constraint index appears at most once in each product. Let \(R\) bound the retained mask weight, with \(0\leq R\leq m\), and consider the nonzero amplitude function \(H(Q)=\sum_{S\subseteq[m],\,|S|\leq R}a_Sg_S(Q)\).
Here \(|1_S\rangle\) denotes the binary support mask of \(S\). To compensate for the normalization introduced by the inverse quantum Fourier transform, the amplitude of this mask includes the factor \(q^{(n-|S|)/2}\). The normalized mask state is
\begin{equation}
|\Psi_H\rangle
=
\frac{1}{\sqrt{C_H}}
\sum_{\substack{S\subseteq[m]\\|S|\leq R}}
q^{(n-|S|)/2}a_S|1_S\rangle,
\qquad
C_H=
\sum_{\substack{S\subseteq[m]\\|S|\leq R}}
q^{n-|S|}|a_S|^2.
\label{eq:sm-opi-general-mask-state}
\end{equation}
Applying \(\bigotimes_i\Gamma_i\) gives the error register state
\begin{equation}
|\Phi_H\rangle_E
=
\frac{1}{\sqrt{C_H}}
\sum_{\substack{S\subseteq[m]\\|S|\leq R}}
q^{(n-|S|)/2}a_S
\sum_{\substack{\bm e\in\mathbb F_q^m\\\operatorname{supp}(\bm e)=S}}
\left(\prod_{i\in S}\widetilde g_i(e_i)\right)|\bm e\rangle_E.
\label{eq:sm-opi-general-error-state}
\end{equation}
The unitaries \(\Gamma_i\) map the subset-mask coefficients \(a_S\) to the error-register amplitudes. The coefficients \(a_S\) are subset-basis coefficients and are distinct from the full field-vector Fourier coefficients \(a_y\) used in Methods. The local Fourier factors in Eq.~\eqref{eq:sm-opi-general-error-state} connect the two representations.

First compute \(\bm s=B^T\bm e\) reversibly in the syndrome register. Suppose that \(\bm e\mapsto B^T\bm e\) is injective on the support of the state in Eq.~\eqref{eq:sm-opi-general-error-state}. Reversible decoding then uses \(\bm s\) to uncompute the error register. The inverse quantum Fourier transform maps \(|\bm s\rangle\) to \(q^{-n/2}\sum_{\bm\xi\in\mathbb F_q^n}\omega_q^{-\bm s\cdot\bm\xi}|\bm\xi\rangle\). Since \(\bm s\cdot\bm\xi=\bm e\cdot B\bm\xi=\sum_i e_iQ(i)\), the amplitude of \(|\bm\xi\rangle\) is
\begin{align}
A_H(\bm\xi)
&=
\frac{1}{\sqrt{C_H}}
\sum_{\substack{S\subseteq[m]\\|S|\leq R}}a_S
\prod_{i\in S}
\left(
\frac{1}{\sqrt q}
\sum_{a\in\mathbb F_q^\times}
\widetilde g_i(a)\omega_q^{-aQ(i)}
\right)\nonumber\\
&=
\frac{1}{\sqrt{C_H}}
\sum_{\substack{S\subseteq[m]\\|S|\leq R}}a_Sg_S(Q)
=
\frac{H(Q)}{\sqrt{C_H}}.
\label{eq:sm-opi-amplitude-proof}
\end{align}
The modified DQI circuit therefore produces the polynomial-coefficient state \(C_H^{-1/2}\sum_{\bm\xi}H(Q)|\bm\xi\rangle\). The decoding cost is polynomial whenever the inverse syndrome map is efficiently and reversibly computable on this support.

Standard OPI DQI is recovered by taking \(H(Q)=P(f(Q))\). Let \(\ell\) be the degree of the score-weighting function \(P\), and let \(w_k\) be the normalized amplitudes in the corresponding superposition of weight-\(k\) Dicke states, with \(\sum_{k=0}^{\ell}|w_k|^2=1\). Set \(c_k=w_k/\sqrt{q^{n-k}\binom{m}{k}}\) and \(e_k(g(Q))=\sum_{|S|=k}g_S(Q)\). Then \(P(f(Q))=\sum_{k=0}^{\ell}c_ke_k(g(Q))\), and Eq.~\eqref{eq:sm-opi-general-mask-state} becomes \(\sum_{k=0}^{\ell}w_k|D_{m,k}\rangle\). Here \(|D_{m,k}\rangle\) denotes the uniform superposition over the \(m\)-bit strings of Hamming weight \(k\). Each \(g_i\) takes one value on \(F_i\) and another outside \(F_i\). Hence \(e_k(g(Q))\), and therefore the standard DQI amplitude, depends only on \(f(Q)\).

\subsection{Two-group state preparation and decoding}
\label{subsec:sm-balanced-profile}

To express the two-group balance preference in the centered basis above, partition the constraints into two prescribed groups \(G_1\sqcup G_2=[m]\), with \(m_g=|G_g|\). Define the group score \(S_g(Q)=\sum_{i\in G_g}f_i(Q)\) and its centered counterpart \(U_g(Q)=\sum_{i\in G_g}g_i(Q(i))\). At fixed total score, the product \(J(Q)=S_1(Q)S_2(Q)\) is larger when the satisfied constraints are distributed more evenly between the groups. Since \(S_g(Q)=m_g\alpha+\sigma U_g(Q)\), this preference factor expands as
\begin{equation}
J(Q)
=m_1m_2\alpha^2
+\alpha\sigma\bigl(m_2U_1(Q)+m_1U_2(Q)\bigr)
+\sigma^2U_1(Q)U_2(Q).
\label{eq:sm-opi-J-expansion}
\end{equation}
Equation~\eqref{eq:sm-opi-quadratic-reduction} reduces every repeated index in products of the centered functions. After this reduction, multiplying by \(J\) adds at most two distinct constraint indices, so \(P(f(Q))J(Q)\) contains products over at most \(R=\ell+2\) distinct indices.

To collect support masks by their weights in the two groups, for \(g\in\{1,2\}\) define \(e_a^{(g)}(Q)=\sum_{A\subseteq G_g,\,|A|=a}g_A(Q)\), and set \(e_a^{(g)}=0\) outside \(0\leq a\leq m_g\).
Splitting each subset between the two groups gives
\[
P(f(Q))=\sum_{a=0}^{m_1}\sum_{b=0}^{m_2}t_{a,b}e_a^{(1)}(Q)e_b^{(2)}(Q),
\]
where \(t_{a,b}=c_{a+b}\) for \(a+b\leq\ell\) and \(t_{a,b}=0\) otherwise. Multiplication by \(U_g\) remains in this basis because \(U_ge_a^{(g)}=(a+1)e_{a+1}^{(g)}+\lambda a e_a^{(g)}+q^{-1}(m_g-a+1)e_{a-1}^{(g)}\).
To represent multiplication by \(U_1\) and \(U_2\) directly on a coefficient array \(x=(x_{a,b})\), define
\[
(\mathcal T_1x)_{a,b}=a x_{a-1,b}+\lambda a x_{a,b}+q^{-1}(m_1-a)x_{a+1,b},
\qquad
(\mathcal T_2x)_{a,b}=b x_{a,b-1}+\lambda b x_{a,b}+q^{-1}(m_2-b)x_{a,b+1},
\]
where coefficients outside the allowed ranges are zero. Combining these transformations with Eq.~\eqref{eq:sm-opi-J-expansion} gives the grouped coefficients \(h_{a,b}\) through
\begin{equation}
\begin{aligned}
P(f(Q))J(Q)
&=
\sum_{a=0}^{m_1}\sum_{b=0}^{m_2}
h_{a,b}e_a^{(1)}(Q)e_b^{(2)}(Q),\\
h_{a,b}
&=m_1m_2\alpha^2t_{a,b}
+m_2\alpha\sigma(\mathcal T_1t)_{a,b}
+m_1\alpha\sigma(\mathcal T_2t)_{a,b}
+\sigma^2(\mathcal T_1\mathcal T_2t)_{a,b}.
\end{aligned}
\label{eq:sm-balanced-hab}
\end{equation}
The coefficients \(h_{a,b}\) vanish when \(a+b>R\). The mask state is
\begin{equation}
\begin{aligned}
|\Psi_H\rangle
&=
\frac{1}{\sqrt{C_H}}
\sum_{\substack{0\leq a\leq m_1,\ 0\leq b\leq m_2\\a+b\leq R}}
q^{(n-a-b)/2}h_{a,b}
\sqrt{\binom{m_1}{a}\binom{m_2}{b}}
|D_{m_1,a}\rangle|D_{m_2,b}\rangle,\\
C_H
&=
\sum_{\substack{0\leq a\leq m_1,\ 0\leq b\leq m_2\\a+b\leq R}}
q^{n-a-b}|h_{a,b}|^2
\binom{m_1}{a}\binom{m_2}{b}.
\end{aligned}
\end{equation}
Applying \(\bigotimes_i\Gamma_i\) converts this mask state into the error register state for \(H(Q)=P(f(Q))J(Q)\). For this choice of \(H\), the normalized encoded state is the state denoted by \(|W_{\rm SFS\text{-}DQI}\rangle\) in Methods. Equation~\eqref{eq:sm-opi-amplitude-proof} then gives final amplitude \(P(f(Q))J(Q)/\sqrt{C_H}\).

The preparation begins by placing the amplitudes associated with all admissible weight pairs \((a,b)\) on the packed representatives \(c_{m_1,a}\sqcup c_{m_2,b}\). Their nearest-neighbor graph is connected and contains \(O(R^2)\) vertices. Givens rotations along a spanning tree set the required amplitude magnitudes, and phases controlled by \((a,b)\) supply the phases of \(h_{a,b}\).

These rotations and phases have constant local cost within the packed representative subspace. The two bits adjacent to a \(01\) boundary identify an interior weight, while one bit suffices at an endpoint. A rotation between neighboring weights changes only their common boundary bit. The adjacent bits in that group isolate the pair, and the boundary bits in the other group fix the unchanged weight. Each operation uses one support bit and at most four controls under the local-update convention of Sec.~\ref{subsec:sm-controlled-updates}.

The fixed Hamming weight construction of Sec.~\ref{subsec:sm-hamming} then maps each packed representative to \(|D_{m_1,a}\rangle|D_{m_2,b}\rangle\). The refinements for the two groups act on disjoint registers and run in parallel. Set \(R_g=\min(R,m_g)\). The preparation depth \(D_{\rm sup}\) and gate count \(G_{\rm sup}\) satisfy
\begin{equation}
\begin{aligned}
D_{\rm sup}
&=
O\!\left(
R^2+
\max_{g=1,2}
\left[
R_g\log\frac{m_g}{R_g}+R_g
\right]
\right),\\
G_{\rm sup}
&=
O\!\left(R^2+R_1m_1+R_2m_2\right).
\end{aligned}
\label{eq:sm-opi-support-resources}
\end{equation}
For comparable group sizes and \(R=\Theta(m)\), both bounds are \(O(m^2)\). Equation~\eqref{eq:sm-balanced-hab} computes all coefficients with \(O(R^2)\) classical arithmetic operations at the chosen precision. The unitaries \(\Gamma_i\) are the same as in standard OPI DQI.

Decodability follows from the distance of the dual code. The kernel \(\ker(B^T)=\{\bm e\in\mathbb F_q^m\mid B^T\bm e=0\}\) is the dual generalized Reed--Solomon code and has minimum distance \(n+1\). Two errors of weight at most \(R\) with the same syndrome differ by a codeword of weight at most \(2R\). The syndrome map is injective when \(2R<n+1\). For \(n\geq5\), choose
\begin{equation}
\ell=\left\lfloor\frac{n-5}{2}\right\rfloor,
\qquad
R=\ell+2=\left\lfloor\frac{n-1}{2}\right\rfloor,
\qquad
2R<n+1,
\qquad
2\ell+5\leq n.
\label{eq:sm-reserved-degree}
\end{equation}
A Reed--Solomon syndrome decoder then recovers every error in polynomial time and admits a reversible implementation with polynomial overhead. As in standard OPI DQI, each \(F_i\) is assumed to be explicit or efficiently accessible so that \(\Gamma_i\) has polynomial cost.

\subsection{Output reweighting and score improvement}
\label{subsec:sm-dqi-reweighting}

Let \(\mathcal Q_n=\{Q\in\mathbb F_q[z]\mid\deg Q<n\}\). We write \(\langle\cdot\rangle_{\rm DQI}\) and \(\langle\cdot\rangle_{\rm SFS\text{-}DQI}\) for expectations under the DQI and SFS-DQI output distributions, respectively. We assume \(P(f(Q))J(Q)\) is not identically zero, which is equivalent to \(\langle J^2\rangle_{\rm DQI}>0\).

Relative to DQI, SFS-DQI multiplies the unnormalized amplitude of each polynomial \(Q\) by \(J(Q)\). Its expected score therefore satisfies
\begin{equation}
\langle f\rangle_{\rm SFS\text{-}DQI}
=
\frac{\langle fJ^2\rangle_{\rm DQI}}
{\langle J^2\rangle_{\rm DQI}}.
\end{equation}
For every score layer \(\mathcal L_t=\{Q\in\mathcal Q_n\mid f(Q)=t\}\) with nonzero SFS-DQI probability, conditioning removes the dependence on \(P\). The resulting probability of sampling \(Q\in\mathcal L_t\) is
\begin{equation}
\frac{J(Q)^2}{\sum_{Q'\in\mathcal L_t}J(Q')^2}.
\label{eq:sm-conditional-layer}
\end{equation}
The covariance between \(f\) and \(J^2\) under the DQI output is
\[
\operatorname{Cov}_{\rm DQI}(f,J^2)
=
\langle fJ^2\rangle_{\rm DQI}
-\langle f\rangle_{\rm DQI}\langle J^2\rangle_{\rm DQI}.
\]
The shift in expected score is
\begin{equation}
\langle f\rangle_{\rm SFS\text{-}DQI}
-
\langle f\rangle_{\rm DQI}
=
\frac{\operatorname{Cov}_{\rm DQI}(f,J^2)}
{\langle J^2\rangle_{\rm DQI}}.
\label{eq:sm-score-covariance}
\end{equation}
SFS-DQI preserves or increases the expected score exactly when \(\operatorname{Cov}_{\rm DQI}(f,J^2)\geq0\). For the two-group preference, this covariance is strictly positive.

\begin{lemma}[Strict improvement for grouped OPI]
\label{lem:sm-opi-score-preference}
Let \(n\geq5\), let \(0<r<q\), and let \(G_1\) and \(G_2\) be nonempty. For every nonzero \(P\) of degree at most \(\ell=\lfloor(n-5)/2\rfloor\), the preference factor \(J=S_1S_2\) satisfies
\begin{equation}
\operatorname{Cov}_{\rm DQI}(f,J^2)>0.
\label{eq:sm-opi-positive-covariance}
\end{equation}
SFS-DQI has strictly larger expected values of \(f\) and \(J\) than score-only DQI with the same \(P\). Within every score layer with nonzero SFS-DQI probability, the conditional expected value of \(J\) does not decrease. It increases strictly unless \(J\) is constant on that layer.
\end{lemma}

\begin{proof}
Let \(X_i=f_i(Q)\) for a uniformly sampled \(Q\in\mathcal Q_n\), and set \(\mathcal W=|P(f)|^2\). Evaluations at any set of at most \(n\) distinct points are independent and uniform over \(\mathbb F_q\). The variables \(X_i\) are \(n\)-wise independent Bernoulli variables with mean \(r/q\).

After multilinear reduction in the \(X_i\), the degrees of \(\mathcal W\), \(f\mathcal W\), \(J^2\mathcal W\) and \(fJ^2\mathcal W\) are at most \(2\ell\), \(2\ell+1\), \(2\ell+4\) and \(2\ell+5\). By Eq.~\eqref{eq:sm-reserved-degree}, all four degrees are at most \(n\). Their expectations equal those for fully independent Bernoulli variables with the same mean. These four moments determine \(\operatorname{Cov}_{\rm DQI}(f,J^2)\).

In the fully independent model, conditioning on \(f=t\) makes the satisfied positions a uniform \(t\)-subset of \([m]\). Couple this subset to a uniform \((t+1)\)-subset by adding a uniformly chosen unsatisfied position. Adding a position in \(G_1\) changes \(J\) by \(S_2\), while adding one in \(G_2\) changes \(J\) by \(S_1\). Hence \(h(t)=\langle J^2\mid f=t\rangle\) satisfies \(h(0)=h(1)=0\) and increases strictly for \(1\leq t<m\).

The weight \(\mathcal W\) depends only on \(f\), so it does not change the conditional distribution at fixed \(t\). In this model, every score has positive probability before reweighting. A nonzero polynomial \(P\) has at most \(\ell<m-1\) roots among the \(m+1\) scores. The reweighted distribution contains at least three score values on which \(h\) is not constant. If \(f'\) is an independent copy under this distribution, then
\[
\operatorname{Cov}(f,h(f))
=\frac{1}{2}\left\langle(f-f')(h(f)-h(f'))\right\rangle>0.
\]
Conditioning on \(f\) identifies this covariance with \(\operatorname{Cov}(f,J^2)\). The moment equalities transfer the strict inequality to the DQI output distribution, proving Eq.~\eqref{eq:sm-opi-positive-covariance}.

Equation~\eqref{eq:sm-score-covariance} now proves the strict increase in expected score. For the preference factor, let \(J'\) be an independent copy under the DQI output distribution. Then
\[
\begin{aligned}
\langle J\rangle_{\rm SFS\text{-}DQI}-\langle J\rangle_{\rm DQI}
&=\frac{\operatorname{Cov}_{\rm DQI}(J,J^2)}{\langle J^2\rangle_{\rm DQI}},\\
\operatorname{Cov}_{\rm DQI}(J,J^2)
&=\frac{1}{2}\left\langle(J-J')^2(J+J')\right\rangle_{\rm DQI}.
\end{aligned}
\]
The last expression is positive because \(J\geq0\) and Eq.~\eqref{eq:sm-opi-positive-covariance} rules out a constant \(J\). The same identity after conditioning on a score layer proves the final claim.
\end{proof}

The strict improvement concerns the expected score. It does not imply that the probability of exceeding every score threshold increases.

\subsection{Asymptotic benchmark and speed-up}
\label{subsec:sm-dqi-speedup-inheritance}

Equation~\eqref{eq:sm-conditional-layer} establishes the \(J(Q)^2\) preference within each score layer, while Lemma~\ref{lem:sm-opi-score-preference} proves a larger expected score for the same \(P\). The degree comparison below relates this degree-reduced choice of \(P\) to the standard asymptotic choice.

The asymptotic OPI DQI score benchmark is expressed in terms of the normalized degree \(\mu=\ell/m\) of \(P\) as~\cite{jordan2025decoded}
\begin{equation}
\rho(\alpha,\mu)
=
\left(
\sqrt{\mu(1-\alpha)}
+\sqrt{\alpha(1-\mu)}
\right)^2
\label{eq:sm-opi-semicircle-benchmark}
\end{equation}
when \(\alpha\leq1-\mu\), and it equals one otherwise. Let \(L_0=\lfloor(n-1)/2\rfloor\) be the largest integer degree allowed by the standard OPI DQI expected-score theorem, satisfying \(2L_0+1<n+1\). Set \(\mu_0=L_0/m\) and \(\mu_\ell=\ell/m\). Equation~\eqref{eq:sm-reserved-degree} gives \(L_0-\ell=2\). When \(r/q\to1/2\) and \(n/q\to\nu\in(0,1)\), both \((\alpha,\mu_0)\) and \((\alpha,\mu_\ell)\) eventually lie in the first branch of Eq.~\eqref{eq:sm-opi-semicircle-benchmark}. The mean value theorem gives \(\rho(\alpha,\mu_0)-\rho(\alpha,\mu_\ell)=O(1/m)\).
The two-degree reduction accommodates the quadratic preference, which raises the error-weight bound from \(\ell\) to \(R=\ell+2=L_0\). At the level of the expected-score moments, \(J^2\) adds four degrees and the same reduction gives \(2\ell+5=2L_0+1\leq n\). The reduction lowers the asymptotic DQI benchmark by \(O(1/m)\), which vanishes as \(m\to\infty\).

\begin{corollary}[Asymptotic guarantee for grouped OPI]
\label{cor:sm-dqi-speedup-inheritance}
Grouped OPI asks for the standard OPI requirement on the asymptotic expected score. Within every score layer with nonzero output probability, it also requires conditional sampling probabilities proportional to \(J(Q)^2\). Consider an OPI family with \(m=q-1\), \(r/q\to1/2\), \(n/q\to\nu\in(0,1)\) and two groups of comparable size. Let \(P\) be the optimal degree-\(\ell\) score-weighting function for DQI, with \(\ell\) given in Eq.~\eqref{eq:sm-reserved-degree}. Assume that each \(F_i\) is efficiently accessible. Then
\begin{equation}
\liminf_{q\to\infty}
\frac{\langle f\rangle_{\rm SFS\text{-}DQI}}{m}
\geq
\frac{1}{2}
+\sqrt{\frac{\nu}{2}\left(1-\frac{\nu}{2}\right)}.
\label{eq:sm-opi-asymptotic-score}
\end{equation}
Within every score layer with nonzero output probability, SFS-DQI assigns conditional probabilities in proportion to \(J(Q)^2\) and therefore realizes the balance preference. Under the standard DQI comparison assumption, the reported superpolynomial speed-up over known classical algorithms extends to grouped OPI.
\end{corollary}

\begin{proof}
The reserved degree satisfies \(\ell/m\to\nu/2\). Equation~\eqref{eq:sm-opi-semicircle-benchmark} gives the right-hand side of Eq.~\eqref{eq:sm-opi-asymptotic-score} for score-only DQI with \(P\) of degree \(\ell\). Lemma~\ref{lem:sm-opi-score-preference} shows that SFS-DQI strictly increases that expected score. Equation~\eqref{eq:sm-opi-support-resources} gives polynomial preparation cost, and Eq.~\eqref{eq:sm-reserved-degree} preserves polynomial Reed--Solomon decoding.

The reported classical Prange benchmark is \(\rho_{\rm C}(\nu)=1/2+\nu/2\). For every fixed \(0<\nu<1\), the gap \(\sqrt{\nu(1-\nu/2)/2}-\nu/2\) is constant and positive. Jordan et al.~\cite{jordan2025decoded} report a superpolynomial speed-up for standard OPI DQI under the assumption that no classical polynomial-time algorithm reaches the asymptotic DQI score benchmark.

A fixed partition into two groups of comparable size can be added to any OPI instance without changing its score. Any classical algorithm meeting both grouped OPI requirements also reaches the original OPI score benchmark after the partition is ignored. Grouped OPI is therefore at least as hard as the original score task under this benchmark. Together with the polynomial SFS overhead, the same comparison assumption gives the stated extension.
\end{proof}

\subsection{Other structural preferences}

More generally, set \(H(Q)=P(f(Q))J(Q)\), where \(J\) encodes a structural preference. This product is represented by the finite-field coefficient convolution described in Methods. Let \(E_{\rm conf}\) be a set of unordered conflict pairs, let \(E_{\rm pre}\) be a set of ordered precedence pairs, and let \(\mathcal T_{\rm syn}\) be a collection of nonempty constraint subsets assigned synergy rewards. The corresponding preference contributions are
\begin{align*}
J_{\rm conf}(Q)
&=
-\sum_{\{i,j\}\in E_{\rm conf}}
f_i(Q)f_j(Q),\\
J_{\rm pre}(Q)
&=
-\sum_{(i,j)\in E_{\rm pre}}
\bigl(1-f_i(Q)\bigr)f_j(Q),\\
J_{\rm syn}(Q)
&=
\sum_{T\in\mathcal T_{\rm syn}}
\prod_{i\in T}f_i(Q).
\end{align*}
Here \((i,j)\in E_{\rm pre}\) means that satisfying constraint \(j\) without satisfying constraint \(i\) is a precedence violation. The conflict score decreases when conflicting constraints are satisfied together. The precedence score decreases with each violation, while the synergy score increases when every constraint in a rewarded subset is satisfied.

The preference factor applied in SFS-DQI is

\begin{equation}
J(Q)=c+J_{\rm conf}(Q)+J_{\rm pre}(Q)+J_{\rm syn}(Q)\geq0
\quad\text{for every polynomial }Q.
\end{equation}
For any sum of the unit-weight terms above, \(c=|E_{\rm conf}|+|E_{\rm pre}|\) is a simple sufficient choice. Terms not included in a chosen preference are omitted. Within each fixed-score layer, reweighting by \(J(Q)^2\) therefore preserves the ordering defined by the chosen contributions. Adding \(c\) does not increase the polynomial degree. If \(J\) has degree \(d_J\) in the indicators \(f_i\), multilinear reduction gives the mask-weight bound \(R\leq\ell+d_J\). The syndrome decoder remains valid provided \(2R<n+1\). For SFS preparation to have polynomial cost, the resulting mask state must also admit a polynomial-size recursive description. The two-group balance preference above is the explicit case for which preparation, decoding and preservation of the DQI score benchmark are established.

\renewcommand{\bibsection}{\section*{Supplementary References}}
\bibliographycontrols
\bibliographystyle{apsrev4-2}
\bibliography{refs}

\begin{thebibliography}{31}%
\makeatletter
\providecommand \@ifxundefined [1]{%
 \@ifx{#1\undefined}
}%
\providecommand \@ifnum [1]{%
 \ifnum #1\expandafter \@firstoftwo
 \else \expandafter \@secondoftwo
 \fi
}%
\providecommand \@ifx [1]{%
 \ifx #1\expandafter \@firstoftwo
 \else \expandafter \@secondoftwo
 \fi
}%
\providecommand \natexlab [1]{#1}%
\providecommand \enquote  [1]{``#1''}%
\providecommand \bibnamefont  [1]{#1}%
\providecommand \bibfnamefont [1]{#1}%
\providecommand \citenamefont [1]{#1}%
\providecommand \href@noop [0]{\@secondoftwo}%
\providecommand \href [0]{\begingroup \@sanitize@url \@href}%
\providecommand \@href[1]{\@@startlink{#1}\@@href}%
\providecommand \@@href[1]{\endgroup#1\@@endlink}%
\providecommand \@sanitize@url [0]{\catcode `\\12\catcode `\$12\catcode `\&12\catcode `\#12\catcode `\^12\catcode `\_12\catcode `\%12\relax}%
\providecommand \@@startlink[1]{}%
\providecommand \@@endlink[0]{}%
\providecommand \url  [0]{\begingroup\@sanitize@url \@url }%
\providecommand \@url [1]{\endgroup\@href {#1}{\urlprefix }}%
\providecommand \urlprefix  [0]{URL }%
\providecommand \Eprint [0]{\href }%
\providecommand \doibase [0]{https://doi.org/}%
\providecommand \selectlanguage [0]{\@gobble}%
\providecommand \bibinfo  [0]{\@secondoftwo}%
\providecommand \bibfield  [0]{\@secondoftwo}%
\providecommand \translation [1]{[#1]}%
\providecommand \BibitemOpen [0]{}%
\providecommand \bibitemStop [0]{}%
\providecommand \bibitemNoStop [0]{.\EOS\space}%
\providecommand \EOS [0]{\spacefactor3000\relax}%
\providecommand \BibitemShut  [1]{\csname bibitem#1\endcsname}%
\let\auto@bib@innerbib\@empty
\bibitem [{\citenamefont {Shor}(1994)}]{shor1994algorithms}%
  \BibitemOpen
  \bibfield  {author} {\bibinfo {author} {\bibfnamefont {P.~W.}\ \bibnamefont {Shor}},\ }\bibfield  {title} {\bibinfo {title} {Algorithms for quantum computation: discrete logarithms and factoring},\ }in\ \href@noop {} {\emph {\bibinfo {booktitle} {Proceedings 35th annual symposium on foundations of computer science}}}\ (\bibinfo {organization} {Ieee},\ \bibinfo {year} {1994})\ pp.\ \bibinfo {pages} {124--134}\BibitemShut {NoStop}%
\bibitem [{\citenamefont {Grover}(1996)}]{grover1996fast}%
  \BibitemOpen
  \bibfield  {author} {\bibinfo {author} {\bibfnamefont {L.~K.}\ \bibnamefont {Grover}},\ }\bibfield  {title} {\bibinfo {title} {A fast quantum mechanical algorithm for database search},\ }in\ \href@noop {} {\emph {\bibinfo {booktitle} {Proceedings of the twenty-eighth annual ACM symposium on Theory of computing}}}\ (\bibinfo {year} {1996})\ pp.\ \bibinfo {pages} {212--219}\BibitemShut {NoStop}%
\bibitem [{\citenamefont {Jordan}\ \emph {et~al.}(2025)\citenamefont {Jordan}, \citenamefont {Shutty}, \citenamefont {Wootters}, \citenamefont {Zalcman}, \citenamefont {Schmidhuber}, \citenamefont {King}, \citenamefont {Isakov}, \citenamefont {Khattar},\ and\ \citenamefont {Babbush}}]{jordan2025decoded}%
  \BibitemOpen
  \bibfield  {author} {\bibinfo {author} {\bibfnamefont {S.~P.}\ \bibnamefont {Jordan}}, \bibinfo {author} {\bibfnamefont {N.}~\bibnamefont {Shutty}}, \bibinfo {author} {\bibfnamefont {M.}~\bibnamefont {Wootters}}, \bibinfo {author} {\bibfnamefont {A.}~\bibnamefont {Zalcman}}, \bibinfo {author} {\bibfnamefont {A.}~\bibnamefont {Schmidhuber}}, \bibinfo {author} {\bibfnamefont {R.}~\bibnamefont {King}}, \bibinfo {author} {\bibfnamefont {S.~V.}\ \bibnamefont {Isakov}}, \bibinfo {author} {\bibfnamefont {T.}~\bibnamefont {Khattar}},\ and\ \bibinfo {author} {\bibfnamefont {R.}~\bibnamefont {Babbush}},\ }\bibfield  {title} {\bibinfo {title} {Optimization by decoded quantum interferometry},\ }\href@noop {} {\bibfield  {journal} {\bibinfo  {journal} {Nature}\ }\textbf {\bibinfo {volume} {646}},\ \bibinfo {pages} {831} (\bibinfo {year} {2025})}\BibitemShut {NoStop}%
\bibitem [{\citenamefont {Stoudenmire}\ and\ \citenamefont {Waintal}(2024)}]{stoudenmire2024opening}%
  \BibitemOpen
  \bibfield  {author} {\bibinfo {author} {\bibfnamefont {E.}~\bibnamefont {Stoudenmire}}\ and\ \bibinfo {author} {\bibfnamefont {X.}~\bibnamefont {Waintal}},\ }\bibfield  {title} {\bibinfo {title} {Opening the black box inside grover’s algorithm},\ }\href@noop {} {\bibfield  {journal} {\bibinfo  {journal} {Physical Review X}\ }\textbf {\bibinfo {volume} {14}},\ \bibinfo {pages} {041029} (\bibinfo {year} {2024})}\BibitemShut {NoStop}%
\bibitem [{\citenamefont {Biamonte}\ \emph {et~al.}(2017)\citenamefont {Biamonte}, \citenamefont {Wittek}, \citenamefont {Pancotti}, \citenamefont {Rebentrost}, \citenamefont {Wiebe},\ and\ \citenamefont {Lloyd}}]{biamonte2017quantum}%
  \BibitemOpen
  \bibfield  {author} {\bibinfo {author} {\bibfnamefont {J.}~\bibnamefont {Biamonte}}, \bibinfo {author} {\bibfnamefont {P.}~\bibnamefont {Wittek}}, \bibinfo {author} {\bibfnamefont {N.}~\bibnamefont {Pancotti}}, \bibinfo {author} {\bibfnamefont {P.}~\bibnamefont {Rebentrost}}, \bibinfo {author} {\bibfnamefont {N.}~\bibnamefont {Wiebe}},\ and\ \bibinfo {author} {\bibfnamefont {S.}~\bibnamefont {Lloyd}},\ }\bibfield  {title} {\bibinfo {title} {Quantum machine learning},\ }\href@noop {} {\bibfield  {journal} {\bibinfo  {journal} {Nature}\ }\textbf {\bibinfo {volume} {549}},\ \bibinfo {pages} {195} (\bibinfo {year} {2017})}\BibitemShut {NoStop}%
\bibitem [{\citenamefont {B{\"a}rtschi}\ and\ \citenamefont {Eidenbenz}(2019)}]{bartschi2019deterministic}%
  \BibitemOpen
  \bibfield  {author} {\bibinfo {author} {\bibfnamefont {A.}~\bibnamefont {B{\"a}rtschi}}\ and\ \bibinfo {author} {\bibfnamefont {S.}~\bibnamefont {Eidenbenz}},\ }\bibfield  {title} {\bibinfo {title} {Deterministic preparation of dicke states},\ }in\ \href@noop {} {\emph {\bibinfo {booktitle} {International Symposium on Fundamentals of Computation Theory}}}\ (\bibinfo {organization} {Springer},\ \bibinfo {year} {2019})\ pp.\ \bibinfo {pages} {126--139}\BibitemShut {NoStop}%
\bibitem [{\citenamefont {B{\"a}rtschi}\ and\ \citenamefont {Eidenbenz}(2022)}]{bartschi2022short}%
  \BibitemOpen
  \bibfield  {author} {\bibinfo {author} {\bibfnamefont {A.}~\bibnamefont {B{\"a}rtschi}}\ and\ \bibinfo {author} {\bibfnamefont {S.}~\bibnamefont {Eidenbenz}},\ }\bibfield  {title} {\bibinfo {title} {Short-depth circuits for dicke state preparation},\ }in\ \href@noop {} {\emph {\bibinfo {booktitle} {2022 IEEE International Conference on Quantum Computing and Engineering (QCE)}}}\ (\bibinfo {organization} {IEEE},\ \bibinfo {year} {2022})\ pp.\ \bibinfo {pages} {87--96}\BibitemShut {NoStop}%
\bibitem [{\citenamefont {Gleinig}\ and\ \citenamefont {Hoefler}(2021)}]{gleinig2021efficient}%
  \BibitemOpen
  \bibfield  {author} {\bibinfo {author} {\bibfnamefont {N.}~\bibnamefont {Gleinig}}\ and\ \bibinfo {author} {\bibfnamefont {T.}~\bibnamefont {Hoefler}},\ }\bibfield  {title} {\bibinfo {title} {An efficient algorithm for sparse quantum state preparation},\ }in\ \href@noop {} {\emph {\bibinfo {booktitle} {2021 58th ACM/IEEE Design Automation Conference (DAC)}}}\ (\bibinfo {organization} {IEEE},\ \bibinfo {year} {2021})\ pp.\ \bibinfo {pages} {433--438}\BibitemShut {NoStop}%
\bibitem [{\citenamefont {Mozafari}\ \emph {et~al.}(2022)\citenamefont {Mozafari}, \citenamefont {De~Micheli},\ and\ \citenamefont {Yang}}]{mozafari2022decision}%
  \BibitemOpen
  \bibfield  {author} {\bibinfo {author} {\bibfnamefont {F.}~\bibnamefont {Mozafari}}, \bibinfo {author} {\bibfnamefont {G.}~\bibnamefont {De~Micheli}},\ and\ \bibinfo {author} {\bibfnamefont {Y.}~\bibnamefont {Yang}},\ }\bibfield  {title} {\bibinfo {title} {Efficient deterministic preparation of quantum states using decision diagrams},\ }\href@noop {} {\bibfield  {journal} {\bibinfo  {journal} {Physical Review A}\ }\textbf {\bibinfo {volume} {106}},\ \bibinfo {pages} {022617} (\bibinfo {year} {2022})}\BibitemShut {NoStop}%
\bibitem [{\citenamefont {Araujo}\ \emph {et~al.}(2023)\citenamefont {Araujo}, \citenamefont {Blank}, \citenamefont {Ara{\'u}jo},\ and\ \citenamefont {da~Silva}}]{araujo2024lowrank}%
  \BibitemOpen
  \bibfield  {author} {\bibinfo {author} {\bibfnamefont {I.~F.}\ \bibnamefont {Araujo}}, \bibinfo {author} {\bibfnamefont {C.}~\bibnamefont {Blank}}, \bibinfo {author} {\bibfnamefont {I.~C.}\ \bibnamefont {Ara{\'u}jo}},\ and\ \bibinfo {author} {\bibfnamefont {A.~J.}\ \bibnamefont {da~Silva}},\ }\bibfield  {title} {\bibinfo {title} {Low-rank quantum state preparation},\ }\href@noop {} {\bibfield  {journal} {\bibinfo  {journal} {IEEE Transactions on Computer-Aided Design of Integrated Circuits and Systems}\ }\textbf {\bibinfo {volume} {43}},\ \bibinfo {pages} {161} (\bibinfo {year} {2023})}\BibitemShut {NoStop}%
\bibitem [{\citenamefont {Malz}\ \emph {et~al.}(2024)\citenamefont {Malz}, \citenamefont {Styliaris}, \citenamefont {Wei},\ and\ \citenamefont {Cirac}}]{malz2024mps}%
  \BibitemOpen
  \bibfield  {author} {\bibinfo {author} {\bibfnamefont {D.}~\bibnamefont {Malz}}, \bibinfo {author} {\bibfnamefont {G.}~\bibnamefont {Styliaris}}, \bibinfo {author} {\bibfnamefont {Z.-Y.}\ \bibnamefont {Wei}},\ and\ \bibinfo {author} {\bibfnamefont {J.~I.}\ \bibnamefont {Cirac}},\ }\bibfield  {title} {\bibinfo {title} {Preparation of matrix product states with log-depth quantum circuits},\ }\href@noop {} {\bibfield  {journal} {\bibinfo  {journal} {Physical Review Letters}\ }\textbf {\bibinfo {volume} {132}},\ \bibinfo {pages} {040404} (\bibinfo {year} {2024})}\BibitemShut {NoStop}%
\bibitem [{\citenamefont {Zhang}\ \emph {et~al.}(2026)\citenamefont {Zhang}, \citenamefont {Rattew}, \citenamefont {Wu}, \citenamefont {Styliaris}, \citenamefont {Sun}, \citenamefont {Koczor},\ and\ \citenamefont {Yuan}}]{zhang2026bits}%
  \BibitemOpen
  \bibfield  {author} {\bibinfo {author} {\bibfnamefont {X.-M.}\ \bibnamefont {Zhang}}, \bibinfo {author} {\bibfnamefont {A.~G.}\ \bibnamefont {Rattew}}, \bibinfo {author} {\bibfnamefont {B.}~\bibnamefont {Wu}}, \bibinfo {author} {\bibfnamefont {G.}~\bibnamefont {Styliaris}}, \bibinfo {author} {\bibfnamefont {X.}~\bibnamefont {Sun}}, \bibinfo {author} {\bibfnamefont {B.}~\bibnamefont {Koczor}},\ and\ \bibinfo {author} {\bibfnamefont {X.}~\bibnamefont {Yuan}},\ }\bibfield  {title} {\bibinfo {title} {From bits to qubits: The theory and practice of quantum data encoding},\ }\href@noop {} {\bibfield  {journal} {\bibinfo  {journal} {arXiv preprint arXiv:2609.08058}\ } (\bibinfo {year} {2026})}\BibitemShut {NoStop}%
\bibitem [{\citenamefont {Bernien}\ \emph {et~al.}(2017)\citenamefont {Bernien}, \citenamefont {Schwartz}, \citenamefont {Keesling}, \citenamefont {Levine}, \citenamefont {Omran}, \citenamefont {Pichler}, \citenamefont {Choi}, \citenamefont {Zibrov}, \citenamefont {Endres}, \citenamefont {Greiner} \emph {et~al.}}]{bernien2017probing}%
  \BibitemOpen
  \bibfield  {author} {\bibinfo {author} {\bibfnamefont {H.}~\bibnamefont {Bernien}}, \bibinfo {author} {\bibfnamefont {S.}~\bibnamefont {Schwartz}}, \bibinfo {author} {\bibfnamefont {A.}~\bibnamefont {Keesling}}, \bibinfo {author} {\bibfnamefont {H.}~\bibnamefont {Levine}}, \bibinfo {author} {\bibfnamefont {A.}~\bibnamefont {Omran}}, \bibinfo {author} {\bibfnamefont {H.}~\bibnamefont {Pichler}}, \bibinfo {author} {\bibfnamefont {S.}~\bibnamefont {Choi}}, \bibinfo {author} {\bibfnamefont {A.~S.}\ \bibnamefont {Zibrov}}, \bibinfo {author} {\bibfnamefont {M.}~\bibnamefont {Endres}}, \bibinfo {author} {\bibfnamefont {M.}~\bibnamefont {Greiner}}, \emph {et~al.},\ }\bibfield  {title} {\bibinfo {title} {Probing many-body dynamics on a 51-atom quantum simulator},\ }\href@noop {} {\bibfield  {journal} {\bibinfo  {journal} {Nature}\ }\textbf {\bibinfo {volume} {551}},\ \bibinfo {pages} {579} (\bibinfo {year} {2017})}\BibitemShut {NoStop}%
\bibitem [{\citenamefont {Ebadi}\ \emph {et~al.}(2022)\citenamefont {Ebadi}, \citenamefont {Keesling}, \citenamefont {Cain}, \citenamefont {Wang}, \citenamefont {Levine}, \citenamefont {Bluvstein}, \citenamefont {Semeghini}, \citenamefont {Omran}, \citenamefont {Liu}, \citenamefont {Samajdar} \emph {et~al.}}]{ebadi2022quantum}%
  \BibitemOpen
  \bibfield  {author} {\bibinfo {author} {\bibfnamefont {S.}~\bibnamefont {Ebadi}}, \bibinfo {author} {\bibfnamefont {A.}~\bibnamefont {Keesling}}, \bibinfo {author} {\bibfnamefont {M.}~\bibnamefont {Cain}}, \bibinfo {author} {\bibfnamefont {T.~T.}\ \bibnamefont {Wang}}, \bibinfo {author} {\bibfnamefont {H.}~\bibnamefont {Levine}}, \bibinfo {author} {\bibfnamefont {D.}~\bibnamefont {Bluvstein}}, \bibinfo {author} {\bibfnamefont {G.}~\bibnamefont {Semeghini}}, \bibinfo {author} {\bibfnamefont {A.}~\bibnamefont {Omran}}, \bibinfo {author} {\bibfnamefont {J.-G.}\ \bibnamefont {Liu}}, \bibinfo {author} {\bibfnamefont {R.}~\bibnamefont {Samajdar}}, \emph {et~al.},\ }\bibfield  {title} {\bibinfo {title} {Quantum optimization of maximum independent set using rydberg atom arrays},\ }\href@noop {} {\bibfield  {journal} {\bibinfo  {journal} {Science}\ }\textbf {\bibinfo {volume} {376}},\ \bibinfo {pages} {1209} (\bibinfo {year} {2022})}\BibitemShut {NoStop}%
\bibitem [{\citenamefont {Pardo}\ \emph {et~al.}(2023)\citenamefont {Pardo}, \citenamefont {Greenberg}, \citenamefont {Fortinsky}, \citenamefont {Katz},\ and\ \citenamefont {Zohar}}]{pardo2023resource}%
  \BibitemOpen
  \bibfield  {author} {\bibinfo {author} {\bibfnamefont {G.}~\bibnamefont {Pardo}}, \bibinfo {author} {\bibfnamefont {T.}~\bibnamefont {Greenberg}}, \bibinfo {author} {\bibfnamefont {A.}~\bibnamefont {Fortinsky}}, \bibinfo {author} {\bibfnamefont {N.}~\bibnamefont {Katz}},\ and\ \bibinfo {author} {\bibfnamefont {E.}~\bibnamefont {Zohar}},\ }\bibfield  {title} {\bibinfo {title} {Resource-efficient quantum simulation of lattice gauge theories in arbitrary dimensions: Solving for gauss's law and fermion elimination},\ }\href@noop {} {\bibfield  {journal} {\bibinfo  {journal} {Physical Review Research}\ }\textbf {\bibinfo {volume} {5}},\ \bibinfo {pages} {023077} (\bibinfo {year} {2023})}\BibitemShut {NoStop}%
\bibitem [{\citenamefont {Sharma}\ and\ \citenamefont {Mueller}(2024)}]{sharma2024gauss}%
  \BibitemOpen
  \bibfield  {author} {\bibinfo {author} {\bibfnamefont {V.}~\bibnamefont {Sharma}}\ and\ \bibinfo {author} {\bibfnamefont {E.~J.}\ \bibnamefont {Mueller}},\ }\bibfield  {title} {\bibinfo {title} {One-dimensional z 2 lattice gauge theory in periodic gauss-law sectors},\ }\href@noop {} {\bibfield  {journal} {\bibinfo  {journal} {Physical Review A}\ }\textbf {\bibinfo {volume} {110}},\ \bibinfo {pages} {033314} (\bibinfo {year} {2024})}\BibitemShut {NoStop}%
\bibitem [{\citenamefont {Li}\ and\ \citenamefont {Luo}(2024)}]{li2024nearly}%
  \BibitemOpen
  \bibfield  {author} {\bibinfo {author} {\bibfnamefont {L.}~\bibnamefont {Li}}\ and\ \bibinfo {author} {\bibfnamefont {J.}~\bibnamefont {Luo}},\ }\bibfield  {title} {\bibinfo {title} {Nearly optimal circuit size for sparse quantum state preparation},\ }\href@noop {} {\bibfield  {journal} {\bibinfo  {journal} {arXiv preprint arXiv:2406.16142}\ } (\bibinfo {year} {2024})}\BibitemShut {NoStop}%
\bibitem [{\citenamefont {Mao}\ \emph {et~al.}(2024)\citenamefont {Mao}, \citenamefont {Tian},\ and\ \citenamefont {Sun}}]{mao2024toward}%
  \BibitemOpen
  \bibfield  {author} {\bibinfo {author} {\bibfnamefont {R.}~\bibnamefont {Mao}}, \bibinfo {author} {\bibfnamefont {G.}~\bibnamefont {Tian}},\ and\ \bibinfo {author} {\bibfnamefont {X.}~\bibnamefont {Sun}},\ }\bibfield  {title} {\bibinfo {title} {Toward optimal circuit size for sparse quantum state preparation},\ }\href@noop {} {\bibfield  {journal} {\bibinfo  {journal} {Physical Review A}\ }\textbf {\bibinfo {volume} {110}},\ \bibinfo {pages} {032439} (\bibinfo {year} {2024})}\BibitemShut {NoStop}%
\bibitem [{\citenamefont {Melo}\ \emph {et~al.}(2009)\citenamefont {Melo}, \citenamefont {Nickel},\ and\ \citenamefont {Saldanha-Da-Gama}}]{melo2009facility}%
  \BibitemOpen
  \bibfield  {author} {\bibinfo {author} {\bibfnamefont {M.~T.}\ \bibnamefont {Melo}}, \bibinfo {author} {\bibfnamefont {S.}~\bibnamefont {Nickel}},\ and\ \bibinfo {author} {\bibfnamefont {F.}~\bibnamefont {Saldanha-Da-Gama}},\ }\bibfield  {title} {\bibinfo {title} {Facility location and supply chain management--a review},\ }\href@noop {} {\bibfield  {journal} {\bibinfo  {journal} {European journal of operational research}\ }\textbf {\bibinfo {volume} {196}},\ \bibinfo {pages} {401} (\bibinfo {year} {2009})}\BibitemShut {NoStop}%
\bibitem [{\citenamefont {Hegerfeld}\ and\ \citenamefont {Kratsch}(2023)}]{hegerfeld2023tight}%
  \BibitemOpen
  \bibfield  {author} {\bibinfo {author} {\bibfnamefont {F.}~\bibnamefont {Hegerfeld}}\ and\ \bibinfo {author} {\bibfnamefont {S.}~\bibnamefont {Kratsch}},\ }\bibfield  {title} {\bibinfo {title} {Tight algorithms for connectivity problems parameterized by clique-width},\ }in\ \href@noop {} {\emph {\bibinfo {booktitle} {31st Annual European Symposium on Algorithms (ESA 2023)}}},\ \bibinfo {series} {Leibniz International Proceedings in Informatics (LIPIcs)}, Vol.\ \bibinfo {volume} {274}\ (\bibinfo  {publisher} {Schloss Dagstuhl -- Leibniz-Zentrum f{\"u}r Informatik},\ \bibinfo {year} {2023})\ pp.\ \bibinfo {pages} {59:1--59:19}\BibitemShut {NoStop}%
\bibitem [{\citenamefont {Durr}\ and\ \citenamefont {Hoyer}(1996)}]{durr1996quantum}%
  \BibitemOpen
  \bibfield  {author} {\bibinfo {author} {\bibfnamefont {C.}~\bibnamefont {Durr}}\ and\ \bibinfo {author} {\bibfnamefont {P.}~\bibnamefont {Hoyer}},\ }\bibfield  {title} {\bibinfo {title} {A quantum algorithm for finding the minimum},\ }\href@noop {} {\bibfield  {journal} {\bibinfo  {journal} {arXiv preprint quant-ph/9607014}\ } (\bibinfo {year} {1996})}\BibitemShut {NoStop}%
\bibitem [{\citenamefont {Gilliam}\ \emph {et~al.}(2021)\citenamefont {Gilliam}, \citenamefont {Woerner},\ and\ \citenamefont {Gonciulea}}]{gilliam2021grover}%
  \BibitemOpen
  \bibfield  {author} {\bibinfo {author} {\bibfnamefont {A.}~\bibnamefont {Gilliam}}, \bibinfo {author} {\bibfnamefont {S.}~\bibnamefont {Woerner}},\ and\ \bibinfo {author} {\bibfnamefont {C.}~\bibnamefont {Gonciulea}},\ }\bibfield  {title} {\bibinfo {title} {Grover adaptive search for constrained polynomial binary optimization},\ }\href@noop {} {\bibfield  {journal} {\bibinfo  {journal} {Quantum}\ }\textbf {\bibinfo {volume} {5}},\ \bibinfo {pages} {428} (\bibinfo {year} {2021})}\BibitemShut {NoStop}%
\bibitem [{\citenamefont {Bu}\ \emph {et~al.}(2026)\citenamefont {Bu}, \citenamefont {Gu},\ and\ \citenamefont {Li}}]{bu2026multivariate}%
  \BibitemOpen
  \bibfield  {author} {\bibinfo {author} {\bibfnamefont {K.}~\bibnamefont {Bu}}, \bibinfo {author} {\bibfnamefont {W.}~\bibnamefont {Gu}},\ and\ \bibinfo {author} {\bibfnamefont {X.}~\bibnamefont {Li}},\ }\bibfield  {title} {\bibinfo {title} {Multivariate decoded quantum interferometry for weighted optimization},\ }\href@noop {} {\bibfield  {journal} {\bibinfo  {journal} {arXiv preprint arXiv:2605.10666}\ } (\bibinfo {year} {2026})}\BibitemShut {NoStop}%
\bibitem [{\citenamefont {{OpenAI}}(2026)}]{openai2026eth}%
  \BibitemOpen
  \bibfield  {author} {\bibinfo {author} {\bibnamefont {{OpenAI}}},\ }\href@noop {} {\bibinfo {title} {{The Exponential Time Hypothesis is False}}},\ \bibinfo {howpublished} {OpenAI Math Release preprint \href{https://github.com/openai/math/blob/main/preprints/The-Exponential-Time-Hypothesis-is-False-October-25-2026/paper.pdf}{OAI:The-Exponential-Time-Hypothesis-is-False-October-25-2026}} (\bibinfo {year} {2026})\BibitemShut {NoStop}%
\bibitem [{\citenamefont {Impagliazzo}\ and\ \citenamefont {Paturi}(2001)}]{impagliazzo2001complexity}%
  \BibitemOpen
  \bibfield  {author} {\bibinfo {author} {\bibfnamefont {R.}~\bibnamefont {Impagliazzo}}\ and\ \bibinfo {author} {\bibfnamefont {R.}~\bibnamefont {Paturi}},\ }\bibfield  {title} {\bibinfo {title} {On the complexity of k-sat},\ }\href@noop {} {\bibfield  {journal} {\bibinfo  {journal} {Journal of Computer and System Sciences}\ }\textbf {\bibinfo {volume} {62}},\ \bibinfo {pages} {367} (\bibinfo {year} {2001})}\BibitemShut {NoStop}%
\bibitem [{\citenamefont {Harrow}\ \emph {et~al.}(2009)\citenamefont {Harrow}, \citenamefont {Hassidim},\ and\ \citenamefont {Lloyd}}]{harrow2009quantum}%
  \BibitemOpen
  \bibfield  {author} {\bibinfo {author} {\bibfnamefont {A.~W.}\ \bibnamefont {Harrow}}, \bibinfo {author} {\bibfnamefont {A.}~\bibnamefont {Hassidim}},\ and\ \bibinfo {author} {\bibfnamefont {S.}~\bibnamefont {Lloyd}},\ }\bibfield  {title} {\bibinfo {title} {Quantum algorithm for linear systems of equations},\ }\href@noop {} {\bibfield  {journal} {\bibinfo  {journal} {Physical review letters}\ }\textbf {\bibinfo {volume} {103}},\ \bibinfo {pages} {150502} (\bibinfo {year} {2009})}\BibitemShut {NoStop}%
\bibitem [{\citenamefont {Chia}\ \emph {et~al.}(2022)\citenamefont {Chia}, \citenamefont {Gily{\'e}n}, \citenamefont {Li}, \citenamefont {Lin}, \citenamefont {Tang},\ and\ \citenamefont {Wang}}]{chia2022sampling}%
  \BibitemOpen
  \bibfield  {author} {\bibinfo {author} {\bibfnamefont {N.-H.}\ \bibnamefont {Chia}}, \bibinfo {author} {\bibfnamefont {A.~P.}\ \bibnamefont {Gily{\'e}n}}, \bibinfo {author} {\bibfnamefont {T.}~\bibnamefont {Li}}, \bibinfo {author} {\bibfnamefont {H.-H.}\ \bibnamefont {Lin}}, \bibinfo {author} {\bibfnamefont {E.}~\bibnamefont {Tang}},\ and\ \bibinfo {author} {\bibfnamefont {C.}~\bibnamefont {Wang}},\ }\bibfield  {title} {\bibinfo {title} {Sampling-based sublinear low-rank matrix arithmetic framework for dequantizing quantum machine learning},\ }\href@noop {} {\bibfield  {journal} {\bibinfo  {journal} {Journal of the ACM}\ }\textbf {\bibinfo {volume} {69}},\ \bibinfo {pages} {1} (\bibinfo {year} {2022})}\BibitemShut {NoStop}%
\bibitem [{\citenamefont {Bryant}(1986)}]{bryant1986graph}%
  \BibitemOpen
  \bibfield  {author} {\bibinfo {author} {\bibfnamefont {R.~E.}\ \bibnamefont {Bryant}},\ }\bibfield  {title} {\bibinfo {title} {Graph-based algorithms for boolean function manipulation},\ }\href@noop {} {\bibfield  {journal} {\bibinfo  {journal} {Computers, IEEE Transactions on}\ }\textbf {\bibinfo {volume} {C-35}},\ \bibinfo {pages} {677} (\bibinfo {year} {1986})}\BibitemShut {NoStop}%
\bibitem [{\citenamefont {Javadi-Abhari}\ \emph {et~al.}(2024)\citenamefont {Javadi-Abhari}, \citenamefont {Treinish}, \citenamefont {Krsulich}, \citenamefont {Wood}, \citenamefont {Lishman}, \citenamefont {Gacon}, \citenamefont {Martiel}, \citenamefont {Nation}, \citenamefont {Bishop}, \citenamefont {Cross} \emph {et~al.}}]{qiskit2024}%
  \BibitemOpen
  \bibfield  {author} {\bibinfo {author} {\bibfnamefont {A.}~\bibnamefont {Javadi-Abhari}}, \bibinfo {author} {\bibfnamefont {M.}~\bibnamefont {Treinish}}, \bibinfo {author} {\bibfnamefont {K.}~\bibnamefont {Krsulich}}, \bibinfo {author} {\bibfnamefont {C.~J.}\ \bibnamefont {Wood}}, \bibinfo {author} {\bibfnamefont {J.}~\bibnamefont {Lishman}}, \bibinfo {author} {\bibfnamefont {J.}~\bibnamefont {Gacon}}, \bibinfo {author} {\bibfnamefont {S.}~\bibnamefont {Martiel}}, \bibinfo {author} {\bibfnamefont {P.~D.}\ \bibnamefont {Nation}}, \bibinfo {author} {\bibfnamefont {L.~S.}\ \bibnamefont {Bishop}}, \bibinfo {author} {\bibfnamefont {A.~W.}\ \bibnamefont {Cross}}, \emph {et~al.},\ }\bibfield  {title} {\bibinfo {title} {Quantum computing with qiskit},\ }\href@noop {} {\bibfield  {journal} {\bibinfo  {journal} {arXiv preprint arXiv:2405.08810}\ } (\bibinfo {year} {2024})}\BibitemShut {NoStop}%
\bibitem [{\citenamefont {P{\u{a}}tra{\c{s}}cu}\ and\ \citenamefont {Williams}(2010)}]{patrascu2010possibility}%
  \BibitemOpen
  \bibfield  {author} {\bibinfo {author} {\bibfnamefont {M.}~\bibnamefont {P{\u{a}}tra{\c{s}}cu}}\ and\ \bibinfo {author} {\bibfnamefont {R.}~\bibnamefont {Williams}},\ }\bibfield  {title} {\bibinfo {title} {On the possibility of faster sat algorithms},\ }in\ \href@noop {} {\emph {\bibinfo {booktitle} {Proceedings of the twenty-first annual ACM-SIAM symposium on Discrete Algorithms}}}\ (\bibinfo {organization} {SIAM},\ \bibinfo {year} {2010})\ pp.\ \bibinfo {pages} {1065--1075}\BibitemShut {NoStop}%
\bibitem [{\citenamefont {Jiang}\ \emph {et~al.}(2026)\citenamefont {Jiang}, \citenamefont {Zhang}, \citenamefont {Xiang}, \citenamefont {Yuan}, \citenamefont {Lu},\ and\ \citenamefont {Yin}}]{jiang_2026_23042152}%
  \BibitemOpen
  \bibfield  {author} {\bibinfo {author} {\bibfnamefont {Q.}~\bibnamefont {Jiang}}, \bibinfo {author} {\bibfnamefont {X.-M.}\ \bibnamefont {Zhang}}, \bibinfo {author} {\bibfnamefont {D.}~\bibnamefont {Xiang}}, \bibinfo {author} {\bibfnamefont {X.}~\bibnamefont {Yuan}}, \bibinfo {author} {\bibfnamefont {L.}~\bibnamefont {Lu}},\ and\ \bibinfo {author} {\bibfnamefont {J.}~\bibnamefont {Yin}},\ }\href {https://doi.org/10.5281/zenodo.23042152} {\bibinfo {title} {{Harnessing} problem structure for end-to-end quantum speed-ups}},\ \bibinfo {howpublished} {\emph{Zenodo} \url{https://doi.org/10.5281/zenodo.23042152}} (\bibinfo {year} {2026})\BibitemShut {NoStop}%
\end{thebibliography}%


\begin{thebibliography}{19}%
\makeatletter
\providecommand \@ifxundefined [1]{%
 \@ifx{#1\undefined}
}%
\providecommand \@ifnum [1]{%
 \ifnum #1\expandafter \@firstoftwo
 \else \expandafter \@secondoftwo
 \fi
}%
\providecommand \@ifx [1]{%
 \ifx #1\expandafter \@firstoftwo
 \else \expandafter \@secondoftwo
 \fi
}%
\providecommand \natexlab [1]{#1}%
\providecommand \enquote  [1]{``#1''}%
\providecommand \bibnamefont  [1]{#1}%
\providecommand \bibfnamefont [1]{#1}%
\providecommand \citenamefont [1]{#1}%
\providecommand \href@noop [0]{\@secondoftwo}%
\providecommand \href [0]{\begingroup \@sanitize@url \@href}%
\providecommand \@href[1]{\@@startlink{#1}\@@href}%
\providecommand \@@href[1]{\endgroup#1\@@endlink}%
\providecommand \@sanitize@url [0]{\catcode `\\12\catcode `\$12\catcode `\&12\catcode `\#12\catcode `\^12\catcode `\_12\catcode `\%12\relax}%
\providecommand \@@startlink[1]{}%
\providecommand \@@endlink[0]{}%
\providecommand \url  [0]{\begingroup\@sanitize@url \@url }%
\providecommand \@url [1]{\endgroup\@href {#1}{\urlprefix }}%
\providecommand \urlprefix  [0]{URL }%
\providecommand \Eprint [0]{\href }%
\providecommand \doibase [0]{https://doi.org/}%
\providecommand \selectlanguage [0]{\@gobble}%
\providecommand \bibinfo  [0]{\@secondoftwo}%
\providecommand \bibfield  [0]{\@secondoftwo}%
\providecommand \translation [1]{[#1]}%
\providecommand \BibitemOpen [0]{}%
\providecommand \bibitemStop [0]{}%
\providecommand \bibitemNoStop [0]{.\EOS\space}%
\providecommand \EOS [0]{\spacefactor3000\relax}%
\providecommand \BibitemShut  [1]{\csname bibitem#1\endcsname}%
\let\auto@bib@innerbib\@empty
\bibitem [{\citenamefont {Huang}\ and\ \citenamefont {Palsberg}(2024)}]{huang2024compiling}%
  \BibitemOpen
  \bibfield  {author} {\bibinfo {author} {\bibfnamefont {K.}~\bibnamefont {Huang}}\ and\ \bibinfo {author} {\bibfnamefont {J.}~\bibnamefont {Palsberg}},\ }\bibfield  {title} {\bibinfo {title} {Compiling conditional quantum gates without using helper qubits},\ }\href@noop {} {\bibfield  {journal} {\bibinfo  {journal} {Proceedings of the ACM on Programming Languages}\ }\textbf {\bibinfo {volume} {8}},\ \bibinfo {pages} {1463} (\bibinfo {year} {2024})}\BibitemShut {NoStop}%
\bibitem [{\citenamefont {Araujo}\ \emph {et~al.}(2023)\citenamefont {Araujo}, \citenamefont {Blank}, \citenamefont {Ara{\'u}jo},\ and\ \citenamefont {da~Silva}}]{araujo2024lowrank}%
  \BibitemOpen
  \bibfield  {author} {\bibinfo {author} {\bibfnamefont {I.~F.}\ \bibnamefont {Araujo}}, \bibinfo {author} {\bibfnamefont {C.}~\bibnamefont {Blank}}, \bibinfo {author} {\bibfnamefont {I.~C.}\ \bibnamefont {Ara{\'u}jo}},\ and\ \bibinfo {author} {\bibfnamefont {A.~J.}\ \bibnamefont {da~Silva}},\ }\bibfield  {title} {\bibinfo {title} {Low-rank quantum state preparation},\ }\href@noop {} {\bibfield  {journal} {\bibinfo  {journal} {IEEE Transactions on Computer-Aided Design of Integrated Circuits and Systems}\ }\textbf {\bibinfo {volume} {43}},\ \bibinfo {pages} {161} (\bibinfo {year} {2023})}\BibitemShut {NoStop}%
\bibitem [{\citenamefont {Malz}\ \emph {et~al.}(2024)\citenamefont {Malz}, \citenamefont {Styliaris}, \citenamefont {Wei},\ and\ \citenamefont {Cirac}}]{malz2024mps}%
  \BibitemOpen
  \bibfield  {author} {\bibinfo {author} {\bibfnamefont {D.}~\bibnamefont {Malz}}, \bibinfo {author} {\bibfnamefont {G.}~\bibnamefont {Styliaris}}, \bibinfo {author} {\bibfnamefont {Z.-Y.}\ \bibnamefont {Wei}},\ and\ \bibinfo {author} {\bibfnamefont {J.~I.}\ \bibnamefont {Cirac}},\ }\bibfield  {title} {\bibinfo {title} {Preparation of matrix product states with log-depth quantum circuits},\ }\href@noop {} {\bibfield  {journal} {\bibinfo  {journal} {Physical Review Letters}\ }\textbf {\bibinfo {volume} {132}},\ \bibinfo {pages} {040404} (\bibinfo {year} {2024})}\BibitemShut {NoStop}%
\bibitem [{\citenamefont {Jin}\ and\ \citenamefont {Jing}(2025)}]{jin2025ghz}%
  \BibitemOpen
  \bibfield  {author} {\bibinfo {author} {\bibfnamefont {Z.-y.}\ \bibnamefont {Jin}}\ and\ \bibinfo {author} {\bibfnamefont {J.}~\bibnamefont {Jing}},\ }\bibfield  {title} {\bibinfo {title} {Preparing greenberger-horne-zeilinger states on ground levels of neutral atoms},\ }\href@noop {} {\bibfield  {journal} {\bibinfo  {journal} {Physical Review A}\ }\textbf {\bibinfo {volume} {112}},\ \bibinfo {pages} {022602} (\bibinfo {year} {2025})}\BibitemShut {NoStop}%
\bibitem [{\citenamefont {Cao}\ \emph {et~al.}(2024)\citenamefont {Cao}, \citenamefont {Hansen}, \citenamefont {Giorgino}, \citenamefont {Carosini}, \citenamefont {Zah{\'a}lka}, \citenamefont {Zilk}, \citenamefont {Loredo},\ and\ \citenamefont {Walther}}]{cao2024ghz}%
  \BibitemOpen
  \bibfield  {author} {\bibinfo {author} {\bibfnamefont {H.}~\bibnamefont {Cao}}, \bibinfo {author} {\bibfnamefont {L.}~\bibnamefont {Hansen}}, \bibinfo {author} {\bibfnamefont {F.}~\bibnamefont {Giorgino}}, \bibinfo {author} {\bibfnamefont {L.}~\bibnamefont {Carosini}}, \bibinfo {author} {\bibfnamefont {P.}~\bibnamefont {Zah{\'a}lka}}, \bibinfo {author} {\bibfnamefont {F.}~\bibnamefont {Zilk}}, \bibinfo {author} {\bibfnamefont {J.}~\bibnamefont {Loredo}},\ and\ \bibinfo {author} {\bibfnamefont {P.}~\bibnamefont {Walther}},\ }\bibfield  {title} {\bibinfo {title} {Photonic source of heralded greenberger-horne-zeilinger states},\ }\href@noop {} {\bibfield  {journal} {\bibinfo  {journal} {Physical Review Letters}\ }\textbf {\bibinfo {volume} {132}},\ \bibinfo {pages} {130604} (\bibinfo {year} {2024})}\BibitemShut {NoStop}%
\bibitem [{\citenamefont {Gilliam}\ \emph {et~al.}(2019)\citenamefont {Gilliam}, \citenamefont {Venci}, \citenamefont {Muralidharan}, \citenamefont {Dorum}, \citenamefont {May}, \citenamefont {Narasimhan},\ and\ \citenamefont {Gonciulea}}]{gilliam2019foundational}%
  \BibitemOpen
  \bibfield  {author} {\bibinfo {author} {\bibfnamefont {A.}~\bibnamefont {Gilliam}}, \bibinfo {author} {\bibfnamefont {C.}~\bibnamefont {Venci}}, \bibinfo {author} {\bibfnamefont {S.}~\bibnamefont {Muralidharan}}, \bibinfo {author} {\bibfnamefont {V.}~\bibnamefont {Dorum}}, \bibinfo {author} {\bibfnamefont {E.}~\bibnamefont {May}}, \bibinfo {author} {\bibfnamefont {R.}~\bibnamefont {Narasimhan}},\ and\ \bibinfo {author} {\bibfnamefont {C.}~\bibnamefont {Gonciulea}},\ }\bibfield  {title} {\bibinfo {title} {Foundational patterns for efficient quantum computing},\ }\href@noop {} {\bibfield  {journal} {\bibinfo  {journal} {arXiv preprint arXiv:1907.11513}\ } (\bibinfo {year} {2019})}\BibitemShut {NoStop}%
\bibitem [{\citenamefont {Gilliam}\ \emph {et~al.}(2021)\citenamefont {Gilliam}, \citenamefont {Woerner},\ and\ \citenamefont {Gonciulea}}]{gilliam2021grover}%
  \BibitemOpen
  \bibfield  {author} {\bibinfo {author} {\bibfnamefont {A.}~\bibnamefont {Gilliam}}, \bibinfo {author} {\bibfnamefont {S.}~\bibnamefont {Woerner}},\ and\ \bibinfo {author} {\bibfnamefont {C.}~\bibnamefont {Gonciulea}},\ }\bibfield  {title} {\bibinfo {title} {Grover adaptive search for constrained polynomial binary optimization},\ }\href@noop {} {\bibfield  {journal} {\bibinfo  {journal} {Quantum}\ }\textbf {\bibinfo {volume} {5}},\ \bibinfo {pages} {428} (\bibinfo {year} {2021})}\BibitemShut {NoStop}%
\bibitem [{\citenamefont {Grover}(1996)}]{grover1996fast}%
  \BibitemOpen
  \bibfield  {author} {\bibinfo {author} {\bibfnamefont {L.~K.}\ \bibnamefont {Grover}},\ }\bibfield  {title} {\bibinfo {title} {A fast quantum mechanical algorithm for database search},\ }in\ \href@noop {} {\emph {\bibinfo {booktitle} {Proceedings of the twenty-eighth annual ACM symposium on Theory of computing}}}\ (\bibinfo {year} {1996})\ pp.\ \bibinfo {pages} {212--219}\BibitemShut {NoStop}%
\bibitem [{\citenamefont {Brassard}\ \emph {et~al.}(2000)\citenamefont {Brassard}, \citenamefont {Hoyer}, \citenamefont {Mosca},\ and\ \citenamefont {Tapp}}]{brassard2000quantum}%
  \BibitemOpen
  \bibfield  {author} {\bibinfo {author} {\bibfnamefont {G.}~\bibnamefont {Brassard}}, \bibinfo {author} {\bibfnamefont {P.}~\bibnamefont {Hoyer}}, \bibinfo {author} {\bibfnamefont {M.}~\bibnamefont {Mosca}},\ and\ \bibinfo {author} {\bibfnamefont {A.}~\bibnamefont {Tapp}},\ }\bibfield  {title} {\bibinfo {title} {Quantum amplitude amplification and estimation},\ }\href@noop {} {\bibfield  {journal} {\bibinfo  {journal} {arXiv preprint quant-ph/0005055}\ } (\bibinfo {year} {2000})}\BibitemShut {NoStop}%
\bibitem [{\citenamefont {Vedral}\ \emph {et~al.}(1996)\citenamefont {Vedral}, \citenamefont {Barenco},\ and\ \citenamefont {Ekert}}]{vedral1996quantum}%
  \BibitemOpen
  \bibfield  {author} {\bibinfo {author} {\bibfnamefont {V.}~\bibnamefont {Vedral}}, \bibinfo {author} {\bibfnamefont {A.}~\bibnamefont {Barenco}},\ and\ \bibinfo {author} {\bibfnamefont {A.}~\bibnamefont {Ekert}},\ }\bibfield  {title} {\bibinfo {title} {Quantum networks for elementary arithmetic operations},\ }\href@noop {} {\bibfield  {journal} {\bibinfo  {journal} {Physical Review A}\ }\textbf {\bibinfo {volume} {54}},\ \bibinfo {pages} {147} (\bibinfo {year} {1996})}\BibitemShut {NoStop}%
\bibitem [{\citenamefont {Barenco}\ \emph {et~al.}(1995)\citenamefont {Barenco}, \citenamefont {Bennett}, \citenamefont {Cleve}, \citenamefont {DiVincenzo}, \citenamefont {Margolus}, \citenamefont {Shor}, \citenamefont {Sleator}, \citenamefont {Smolin},\ and\ \citenamefont {Weinfurter}}]{barenco1995elementary}%
  \BibitemOpen
  \bibfield  {author} {\bibinfo {author} {\bibfnamefont {A.}~\bibnamefont {Barenco}}, \bibinfo {author} {\bibfnamefont {C.~H.}\ \bibnamefont {Bennett}}, \bibinfo {author} {\bibfnamefont {R.}~\bibnamefont {Cleve}}, \bibinfo {author} {\bibfnamefont {D.~P.}\ \bibnamefont {DiVincenzo}}, \bibinfo {author} {\bibfnamefont {N.}~\bibnamefont {Margolus}}, \bibinfo {author} {\bibfnamefont {P.}~\bibnamefont {Shor}}, \bibinfo {author} {\bibfnamefont {T.}~\bibnamefont {Sleator}}, \bibinfo {author} {\bibfnamefont {J.~A.}\ \bibnamefont {Smolin}},\ and\ \bibinfo {author} {\bibfnamefont {H.}~\bibnamefont {Weinfurter}},\ }\bibfield  {title} {\bibinfo {title} {Elementary gates for quantum computation},\ }\href@noop {} {\bibfield  {journal} {\bibinfo  {journal} {Physical review A}\ }\textbf {\bibinfo {volume} {52}},\ \bibinfo {pages} {3457} (\bibinfo {year} {1995})}\BibitemShut {NoStop}%
\bibitem [{\citenamefont {Durr}\ and\ \citenamefont {Hoyer}(1996)}]{durr1996quantum}%
  \BibitemOpen
  \bibfield  {author} {\bibinfo {author} {\bibfnamefont {C.}~\bibnamefont {Durr}}\ and\ \bibinfo {author} {\bibfnamefont {P.}~\bibnamefont {Hoyer}},\ }\bibfield  {title} {\bibinfo {title} {A quantum algorithm for finding the minimum},\ }\href@noop {} {\bibfield  {journal} {\bibinfo  {journal} {arXiv preprint quant-ph/9607014}\ } (\bibinfo {year} {1996})}\BibitemShut {NoStop}%
\bibitem [{\citenamefont {Hegerfeld}\ and\ \citenamefont {Kratsch}(2023)}]{hegerfeld2023tight}%
  \BibitemOpen
  \bibfield  {author} {\bibinfo {author} {\bibfnamefont {F.}~\bibnamefont {Hegerfeld}}\ and\ \bibinfo {author} {\bibfnamefont {S.}~\bibnamefont {Kratsch}},\ }\bibfield  {title} {\bibinfo {title} {Tight algorithms for connectivity problems parameterized by clique-width},\ }in\ \href@noop {} {\emph {\bibinfo {booktitle} {31st Annual European Symposium on Algorithms (ESA 2023)}}},\ \bibinfo {series} {Leibniz International Proceedings in Informatics (LIPIcs)}, Vol.\ \bibinfo {volume} {274}\ (\bibinfo  {publisher} {Schloss Dagstuhl -- Leibniz-Zentrum f{\"u}r Informatik},\ \bibinfo {year} {2023})\ pp.\ \bibinfo {pages} {59:1--59:19}\BibitemShut {NoStop}%
\bibitem [{\citenamefont {Ross}\ and\ \citenamefont {Selinger}(2016)}]{ross2016optimal}%
  \BibitemOpen
  \bibfield  {author} {\bibinfo {author} {\bibfnamefont {N.~J.}\ \bibnamefont {Ross}}\ and\ \bibinfo {author} {\bibfnamefont {P.}~\bibnamefont {Selinger}},\ }\bibfield  {title} {\bibinfo {title} {Optimal ancilla-free clifford+ t approximation of z-rotations.},\ }\href@noop {} {\bibfield  {journal} {\bibinfo  {journal} {Quantum Inf. Comput.}\ }\textbf {\bibinfo {volume} {16}},\ \bibinfo {pages} {901} (\bibinfo {year} {2016})}\BibitemShut {NoStop}%
\bibitem [{\citenamefont {Dantsin}\ and\ \citenamefont {Wolpert}(2005)}]{dantsin2005derandomization}%
  \BibitemOpen
  \bibfield  {author} {\bibinfo {author} {\bibfnamefont {E.}~\bibnamefont {Dantsin}}\ and\ \bibinfo {author} {\bibfnamefont {A.}~\bibnamefont {Wolpert}},\ }\bibfield  {title} {\bibinfo {title} {Derandomization of {Schuler}'s algorithm for {SAT}},\ }in\ \href@noop {} {\emph {\bibinfo {booktitle} {Theory and Applications of Satisfiability Testing}}},\ \bibinfo {series} {Lecture Notes in Computer Science}, Vol.\ \bibinfo {volume} {3542}\ (\bibinfo  {publisher} {Springer Berlin Heidelberg},\ \bibinfo {year} {2005})\ pp.\ \bibinfo {pages} {80--88}\BibitemShut {NoStop}%
\bibitem [{\citenamefont {Cygan}\ \emph {et~al.}(2016)\citenamefont {Cygan}, \citenamefont {Dell}, \citenamefont {Lokshtanov}, \citenamefont {Marx}, \citenamefont {Nederlof}, \citenamefont {Okamoto}, \citenamefont {Paturi}, \citenamefont {Saurabh},\ and\ \citenamefont {Wahlstr{\"o}m}}]{cygan2016cnfsat}%
  \BibitemOpen
  \bibfield  {author} {\bibinfo {author} {\bibfnamefont {M.}~\bibnamefont {Cygan}}, \bibinfo {author} {\bibfnamefont {H.}~\bibnamefont {Dell}}, \bibinfo {author} {\bibfnamefont {D.}~\bibnamefont {Lokshtanov}}, \bibinfo {author} {\bibfnamefont {D.}~\bibnamefont {Marx}}, \bibinfo {author} {\bibfnamefont {J.}~\bibnamefont {Nederlof}}, \bibinfo {author} {\bibfnamefont {Y.}~\bibnamefont {Okamoto}}, \bibinfo {author} {\bibfnamefont {R.}~\bibnamefont {Paturi}}, \bibinfo {author} {\bibfnamefont {S.}~\bibnamefont {Saurabh}},\ and\ \bibinfo {author} {\bibfnamefont {M.}~\bibnamefont {Wahlstr{\"o}m}},\ }\bibfield  {title} {\bibinfo {title} {On problems as hard as cnf-sat},\ }\href@noop {} {\bibfield  {journal} {\bibinfo  {journal} {ACM Transactions on Algorithms (TALG)}\ }\textbf {\bibinfo {volume} {12}},\ \bibinfo {pages} {1} (\bibinfo {year} {2016})}\BibitemShut {NoStop}%
\bibitem [{\citenamefont {P{\u{a}}tra{\c{s}}cu}\ and\ \citenamefont {Williams}(2010)}]{patrascu2010possibility}%
  \BibitemOpen
  \bibfield  {author} {\bibinfo {author} {\bibfnamefont {M.}~\bibnamefont {P{\u{a}}tra{\c{s}}cu}}\ and\ \bibinfo {author} {\bibfnamefont {R.}~\bibnamefont {Williams}},\ }\bibfield  {title} {\bibinfo {title} {On the possibility of faster sat algorithms},\ }in\ \href@noop {} {\emph {\bibinfo {booktitle} {Proceedings of the twenty-first annual ACM-SIAM symposium on Discrete Algorithms}}}\ (\bibinfo {organization} {SIAM},\ \bibinfo {year} {2010})\ pp.\ \bibinfo {pages} {1065--1075}\BibitemShut {NoStop}%
\bibitem [{\citenamefont {Boyer}\ \emph {et~al.}(1998)\citenamefont {Boyer}, \citenamefont {Brassard}, \citenamefont {H{\o}yer},\ and\ \citenamefont {Tapp}}]{boyer1998tight}%
  \BibitemOpen
  \bibfield  {author} {\bibinfo {author} {\bibfnamefont {M.}~\bibnamefont {Boyer}}, \bibinfo {author} {\bibfnamefont {G.}~\bibnamefont {Brassard}}, \bibinfo {author} {\bibfnamefont {P.}~\bibnamefont {H{\o}yer}},\ and\ \bibinfo {author} {\bibfnamefont {A.}~\bibnamefont {Tapp}},\ }\bibfield  {title} {\bibinfo {title} {Tight bounds on quantum searching},\ }\href@noop {} {\bibfield  {journal} {\bibinfo  {journal} {Fortschritte der Physik: Progress of Physics}\ }\textbf {\bibinfo {volume} {46}},\ \bibinfo {pages} {493} (\bibinfo {year} {1998})}\BibitemShut {NoStop}%
\bibitem [{\citenamefont {Jordan}\ \emph {et~al.}(2025)\citenamefont {Jordan}, \citenamefont {Shutty}, \citenamefont {Wootters}, \citenamefont {Zalcman}, \citenamefont {Schmidhuber}, \citenamefont {King}, \citenamefont {Isakov}, \citenamefont {Khattar},\ and\ \citenamefont {Babbush}}]{jordan2025decoded}%
  \BibitemOpen
  \bibfield  {author} {\bibinfo {author} {\bibfnamefont {S.~P.}\ \bibnamefont {Jordan}}, \bibinfo {author} {\bibfnamefont {N.}~\bibnamefont {Shutty}}, \bibinfo {author} {\bibfnamefont {M.}~\bibnamefont {Wootters}}, \bibinfo {author} {\bibfnamefont {A.}~\bibnamefont {Zalcman}}, \bibinfo {author} {\bibfnamefont {A.}~\bibnamefont {Schmidhuber}}, \bibinfo {author} {\bibfnamefont {R.}~\bibnamefont {King}}, \bibinfo {author} {\bibfnamefont {S.~V.}\ \bibnamefont {Isakov}}, \bibinfo {author} {\bibfnamefont {T.}~\bibnamefont {Khattar}},\ and\ \bibinfo {author} {\bibfnamefont {R.}~\bibnamefont {Babbush}},\ }\bibfield  {title} {\bibinfo {title} {Optimization by decoded quantum interferometry},\ }\href@noop {} {\bibfield  {journal} {\bibinfo  {journal} {Nature}\ }\textbf {\bibinfo {volume} {646}},\ \bibinfo {pages} {831} (\bibinfo {year} {2025})}\BibitemShut {NoStop}%
\end{thebibliography}%

\end{document}